\documentclass[11pt]{article}
\usepackage{amsmath}
\usepackage{graphicx}
\usepackage{enumerate}
\usepackage{natbib}
\usepackage{url}
\usepackage[T1]{fontenc}
\usepackage{enumitem}
\usepackage{comment}
\usepackage{amssymb,amsthm,amsmath}
\usepackage{hyperref}
\usepackage{cleveref}
\usepackage{marginnote}
\usepackage{scalerel}
\usepackage{bigints}
\usepackage{upgreek}
\usepackage{dsfont}
\usepackage{float}
\usepackage[dvipsnames]{xcolor}
\usepackage{thmtools}
\usepackage{thm-restate}
\usepackage{titling}
\usepackage{upgreek}
\usepackage{soul}
\usepackage{framed}
\usepackage{titlesec}
\usepackage[justification=centering]{caption}
\usepackage{listings,lstautogobble}

\newcommand{\R}{\mathbb{R}}
\newcommand{\N}{\mathbb{N}}

\newcommand{\EW}{\mathbb{E}}
\newcommand{\ind}[1]{\mathds{\large 1}_{#1}}

\newcommand{\measure}{\Delta_s}
\newcommand{\measurei}[1]{\Delta_s^{[#1]}}
\newcommand{\regf}{g_\theta}
\newcommand{\regF}[1]{g_{#1}}
\newcommand{\standB}[1]{\mathbf{e}^{(#1)}}
\newcommand{\standV}[1]{\mathbf{e}_{#1}}
\newcommand{\regfspace}{\mathcal{M}}
\newcommand{\eff}{s}
\newcommand{\Ireg}{I}
\newcommand{\Mreg}{M}
\newcommand{\Creg}{C}

\newcommand{\RI}{\ensuremath{\mathbb{X}}}

\newcommand{\supp}{\text{supp}}

\newcommand{\regprod}[2]{\boldsymbol{(}#1\boldsymbol{\cdot} #2\boldsymbol{)}}
\newcommand{\regprodCl}[2]{Cl\Big({\boldsymbol{(}#1\boldsymbol{\cdot} #2\boldsymbol{)}}\Big)}

\newcommand{\aiimeasure}{\bar{\mu}_{X_{[\text{MC}]}\vert X^I}}

\newcommand{\efun}{g_{avg}^{[I]}}

\newcommand{\realx}[1]{x_{[#1]}}

\newcommand{\UR}{$100(1-\alpha)$\%-uncertainty region}

\usepackage{fancyhdr}
\newcommand{\JASA}{0}
\newcommand{\Preprint}{1}

\makeatletter
\newcommand\footnoteref[1]{\protected@xdef\@thefnmark{\ref{#1}}\@footnotemark}
\makeatother

\newenvironment{changemargin}[2]{%
  \begin{list}{}{%
    \setlength{\topsep}{0pt}%
    \setlength{\leftmargin}{#1}%
    \setlength{\rightmargin}{#2}%
    \setlength{\listparindent}{\parindent}%
    \setlength{\itemindent}{\parindent}%
    \setlength{\parsep}{\parskip}%
  }%
  \item[]}{\end{list}}

\theoremstyle{plain}

\newtheorem{lemma}{Lemma}[section]
\newtheorem{ex}{Example}[section]
\newtheorem*{ex*}{Example}

\declaretheorem[name=Corollary]{cor}

\theoremstyle{definition}
\newtheorem{defi}{Definition}
\newtheorem*{defi*}{Definition}

\newtheorem{rems}{Remark}[section]
\newtheorem{Adefi}{Definition}[section]
\newtheorem{Aadefi}{Definition}[subsection]

\theoremstyle{remark}
\newtheorem*{rem}{Remark}

\newtheorem*{n}{Note}

\crefname{defi}{definition}{definitions}
\Crefname{defi}{Definition}{Definitions}
\crefname{subdefi}{definition}{definitions}
\Crefname{subdefi}{Definition}{Definitions}
\crefname{Adefi}{definition}{definitions}
\Crefname{Adefi}{Definition}{Definitions}
\crefname{Aadefi}{definition}{definitions}
\Crefname{Aadefi}{Definition}{Definitions}
\crefname{rems}{remark}{remarks}
\Crefname{rems}{Remark}{Remarks}
\crefname{prop}{proposition}{propositions}
\Crefname{prop}{Proposition}{Propositions}
\crefname{Aprop}{proposition}{propositions}
\Crefname{Aprop}{Proposition}{Propositions}
\crefname{ex}{example}{examples}
\Crefname{ex}{Example}{Examples}

\definecolor{lgray}{gray}{0.99}
\title{\huge\textsc{\textbf{A formal framework for generalized reporting methods in parametric settings}}}
\author{
 \large{Hannah Kümpel}\thanks{The Institute for Medical Information Processing, Biometry, and Epidemiology, Ludwig-Maximilians-Universität of Munich}\medskip\\
  \texttt{hannah.kuempel@ibe.med.uni-muenchen.de}\\\medskip
  \and
  \large{Sabine Hoffmann}\thanks{Department of Statistics, Ludwig-Maximilians-Universität of Munich}\medskip\\
  \texttt{sabine.hoffmann@stat.uni-muenchen.de}
}

\date{}

\begin{document}
\maketitle

\renewcommand{\abstractname}{\textsc{\large Abstract}}
\begin{abstract}
\noindent Effect size measures and visualization techniques aimed at maximizing the interpretability and comparability of results from statistical models have long been of great importance and are recently again receiving increased attention in the literature. However, since the methods proposed in this context originate from a wide variety of disciplines and are more often than not practically motivated, they lack a common theoretical framework and many quantities are narrowly or heuristically defined. In this work, we put forward a common mathematical setting for effect size measures and visualization techniques aimed at the results of parametric regression and define a formal framework for the consistent derivation of both existing and new variants of such quantities. Throughout the presented theory, we utilize probability measures to derive weighted means over areas of interest. While we take a Bayesian approach to quantifying uncertainty in order to derive consistent results for every defined quantity, all proposed methods apply to the results of both frequentist and Bayesian inference. We apply selected specifications derived from the proposed framework to data from a clinical trial and a multi-analyst study to illustrate its versatility and relevance.
\end{abstract}


\section{Introduction}
\label{sec:intro}
The interpretation of results of statistical models is notoriously difficult. While the question of interpretability has recently received considerable attention in the machine learning literature, 
it is important to recognize that interpretability is also an important issue when it comes to classical statistical quantities. Despite numerous cautions and clarifications \citep{wasserstein2016asa,Wasserstein2019}, some  misinterpretations of statistical significance are so widespread 
that they fill entire textbooks  \citep{Cassidy2019}, making it quite difficult for researchers to recognize that these interpretations are actually incorrect. Similarly, classical effect estimates such as
odds ratios and hazard ratios 
are staggeringly often misinterpreted (\cite{Sackett1996}, \cite{Greenland1987}).
And even if applied researchers are familiar with the correct definitions, they often prefer to think of these quantities as risk ratios or even as absolute risk differences, implicitly choosing a wrong answer to the right question over a right answer to the wrong question. Meanwhile, for model classes that do provide interpretable effect size estimates, researchers still often face a problem of comparability of these estimates when they try to compare 
the results of different model classes or even the results of the same model class on different data sets.

Recent interest in these issues has led to a plethora of visualization techniques and effect size measures that have been proposed in different disciplines to improve the 
interpretability and comparability of model results \citep{Gelman_AveragePredictiveComparison, Williams2012, Kaufman1996, ALEml, Goldstein2015,  BischlHeumann}. 
However, so far these approaches are lacking a common theoretical framework  and many present a number of unclarities and inconsistencies, especially regarding categorical independent variables. 
\noindent In this work, we define a formal framework that not only unifies this collection of methods in parametric settings but 
also allows for the consistent derivation of definitions for special cases, thereby avoiding discrepancies and facilitating the deduction of further methodological variants.
    After outlining the contributions and 
    setting of the proposed framework, the following will define all therein contained quantities and place a variety of existing methodologies from various disciplines within them in sections \ref{PropQuantSec} and \ref{SpecSec}, respectively. Subsequently, \cref{Datasection} illustrates selected aspects of the proposed methodology on data from a clinical trial \cite{RCT} and a multi-analyst study \cite{SILBERZAHN}. \if1\JASA{All proofs may be found in \cref{Proofs}, which is part of the supplementary material.}\fi \if1\Preprint{All proofs are given in \cref{Proofs} and the code used for the analyses of \cref{Datasection} may be found in \cref{CodeApp}.}\fi
\paragraph{Contributions} The main contribution of this work is the providing of a formal framework for reporting methods aimed at maximizing the interpretability and comparability of model results, in particular, one which allows researchers to customize these methods for the specific research question at hand. In doing so, we identify three axiomatic assumptions regarding the dependence structure of models' independent variables that are useful for clarifying the approach for different specifications, including existing methods.  Notably, we provide a consistent method of deriving point estimates and well-interpretable uncertainty regions for all proposed quantities. When defining a generalized version of marginal effects, we include the additional option to separately quantify main- and interaction- effects. Finally, in the interest of increasing the communicability of results from models with parametric distribution assumptions, we propose a method for visualizing the different expected distributions of the target variable resulting from changes in the values of independent variables.

\subsection{Setting}\label{setting}

Within this work, we restrict ourselves to cases of parametric regression. 
Specifically, we will consider sequences $\{(y_i,\mathbf{x}_i)\}_{i=1,...,n}$ of $n\in \N$ observations, where the vectors $\mathbf{x}_i\in \R^p$, $p\in\N$, contain the observed values of all regressors, dummy coded in the case of categorical inputs, corresponding to an observation $y_i\in \R$ of the target variable. To be precise, we consider each element of the vectors $\mathbf{x}_i$ to be a realization of a variable $X_j$, $1\leq j\leq p$; 
with $X_j$ representing either a metric regressor, i.e. taking values in $\R$, or
 one category of a categorical regressor, i.e. taking values in $\{0,1\}$, with the \emph{reference category} being equated with all other categories taking the value zero. Generally, we denote the \mbox{realization of any such variable $X_j$ by $x_j$.} Thereby, a categorical regressor with $d\in \N$ categories, the first of which being represented by $X_j$, is defined by the vector of binary variables  $
	(X_j,...,X_{j+d-2})\text{, with }(x_j,...,x_{j+d-2})^\top\in\{\mathbf{e}_i,1\leq i\leq d-1\}\cup \mathbf{0},
$
	where $\{\mathbf{e}_i,1\leq i\leq d-1\}$ denotes the standard basis of $\R^{d-1}$ and $\mathbf{0}$ the zero vector, specifically the $(d-1)$-dimensional zero vector in this case. Since we will use such sets repeatedly, we define for any $d\in\N$ the set $\standB{d}:=\{\mathbf{e}_i,1\leq i\leq d\}\cup \mathbf{0}\subseteq\R^{d}$.

\noindent For the purposes of this work, we furthermore separate the regressors into three groups:

\begin{enumerate}[label=(\arabic*)]
	\item The regressor whose effect on the target variable is of interest, referred to as \emph{the regressor of interest}, $X^\Ireg$, which can either be metric
	or a categorical variable with $d_I+1$ categories (or $d_I$ non-reference categories). 
	For consistency of notation, we set $d_I=1$, if $X^\Ireg$ is a metric regressor, and denote by $I$ the set that contains only $X^I$.
	\item The set of (remaining) metric regressors $\Mreg$. Any variable representing an element of $M$ will be denoted by $X^{M}_l$, $1\leq l\leq m$, with $m$ being the number of (remaining) metric regressors.
	\item The set of (remaining) categorical regressors $\Creg$. We denote by $c$ the number of (remaining) categorical regressors and by $d_{C_k}$, $1\leq k\leq c$, the number of non-reference categories of the $k$th (remaining) categorical regressor $X^C_k$.
\end{enumerate}
\begin{rems}\label{Notation}
	In essence, we denote the vector of models' independent variables by $X=(X_1,...,X_p)$, with each entry taking values in $\R$ or $\{0,1\}$. We further denote any value the variable $X_j$, $1\leq j\leq p$, takes by $x_j$ and, in a slight abuse of terminology, refer to both $X_j$ and $x_j$ as \emph{regressor (category)}. Moreover, instead of reordering the vector $X$, we assign each of its elements to one of the sets $I$, $M$, and $C$, effectively mapping subsets of the elements of $X$ to elements of the set $\Ireg\cup\Mreg\cup\Creg$, with $p=d_I+m+\sum_{k=1}^{c}d_{C_k}$. Lastly, in
	a slight abuse of notation, we write \mbox{$\realx{I}$, $\realx{M_l}$, and $\realx{C_j}$} for the values that the regressors mapped to $X^I$, $X^M_l$, $l\in\{1,...,m\}$, and $X^C_k$, $k\in\{1,...,c\}$, respectively, take. For example, if the regressor of interest is a categorical regressor whose non-reference categories are represented by the first three elements of $X$, we write $\realx{I}=(x_1,x_2,x_3)$. \enlargethispage{11pt}
\end{rems}

As previously mentioned, the current theory is restricted to the case of parametric regression, by which we mean a regression setting in which only the parameter $\theta\in \Theta$, with $\Theta\subseteq\R^k$, $k\in\N_{\geq{0}}$, being a finite-dimensional vector space, is unknown. Given a set $S$ endowed with $\sigma$-algebra $\mathcal{S}$, we denote by $\regfspace(S)$ the space of measurable functions from \mbox{$S$ to $\R$}, 
\mbox{$\regfspace(S)\hspace{-3pt}:=\big\{f:S\longrightarrow\R\vert\text{ $f$ is $\big(\mathcal{S},\mathcal{B}(\R)\big)$-measurable}\big\}$,} and consider a function $g_\theta$ to be indexed by $\theta\in\Theta$ in the sense of the following mapping
\[
\regF{}:\Theta\longrightarrow\regfspace(\R^p),\quad\theta\longmapsto \regf.
\]

\noindent Given this, we then assume that each $y_i$ is a realization of a real-valued random variable $Y$ and \begin{equation}\label{regMod}
	\EW[Y\vert X]=\regf(X)
\end{equation}
\noindent where the function $\regf:\R^p\longrightarrow\R$ is at least once partially differentiable w.r.t. each metric 
element of $\realx{\Ireg}$ $\forall \theta\in\Theta$.
\begin{n}
In the following, we will refer to both $\regf$ and $\regf(x_1,...,x_p)$, i.e. the expected value of $Y$ given specific regressor values,  as \textit{expectation function} or  \textit{expectation}, even though particularly $\regf(x_1,...,x_p)$ is commonly referred to as \emph{(model) prediction} in the literature.
\end{n}

\section{Definition of the proposed quantities\label{PropQuantSec}}
This section will provide definitions of, among others, expectation plots and generalized marginal effects. Throughout, the underlying idea is that one may utilize probability measures to average over different functions of the regressor vector $X$, with the choice of probability measures depending on the axiomatic assumption made regarding the dependence structure of the regressors. While those axiomatic assumptions will be introduced in \cref{MAsec}, some additional conditions need to be placed on the probability measures. Specifically, for some function $h:\R^{d}\longrightarrow\R$, $d\in\N_{>0}$, and set $D\subseteq\R^{\tilde{d}}$, $d\geq\tilde{d}\in\N_{>0}$, we consider probability measures $\mu$ that satisfy the following requirements:
\begin{itemize}[leftmargin=1.3cm]
	\item[(M1)] $\mu$ is a probability measure on $\big(\R^{\tilde{d}},\mathcal{B}(\,\R^{\tilde{d}}\,)\big)$.
	\item[(M2)] $\supp(\mu)\subseteq D$, if required as a result of $\mu$ being normalized\footnote{See \cref{normmuDef} in \cref{AppDefs} for a formal definition of \emph{normalized} in this context.} w.r.t. $D$.
	\item[(M3)] $\int_{\R^{\tilde{d}}} \vert h(x)\vert d\mu(x)<\infty,$ if $\tilde{d}=d$, or $\forall x_a\in\R^{d-\tilde{d}}:\int_{\R^{\tilde{d}}} \vert h(x_a,x_b)\vert d\mu(x_b)<\infty$, if $\tilde{d}<d$.\footnote{See \cref{SplitRem} in \cref{AppDefs} for details on the notation $h(x_a,x_b)$.}
\end{itemize}
\begin{n}
While we use the term `averaging' throughout, by choosing $\mu$ as the Dirac measure $\delta_x$ for $x\in\R^{\tilde{d}}$ one of course \emph{fixes} regressor values at the point $x$.
\end{n}
In this work, we propose a shared method of deriving point estimates and credible sets for all introduced quantities, which will be detailed in \cref{UncertaintySection}. To this end, every quantity will be defined as a function of the parameter vector $\theta$ in sections \ref{MAsec} and \ref{gMEsec}.

\subsection{Expectation plots and individualized expectation\label{MAsec}}
We begin with an accessible and already prevalent visualization technique, versions of which are referred to under different names, such as \emph{prediction plots}, \emph{marginal plots}, see e.g. \cite{Cook1997};\cite{ALEml}, and \emph{partial dependence plots}, see e.g. \cite{Friedmann};\cite{Zhao2021}.
Given the setting of \cref{setting}, these terms refer to functions which, for any given value $\theta\in\Theta$, visualize the expectation function on the Cartesian plane, either as a line in the case of $X^I$ being metric or as $d_I+1$ points for a categorical $X^I$, both representing the expected value of $Y$ at a value or category, respectively, of the regressor of interest, and averaged, if applicable, over the regressors \mbox{mapped to $M$ and $C$.}

\noindent To incorporate any such function into our formal framework, we define the function
\begin{equation}\label{avgex}
\efun:\Theta\times\R^{d_I}\longrightarrow\R,\quad (\theta,\realx{I})\longmapsto\int_{\bar{\RI}}\regf(x)d\bar{\mu}(\realx{\text{MC}})\,,
\end{equation}
\noindent where
\ $\realx{\text{MC}}\in\R^{p-d_I}$ denotes the vector of realizations of those regressors that are not the regressor of interest, and $\bar{\mu}$ is a measure satisfying requirements (M1)-(M3) w.r.t. $\regf$ and some set $\bar{\RI}\subseteq\R^{p-d_I}$ $\forall\theta\in\Theta$, and chosen according to one of the following axiomatic assumptions, unless the regressor of interest is the only regressor in which case $\efun\hat{=}\regf$.
\paragraph{\textit{Axiomatic assumption} (A.\textrm{I})} \emph{All regressors are independent variables}.\\
	In this case, $\bar{\mu}$ is a product measure, specifically\begin{align*}
\int_{\bar{\RI}}\regf(x)&d\bar{\mu}(\realx{\text{MC}})=\\&\int_{\RI_{\scaleto{M_1}{5pt}}}\hspace*{-8pt}...\int_{\RI_{\scaleto{M_m}{5pt}}}\int_{\RI_{\scaleto{C_1}{5pt}}}\hspace*{-8pt}...\int_{\RI_{\scaleto{C_c}{5pt}}}\hspace*{-5pt}\regf(x)d\mu_{\scaleto{C_c}{5pt}}(\realx{\scaleto{C_c}{5pt}})...d\mu_{\scaleto{C_1}{5pt}}(\realx{\scaleto{C_1}{5pt}})d\mu_{\scaleto{M_m}{5pt}}(\realx{\scaleto{M_m}{5pt}})...d\mu_{\scaleto{M_1}{5pt}}(\realx{\scaleto{M_1}{5pt}})\,,
	\end{align*}
\noindent where, to satisfy requirements (M1)-(M3), for each regressor in $M$, i.e. \mbox{$\forall i\in\{1,...,m\}$,} an appropriate set $\RI_{M_i}\subseteq\R$ and
a probability measure $\mu_{M_i}$ on $\big(\R,\mathcal{B}(\R)\big)$ with $\supp(\mu_{M_i})\subseteq\RI_{M_i}$ (if required as a result of being normalized)
and for each regressor in $C$, i.e. $\forall j\in\{1,...,c\}$,
an appropriate set $\RI_{C_j}\subseteq\standB{d_{C_j}}$ and
a probability measure $\mu_{C_j}$ on $\big(\standB{d_{C_j}},\mathcal{P}\big(\standB{d_{C_j}}\big)\big)$ with $\supp(\mu_{C_j})\subseteq\RI_{C_j}$ are chosen in a way so that $\forall (\theta,\realx{I})\in\Theta\times\R^{d_I}$ the following holds:\footnote{In addition to reasonably limiting the choices of sets and measures, this requirement ensures that, if one or more measures are defined via a density, $\bar{\eff}$ may be computed as an iterated integral by Fubini's theorem, see \cite{Fubini}.}
\begin{center}
$\bigintss_{\R^{p-d_I}}\big\vert \regf(x)\big\vert\, d\big(\mu_{M_1}\times...\times\mu_{M_m}\times\mu_{C_1}\times...\times\mu_{C_c}\big)<\infty\,.
$ 
\end{center}
\paragraph{\textit{Axiomatic assumption} (A.\textrm{II})} \emph{The regressors are not necessarily independent but follow a joint distribution}.\\
In this case, let $\mu_X$ denote the measure associated\footnote{\label{assnote}See \cref{assMrem} in \cref{AppDefs} for a formal explanation of this terminology.} with the given or assumed joint distribution of the regressors $X$.
Specifically, one has two different options for the choice of $\bar{\mu}$, given by the following. \begin{itemize}
	\item[\textbf{(A.II$'$)}] Choose $\bar{\mu}$ as the measure associated\footnoteref{assnote} with the marginal distribution obtained by marginalizing out the regressor of interest. 
	The regressor of interest is then treated as an independent variable.
	\item[\textbf{(A.II$''$)}] Choose $\bar{\mu}$ as the measure associated\footnoteref{assnote} with the conditional distribution of all regressors but the regressor of interest given $X^I$, denoted by $\aiimeasure$.

 Under assumptions (A.II$'$) and (A.II$''$), requirement (M1) is implicitly met. However, one may still want to normalize $\bar{\mu}$ w.r.t. to some set $\bar{\RI}$, either because only the effect on this set is of interest or because requirement (M3) does not hold for $\bar{\RI}=\supp(\bar{\mu})$. In such a case, the sets $\RI_{M_i},\,i\in\{1,...,m\}$ and $\RI_{C_j},\,i\in\{1,...,c\}$ are chosen as in (A.I), but $\bar{\mu}$ is not a product measure and therefore needs to be normalized directly w.r.t. $\bar{\RI}=\RI_{M_1}\times...\times\RI_{M_m}\times\RI_{C_1}\times...\times\RI_{C_c}$. 
\end{itemize}

\noindent Having defined a quantity that gives the average expected value of the target variable as a function of the value of the regressor of interest, the natural next step is additionally averaging over that function w.r.t. some probability measure for the regressor of interest in order to quantify the \emph{average expected value of $Y$ given certain ranges of regressor values}. To that end, the following defines what we refer to as `individualized expectation'.

\begin{defi}\label{IndivExp}
		Given the setting of \cref{setting}, 
		we define the \emph{individualized expectation} as the following function
		\begin{equation}\begin{aligned}
				e:\Theta\longrightarrow\R,\quad\theta\longmapsto\begin{cases}	\bigintsss_{\RI_I}\regf(x)\,d\mu_I(x),&\text{if $p=d_I$,}\\[10pt]
					\bigintsss_{\RI_I}\efun(\theta,x)\,d\mu_I(x),& \text{otherwise,}
				\end{cases}\\[10pt]
			\end{aligned}
		\end{equation}
	
	\noindent where $\mu_I$ is a measure that satisfies requirements (M1)-(M3) w.r.t. $\regf$ and  \mbox{$\RI_I\subseteq\R^{d_I}$} $\forall\theta\in\Theta$, and $\efun$ is given depending on the chosen axiomatic assumption. \begin{n}
	Under assumption (A.II$''$), choosing $\mu_I$ as the measure associated with the marginal distribution of $X^I$ under the joint distribution of all regressors amounts to calculating the integral over the expectation function w.r.t. the (normalized) joint distribution. This follows since by \cite[theorem 4.3.6.]{Park2017} these choices result in the individualized expectation $e(\theta)$ being given by $\int_{\RI_I\times\bar{\RI}}\regf(x)d\mu_X(x)$, with $\mu_X$ denoting the measure associated with the given or assumed joint distribution of all regressors, normalized, if applicable, w.r.t. $\RI_I\times\bar{\RI}$.
	When specifying different choices of sets and measures to integrate over and w.r.t., we write the corresponding individualized expectation under all assumptions as  $e\big(\theta\big\vert \RI_I,\bar{\RI},\mu_I,\bar{\mu}\big)$, or $e\big(\theta\big\vert \RI_I,\mu_I\big)$ if $p=d_I$.
\end{n}
\end{defi}

\subsection{Slope of expectation and generalized marginal effects\label{gMEsec}}
In the previous section, two functions were defined: First, a function of the parameter vector \emph{and} the regressor of interest quantifying the average expected value of the target variable and, second, a function of \emph{only} the parameter vector quantifying the value of the first function averaged over a certain set. Next, we aim to analogously define two functions quantifying not the average expected value of $Y,$ but the \emph{average effect of the regressor of interest on $Y.$} More precisely, we are interested in quantifying the average `effect' as represented by the average slope, specifically the average value of the partial derivative for a metric regressor of interest and the average difference between each non-reference category and the reference category for a categorical regressor of interest. 
We refer to these two new functions as \emph{slope of expectation} and \emph{generalized marginal effect} and the remainder of this section gives their formal definition.

\begin{defi}\label{AverageSlopeDef}
Given the setting from \cref{setting} and
\begin{equation}\label{slope1}
s(\theta,x):=\begin{cases}
	\dfrac{\partial \regf}{\partial \realx{I}} (x),&\text{if $X^I$ is a metric regressor,}\\[10pt]
\regf (x)- \regf (\left . x \right\vert_{\realx{I}=\mathbf{0}}),&\text{otherwise, i.e. if $X^I$ is a categorical regressor,}
\end{cases}
\end{equation}
we define the \emph{slope of expectation} as the following function
	\begin{equation}\label{SuperdefSD}\begin{aligned}
			\bar{\eff}:\Theta\times\R^{d_I}\longrightarrow\R,\quad (\theta,\realx{I})\longmapsto\begin{cases}s(\theta,x),& \text{if $p=d_I$,}\\
				\bigintss_{\bar{\RI}}\eff(\theta,x)d\bar{\mu}(\realx{\text{MC}}),& \text{otherwise,}
			\end{cases}\\[10pt]
		\end{aligned}
\end{equation}
where, correspondingly to \cref{avgex}, $\bar{\mu}$ is a measure satisfying requirements (M1)-(M3) w.r.t. $s(\theta,\boldsymbol{\cdot})$ and $\bar{\RI}\subseteq\R^{p-d_I}$ $\forall\theta\in\Theta$, and chosen according to one of the axiomatic assumptions (A.I), (A.II$'$), or (A.II$''$).\if1\JASA{\enlargethispage{10pt}}\fi

\begin{n}Under assumption (A.II$''$), \cref{SuperdefSD} needs to be modified as follows for cases where the regressor is not the only regressor:
\begin{align*}
    \hspace*{-.25cm}(\theta,\realx{I})&\longmapsto\begin{cases}
				\int_{\bar{\RI}}\eff(\theta,x)\aiimeasure(d\realx{\text{MC}},\realx{I}), &\text{if $X^I$ is a metric regressor,\footnotemark}\\[10pt]
				\begin{array}{l} \int_{\bar{\RI}}\regf(x)\aiimeasure(d\realx{\text{MC}},\realx{I})\,-\\[2pt]\int_{\bar{\RI}}\regf (\left . x \right\vert_{\realx{I}=\mathbf{0}})\aiimeasure(d\realx{\text{MC}},\mathbf{0}),
				\end{array}&\text{otherwise.}
			\end{cases}
\end{align*}
\noindent However, in the interest of readability, we use only the notation of \cref{SuperdefSD} in any definition or statement which applies under all assumptions.\footnotetext{See \cref{Aiinote} in \cref{AppDefs} for a formal explanation of the notation used for this integral term.}
\end{n}

\end{defi}
\begin{rems}\label{Slopefplot}
Equivalently to $\efun$, for any given value $\theta\in\Theta$, the slope of expectation $\bar{\eff}(\theta,\boldsymbol{\cdot}):\R^{d_I}\longrightarrow\R,\, p\mapsto \bar{\eff}(\vartheta,p)$ may be nicely plotted on the Cartesian plane, either as a line in the case of the regressor of interest being metric or as $d_I$ points representing the average difference in expectation to the reference category for each other category.

\end{rems}

Before giving our proposed definition of generalized marginal effects, the following briefly addresses our approach to cases where the regressor whose effect is of interest is involved in an interaction term.
\paragraph{Interaction terms} In existing methodology, marginal effects 
have only been defined to quantify the effect of a regressor of interest \textit{including} all interactions it is involved in. This makes perfect sense from an interpretability perspective, as even in linear regression the $\beta$-coefficients for interactions between metric regressors have no easily accessible interpretation. However, when comparing models, the option to isolatedly quantify the main and the interaction effect of a regressor of interest might well become relevant. Therefore, solely in the interest of comparability, we include in our definition of generalized marginal effects the option of deriving separate main and interaction effects under the simplifying\footnote{See \cref{Interactions} for a detailed motivation and discussion of this assumption.} assumption that interaction terms behave independently of the regressors involved in them.
Specifically, for interactions between metric regressors, we propose to treat the interaction term as a new regressor, with the probability measures chosen appropriately under each assumption, see \cref{I1} for further details. Meanwhile, \if1\JASA{\enlargethispage{10pt}}\fi
for interactions between categorical regressors, we propose to combine the individual regressors and interactions into one new regressor-vector, whose dimension we again denote by $d_I$, and denote each (product of) variable(s) as $X^I_l$, $1\leq l\leq d_I$, so the entire new regressor is written as $X^I_{cat}=(X^I_1, ... , X^I_{d_I})^\top$. Additionally, we require the following two vectors $\forall l\in \{1,...,d_I\}$ to calculate an appropriate distance as slope: $\mathbf{v}_l\in\{0,1\}^{d_{I}}$, which has $1$ as the $l$th entry as well as all entries representing categories that are included in $X^I_l$; and $\mathbf{ref}_l\in\{0,1\}^{d_{I}}$, which has $0$ as every entry but those entries representing categories that are included in $X^I_l$. If $X^I_l$ is not an interaction term, $\mathbf{ref}_l=\mathbf{0}\in\R^{d_I}$. \Cref{CatVecEx} in \cref{Interactions} provides an intuitive illustration of these vectors. Please note that we would not recommend  computing interaction effects separately in settings where the goal is to maximize interpretability.
	
The following now gives our definition of generalized marginal effects.

\begin{defi}\label{SuperDef}
Given the setting from \cref{setting} and a \emph{slope of expectation} $\bar{\eff}$ as defined in \cref{AverageSlopeDef}, the generalized marginal effect $\measure$ is given by the following functions:
	\begin{enumerate}[label=(\roman*)]
		\item For a \emph{metric regressor of interest}, given a set $\RI\subseteq\R$ and measure $\mu$ that satisfies requirements (M1)-(M3) w.r.t. $\bar{\eff}(\theta,\boldsymbol{\cdot})$ and $\RI$ $\forall\theta\in\Theta$, and, if one is interested in separately calculating the effect for interaction terms between metric regressors, chosen according to \cref{I1}, $\measure$ is defined as the function
		\begin{equation}
			\begin{aligned}\label{metricem}
				\measure:\Theta\longrightarrow\R,\quad\theta\longmapsto 	\int_\RI \bar{\eff}(\theta,\realx{I})d\mu(\realx{I})\, .
		\end{aligned}\end{equation}\if1\JASA{\vspace{-1cm}}\fi
\begin{n} Under axiomatic assumption (A.II$''$), the most natural choice of $\mu$ is the measure associated with the  marginal distribution of $X^I$, since this definition then amounts to integration over the derivative w.r.t. the (normalized) joint distribution as was detailed in \cref{IndivExp}. \end{n}
	
	\item  For a \emph{categorical regressor of interest} with $d_I+1$ categories, or $X^I_{cat}$ as defined above for cases where the effect of a categorical regressor (included in) an interaction with (an-) other categorical regressor(s) is of interest, $\measure$ is defined as the function
				\begin{equation}
			\begin{aligned}
				\measure:\Theta\longrightarrow\R^{d_I},\quad\theta\longmapsto \begin{pmatrix}
					\measurei{1}(\theta)\\
					\vdots\\
					\measurei{d_I}(\theta)
				\end{pmatrix}:= \begin{pmatrix}
					\bar{\eff}(\theta,\mathbf{v}_{1})&-&\bar{\eff}(\theta,\mathbf{ref}_1)\\
					&\vdots&\\
					\bar{\eff}(\theta,\mathbf{v}_{d_I})&-&\bar{\eff}(\theta,\mathbf{ref}_{d_I})
				\end{pmatrix}\,.
		\end{aligned}\end{equation}
		Here, we utilize the notation of $\measurei{j}$, $1\leq j\leq d_I$, in order to later be able to define point estimates and uncertainty regions in an elementwise fashion.\begin{n}
		When the categorical regressor of interest is not involved in any interactions or one does not want to separately compute main and interaction effects, $\bar{\eff}(\theta,\mathbf{ref}_{j})=0$ $\forall j\in\{1,...,d_I\}$, and, therefore, $\measure(\theta)=\big(\bar{\eff}(\theta,\standV{1}),...,\bar{\eff}(\theta,\standV{d_I})\big)^\top$.
		\end{n}
	\item The definition of $\measure$ for cases where one is interested in separately quantifying interaction effects between metric and categorical regressors may be found in \cref{DivInt}, specifically \cref{appdefi}, where we have decided to outsource it given the rather complex groundwork required to reach a full, formal definition.

	\end{enumerate}

\end{defi}
\noindent In linear regression, the generalized marginal effects under assumptions (A.I) and (A.II') will always be equal to the $\beta$-coefficients as the following theorem proves.

\begin{restatable}{thm}{BetasProp}
\label{BetasProp}Consider the case of $\regf$ being a linear polynomial in each variable representing a regressor (category), and $\theta$ given by $(\beta,v)^\top\in\Theta$, where $\beta$ denotes the vector containing the linear polynomial's coefficients and $v$ denotes elements of the parameter vector that do not appear in the function term of $\regf$, such as the error variance. Given that the regressor of interest is either not involved in an interaction or $\measure$ is used to separately quantify the main and interaction effects, the following holds
under assumptions (A.I) and (A.II$\,'$) as well as assumption (A.II$\,''$), if the regressor of interest is a metric variable or interaction. Any $\measure(\theta)$ and $\measurei{j}(\theta)$, \mbox{$j\in\{1,...,d_I\}$}, defined in \cref{SuperDef} is equal to the $\beta$-coefficient of the corresponding regressor of interest, or element thereof, respectively. 
\end{restatable}\if1\JASA{\enlargethispage{12pt}}\fi

\begin{rems}\label{TransRem}
 Notation-wise, we have so far not specifically considered the case of the specification of sets and measures being more intuitive for transformed versions of some or all regressors. Naturally, in such cases, one may treat the transform of each concerned element of $X$ as a new regressor and subsequently apply all methods as presented.
\end{rems}

\subsection{Point estimates and uncertainty regions\label{UncertaintySection} }

\noindent We adopt a Bayesian view on uncertainty as it appears to be the most natural one in this context. Specifically, we treat the parameter $\theta\in\Theta$ as a random variable and derive a point estimate and credible set directly for the random variable defined by $q\circ\theta$, where $q$ denotes any of the quantities proposed thus far. 

For frequentist inference results, where no posterior distribution for $\theta$ is available, we propose to follow the suggestion from \cite{Gelman_AveragePredictiveComparison} and,  given a point estimate $\hat{\theta}$ and covariance matrix $\Sigma_{\hat{\theta}}$ for $\theta$, assume that $\theta\sim N(\hat{\theta},\Sigma_{\hat{\theta}})$. Throughout, we denote by $\mu_\theta$ the probability measure that is associated either with a posterior distribution, resulting from Bayesian inference, or the $N(\hat{\theta},\Sigma_{\hat{\theta}})$ distribution, with frequentist point estimate $\hat{\theta}$ and covariance matrix $\Sigma_{\hat{\theta}}$.

\begin{rems}\label{UncertaintyWorks} 
The functions $e:\Theta\longrightarrow\R$, $\measure:\Theta\longrightarrow\R$, and $\measurei{j}:\Theta\longrightarrow\R$, $j\in\{1,...,d_I\}$, are clearly all Borel-measurable. Therefore, given a random variable $\theta$, $e\circ\theta$, $\measure\circ\theta$, and $\measurei{j}\circ\theta$ are again random variables and we denote the measures associated with their respective distributions by $\mu_{e(\theta)}$, $\mu_{\measure(\theta)}$, and $\mu_{\measurei{j}(\theta)}$, respectively. Meanwhile, for any fixed realization of $X^I$, the functions $\efun(\boldsymbol{\cdot},x):\Theta\longrightarrow\R$ and $\bar{\eff}(\boldsymbol{\cdot},x):\Theta\longrightarrow\R$ are also Borel-measurable $\forall x\in\R^p$, equivalently allowing for the computation of a point estimate and credible set for any fixed $x\in\R$, if the regressor of interest is metric, or $x\in\standB{d_I}$, if the regressor of interest is categorical. \if1\JASA{\enlargethispage{11pt}}\fi
\end{rems}

\noindent Given an assumed (posterior) distribution for $\theta$, the choice of point estimate and credible set is largely a matter of personal preference. We would suggest respectively utilizing the \emph{mean} and \emph{highest density region}, or HDR, as put forward by \cite{HDRs}, for which a precise definition in the current setting is given by \cref{HDR} in \cref{AppDefs}. However, other choices such as median and equal-tailed interval would also be appropriate and may be used to replace the mean and HDR, respectively. To emphasize this point, we write \emph{$100(1-\alpha)$\%-uncertainty region} instead of $100(1-\alpha)$\%-HDR in the following definitions.\if1\JASA{\enlargethispage{11pt}}\fi

\begin{defi}\label{UncertaintyDef}
Given the fully specified functions $e\big(\theta\big\vert \RI_I,\bar{\RI},\mu_I,\bar{\mu}\big)$ and $\measure(\theta)$, we define
	\begin{enumerate}[label=(\roman*)]
		\item the point estimates $\widehat{e}$ and $\widehat{\measure}$ as 
		$\widehat{e}:=\EW[e(\theta)]=\int_\Theta e(x)d\mu_\theta(x)$
		and\if1\Preprint{\begin{align*}
		\widehat{\measure}:=\EW[\measure(\theta)]=\begin{cases}
			\int_\Theta \measure(x)d\mu_\theta(x) &,\text{ if $\measure$ has image $\R$}\\[10pt]
			\begin{pmatrix}
					\int_\Theta \measurei{1}(x)d\mu_\theta(x)\\
					\vdots\\ 
						\int_\Theta \measurei{d}(x)d\mu_\theta(x)
			\end{pmatrix}&,\text{ if $\measure$ has image $\R^d$, $d\in\N$,}
		\end{cases}
		\end{align*}}\fi
		\if1\JASA{\begin{align*}
		\widehat{\measure}:=\EW[\measure(\theta)]=\begin{cases}
			\int_\Theta \measure(x)d\mu_\theta(x) &,\text{ if $\measure$ has image $\R$}\\[10pt]
			\begin{pmatrix}
					\int_\Theta \measurei{1}(x)d\mu_\theta(x)\\[-10pt]
					\vdots\\[-10pt] 
						\int_\Theta \measurei{d}(x)d\mu_\theta(x)
			\end{pmatrix}&,\text{ if $\measure$ has image $\R^d$, $d\in\N$,}
		\end{cases}
		\end{align*}}\fi respectively.
		\item for some $\alpha\in(0,1)$, the credible set  $C_{e}(\alpha)$ as the \UR\ 
		of $\mu_{e(\theta)}$; and the credible set $C_{\measure}(\alpha)$ as the \UR\
		of $\mu_{\measure(\theta)}$, if $\measure$ has image $\R$, or the collection of \UR s of each $\mu_{\measurei{j}(\theta)}$, $j\in\{1,...,d\}$,  if $\measure$ has image $\R^d$, $d\in\N$. In the last case, we denote by $C^{[j]}_{\measure}(\alpha)$, the uncertainty region around the $j$th entry of the point estimate $\widehat{\measure}$.
	\end{enumerate}
\end{defi}
\noindent Analogously, for any fixed $x\in\R$, if the regressor of interest is metric, or $x\in\standB{d_I}$, if the regressor of interest is categorical, we can define (i) the point estimates for $\efun$ and slope of expectation $\bar{\eff}$ as  $\EW\big[\efun(\theta,x)\big]$ and $\EW\big[\bar{\eff}(\theta,x)\big]$, respectively; as well as (ii) for some $\alpha\in(0,1)$, the credible sets  $C_{\efun\vert x}(\alpha)$ and $C_{\bar{\eff}\vert x}(\alpha)$ as the \UR s of the measures associated with the distributions of $\efun(\theta,x)$ and $\bar{\eff}(\theta,x)$, respectively.

\begin{rems}\label{ShadeRem} Just like $\efun$ and $\bar{\eff}$ may be plotted for any given value $\theta\in\Theta$, the corresponding point estimate and \UR s may be plotted as functions of the regressor of interest.
\end{rems}

\subsection{Visualizing sampling uncertainty\label{Predictions}}

\noindent In this section, we will briefly define a method of visualizing the distribution of the target variable $Y$ given an individualized expectation, given that a parametric distributional assumption has been made. We include this methodology since, firstly, it is a very natural extension of the quantities defined thus far and, secondly, we believe a simple visual comparison of how the assumed distribution of the target variable changes for different individualized expectations could improve researchers' understanding of model results, for which we provide one possible example in \cref{Silberzahn}.
While we further believe that it could be worthwhile to extend the current theory to settings without parametric distributional assumptions, the current work is limited to the following definition.

\begin{defi}\label{ippd}
	Given the setting of \cref{setting} and a parametric distributional assumption for $Y$, given by a conditional density or probability function $p(y\vert \theta,x)$ with mean $\upmu$, we define for an individualized expectation $e\big(\theta\big\vert \RI_I,\bar{\RI},\mu_I,\bar{\mu}\big)$
	\begin{enumerate}[label=(\roman*)]
		\item the density or probability function of $Y$ under $e$ as $
		p_e(y\vert\theta):=p\big(y;\upmu\text{$=$}e(\theta),\upsilon\big)\,,$
		where $p\big(y;\upmu,\upsilon\big)$  
		denotes the, if necessary re-parametrized, version of $p(y\vert \theta,x)$ with its mean $\upmu$ one entry of $\theta$ and, if applicable, the remaining parameter \mbox{entries represented by $\upsilon$.} 

		\item the density or probability function of the \emph{individualized predictive distribution} of $Y$ and $\theta$ under $e$ as the function\begin{equation}
		\Pi(y,\theta):= p_e(y\vert \theta)\pi(\theta)
		\end{equation}
	where $\pi$ denotes the density or probability function via which $\mu_\theta$, the probability measure that is associated either with a posterior distribution, resulting from Bayesian inference, or the $N(\hat{\theta},\Sigma_{\hat{\theta}})$ distribution, with frequentist point estimate $\hat{\theta}$ and covariance matrix $\Sigma_{\hat{\theta}}$, is defined.
	
	\end{enumerate}
\end{defi}
\noindent Equivalently to $\efun$ and $\bar{\eff}$, $\Pi(y,\boldsymbol{\cdot}):\Theta\longrightarrow\R$ is a measurable function for every fixed $y\in\R$ and, therefore, the point estimate $\EW_\theta[\Pi(y,\theta)]:=\int_\Theta \Pi(y,\theta)d\theta$ as well as the corresponding \UR s  may be plotted as a function of $y$, see \cref{fig:IndivPost} for an example. Note that
the quantity $\EW_\theta[\Pi(y,\theta)]$ may simply be seen as a \emph{marginal posterior predictive distribution}, specifically, marginalized w.r.t. the measure associated with the chosen distribution of the regressors $X$.

\section{Special specifications and transforms\label{SpecSec}}
Having established a formal framework for a wide range of methods which may be used to visualize and quantify the results of parametric models, this section will now, first, detail how many established reporting and visualization methods from different disciplines fit into this framework, and, subsequently, make some suggestions of specifications which we believe to 
have beneficial properties regarding both interpretability and comparability. Importantly, however, we do not claim that any one specification is universally superior to others. On the contrary, making it possible to specify any proposed quantity according to the specific research question at hand as well as prior knowledge about the distribution of the regressors is one very important contribution of the developed framework.
Throughout, we denote by $U(a,b)$ the continuous uniform distribution on $[a,b]\in\mathcal{B}(\R)$, by $\delta$ the Dirac measure, and use the term `measure associated with the (joint) empirical distribution' in the sense of the following definition.
\begin{defi}\label{EDdef}
Within the context of this work, given a sequence of regressor-observations  $\{\mathbf{x}_i\}_{i=1,...,n}=:\mathcal{D}_X$, $n\in\N$, as introduced in \cref{setting} and $n_{\mathcal{D}_X}(x):=\sum_{s\in\mathcal{D}_X}\ind{x=s}$ for any $x\in\R^p$, we refer to the probability measure defined via the probability function 
$$p(x)=\frac{n_{\mathcal{D}_X}(x)}{n}\cdot\ind{x\in\mathcal{D}_X}$$
as the \emph{measure associated with the joint empirical distribution} of all regressors. Furthermore, let $\mathcal{D}_X^{[j]}$, $j\in\{1,...,1+m+c\}$, denote the sequence of observations of some regressor in $I\cup M\cup C$. Correspondingly, we refer to the probability measure defined via the probability function 
$p(x)=n^{-1}\cdot n_{\mathcal{D}^{[j]}_X}(x)\cdot\ind{x\in\mathcal{D}_X^{[j]}}$
as the \emph{measure associated with the empirical distribution} of that regressor. Note that the support of these measures is given by the sets $\{\mathbf{x}_i\vert\mathbf{x}_i\in\mathcal{D}_X\}$ and $\{\realx{j}\vert\realx{j}\in\mathcal{D}^{[j]}_X\}$, respectively.
\end{defi}
\subsection{Connection to existing methods}

\subsubsection{\label{existingMEsec}Marginal Effects and Adjusted Predictions}
The concepts of \emph{marginal effects} and \emph{adjusted predictions} originated in economics, but have since become popular methods in other disciplines.
This section will now lay out how exactly the individualized expectation $e$ and generalized marginal effects $\measure$ proposed in the current work may be specified to equal the three commonly established versions of adjusted predictions and marginal effects, the definitions of which may be found in \cite{Williams2012}. 
\paragraph{Average (AAPs/AMEs)}
\emph{Average Adjusted predictions} amount to calculating $e(\hat{\theta})$ under axiomatic assumption (A.II$''$), with the joint empirical distribution of all regressors being used as $\mu_X$. Equivalently, the commonly used \emph{Average Marginal Effects} for a metric regressor of interest amount to calculating $\measure(\hat{\theta})$ under axiomatic assumption (A.II$''$), with $\mu_X$ again the joint empirical distribution of all regressors. However, for a categorical regressor of interest, the commonly used \emph{Average Marginal Effects} amount to calculating the $\measure(\hat{\theta})$ under axiomatic assumption (A.II$'$) instead of (A.II$''$), but with $\mu_X$ still chosen as joint empirical distribution of all regressors.\begin{n}Interestingly, this difference between the computation of AAPs/AMEs for a metric regressor of interest and a categorical regressor of interest has so far not been directly addressed, despite the fact that the marginal effects of a categorical regressor of interest may of course be quite different when computed under assumption (A.II$''$) instead of (A.II$'$), see \cref{fig:ED} for one example of this.
\end{n}

\paragraph{At Representative values (APRs/MERs)}
\emph{Adjusted predictions} and \emph{Marginal Effects at Representative values} amount to calculating $e(\hat{\theta})$ and $\measure(\hat{\theta})$, respectively, under axiomatic assumption (A.I), with the measure for each of the $p$ regressors being chosen as $\delta_{ \{x_i^\text{rep}\}}$, \mbox{$i\in\{1,...,p\}$,} where $x_i^\text{rep}$ denotes the $i$th entry of the chosen "representative value" \mbox{$x^\text{rep}\in\R^p$.}

\paragraph{At the Means (APMs/MEMs)}
\emph{Adjusted predictions} and \emph{Marginal Effects at the Means} simply amount to calculating the APRs and MERs, respectively, with the representative value for each regressor being chosen as its sample mean. Thereby, APMs and MEMs are equivalently calculated as $e(\hat{\theta})$ and $\measure(\hat{\theta})$, respectively, under axiomatic assumption (A.I), with the measure for each regressor being chosen as $\delta_{\{\bar{x}_i\}}$, $i\in\{1,...,p\}$, where $\bar{x}_i$ denotes the $i$th entry of the observed sample mean if $C$ is an empty set. \begin{n} \cite{Williams2012} does not specify what constitutes the `mean' for regressors representing categories of a categorical regressor.\end{n}

\begin{rems}
   In Stata \cite{StataSoftware} and other software modeled after Stata's marginal effects methods, standard errors and, thereby, confidence intervals for adjusted predictions and marginal effects are derived via the delta method. Importantly, this may lead to issues, e.g. when the adjusted prediction point estimate in logistic regression lies close to $0$ or $1$, as the confidence interval may then include impossible values that lie outside of $[0,1]$. In contrast, the  \texttt{marginaleffects} R-package \cite{marginaleffects} calculates point estimates and credible sets very similarly to the methods proposed in the current work for the three commonly established versions of adjusted predictions and marginal effects.
\end{rems}

\subsubsection{Medicine and Sociology\label{MedicalAppl}}
While the three commonly established versions of marginal effects detailed in the previous section are meanwhile also being used in other disciplines, such as medicine and sociology, other effect size measures which likewise fall into our framework have been proposed in
these disciplines. One example of this are \emph{average predictive comparisons}, which give the average normalized difference in expectation resulting from a certain change of the regressor of interest, and whose place in the current framework is given by the following proposition.\begin{restatable}{prop}{GelmanProp}\label{GelmanProp} The \emph{average predictive comparisons} for scalar inputs as defined in \cite{GelmanHill2007} correspond to the generalized marginal effect $\measure$ for a metric regressor of interest under assumption (A.II$\,'$), with $\mu_X$ chosen as the measure associated with the joint empirical distribution of all regressors and $\mu$ as the measure associated with the $U\big(u^{(lo)},u^{(hi)}\big)$ distribution for some chosen 
$u^{(lo)},u^{(hi)}\in\R$ with $u^{(lo)}<u^{(hi)}$.
\end{restatable}
\begin{rem}
 \cite{Gelman_AveragePredictiveComparison} extend the predictive comparison methods of \cite{GelmanHill2007} and \cref{GelmanPerdoeRem} in \cref{Proofs} relates their methodology to our framework.
\end{rem}
\noindent In addition, the concept of \emph{predicted change in probability}, which refers to the average difference between probabilities for different subgroups as predicted by a binary logistic model, may always be specified in our framework. While several variants of this concept have been proposed, the following proposition gives one formal representative instance of how such a variant, specifically as defined by \cite{Kaufman1996}, may be derived from quantities defined in this work.
\begin{restatable}{prop}{KaufmanProp}\label{KaufmanProp} Consider the case where $Y$ is a binary variable, the regressor of interest is metric, and, for some linear predictor $\eta$, $\regf(X)=\big(1+\exp(-\eta(X))\big)^{-1}$. Then the quantities $\partial P$ and $\Delta P$ as defined by \cite{Kaufman1996} correspond to the quantities $\eff$ and $\measure$ as defined in the current framework for a subset of possible specifications.

\begin{n}
\cite{Kaufman1996} additionally defines standardized versions of these quantities which may equivalently be applied to $\measure$. However, the current work also contains suggestions for standardizing all proposed quantities in \cref{StandAlg}.
\end{n}
\end{restatable}
\noindent Predicted change in probability is particularly often used in medical publications, see, e.g., \cite{Waehrer2020}, \cite{Fuller2015}, and \cite{Buchman2020}, particularly as a more interpretable alternative to odds ratios. \Cref{Datasection} will provide some further illustrations of how the odds ratio is not only difficult to interpret but can even become somewhat misleading.

\subsubsection{Interpretable machine learning\label{MLsec}}
As mentioned in the introduction, the quantities defined in this work are not only a generalization of marginal effects but also closely related to many model-agnostic methods that were recently proposed in the context of interpretable machine learning. 
Specifically, within the parametric setting of \cref{setting}, \emph{partial dependence-} and \emph{marginal- plots}, see, e.g., \cite{ALEml}, are equivalent to $\efun(\hat{\theta},\boldsymbol{\cdot})$ under assumptions (A.II$'$) and (A.II$''$), respectively. Furthermore, \emph{individual conditional expectation plots}, see \cite{Goldstein2015}, are equivalent to the $n$ plots resulting from specifying $\efun(\hat{\theta},\boldsymbol{\cdot})$ under (A.II$'$) and normalizing the measure associated with the joint empirical distribution w.r.t. each observation. Lastly, the following proposition relates the concept of \emph{accumulated local effects plot} as proposed by \cite{ALEml} to the generalized marginal effects defined in the current work.

\begin{restatable}{prop}{MLProp} \label{MLProp}
Consider the case where, for a given $j\in\{1,...,p\}$, (1) $\regf$ is not only partially but continuously partially differentiable w.r.t. $x_j$ and $\realx{I}=x_j$, i.e.  the regressor of interest is metric and the $j$th observed regressor, and (2) the marginal distribution of $X_j$ induces the \mbox{measure $\mu_{j}$}.  Within the parametric setting of \cref{setting}, the following then holds for the \emph{uncentred ALE main effect for differentiable $f(\cdot)$}, $g_{j,\text{ALE}}$, as defined by \mbox{\cite[thm. 1]{ALEml}}  under assumption (A.II$\,''$) with $\mu$ chosen not as $\mu_j$ but as the measure associated with the $U\big(\min\{\supp(\mu_{j})\},z\big)$  distribution $\forall z\in\supp(\mu_j)$\begin{equation}\label{MLProp_ALE}
        g_{j,ALE}(z)=\measure(\hat{\theta})\cdot\big(z-\min\{\supp(\mu_{j})\}\big)\,.
    \end{equation}
 \end{restatable}  
\begin{rem}
    Additionally,  \Cref{ScholbeckProp} in \cref{Proofs} provides more information including a formal proof regarding the connection of the current theory with the \emph{forward marginal effects} as defined by \cite{BischlHeumann}.
\end{rem}

\subsection{Choices of probability measures\label{PMchoiceSec}}
As already highlighted, we view the flexibility regarding how exactly the proposed quantities are averaged over regressor values as a great asset of the established framework, and, accordingly, there can be no universally appropriate choices of probability measures. Still, this section will discuss some properties of the two probability measures which we believe to be the most convenient regarding interpretability and computation in many settings.
\paragraph{Empirical distribution\label{EDsection}} Most existing methodology that falls into our framework utilizes the  measure associated with the joint empirical distribution, including, incidentally, proposed estimators for the quantities detailed in \cref{MLsec}. This option is not only advantageous because of its computational simplicity, but also the best choice for $\mu_X$ under assumptions (A.II$'$) and (A.II$''$) if one does not have any prior knowledge regarding the joint distribution of all regressors. Additionally, one may naturally consider the (joint) empirical distribution of all or some regressors on a fictive sequence of observations to specify each quantity as needed.

However, we would like to alert the reader to two possible drawbacks of utilizing the (joint) empirical distribution.
Firstly, averaging the proposed quantities over a set of observations may often be a reasonable estimate of the `average quantity', but the result may also be distorted in cases of low observation count or notable changes in $\regf$ between observations. Especially for $\measure$ in the case of a metric regressor of interest, other distribution choices will yield much more interpretable results in this regard, e.g. the Uniform distribution, which we discuss below. Secondly, the only possibility of extrapolation beyond the observed regressor values, if one is interested in that, is utilizing a fictive sequence of observations that represents the characteristics of the population to which one wishes to extrapolate. For metric regressors, however, choosing a sequence to average over rather than, e.g., an interval may often be restrictive in the context of extrapolation.

\paragraph{Uniform distribution}\label{UnifSec} Choosing probability measures as measures associated with the continuous Uniform distribution will yield an equally weighted average over a clearly defined set, which is arguably the most accessible version of an average to interpret. Since equally averaging over a set of discrete points may already be achieved by utilizing the (joint) empirical distribution on a fictive sequence of observations, as above, this section will now illustrate the advantage regarding interpretability for the case of choosing the measure associated with the continuous Uniform distribution as $\mu$ when deriving $\measure$ for a metric regressor of interest and, additionally, derive a convenient result regarding interpretability and computability in this setting.

When the regressor of interest is categorical, $\measurei{j}$, $j\in\{1,...,d_I\}$, simply gives the difference in conditional expectation of the outcome between each non-reference category and the reference category, averaged, depending on the choice of assumption and measure(s), over all other regressors and is therefore 
easy to interpret and communicate. When the regressor of interest is metric, the situation becomes slightly more complicated, however, if one chooses $\mu$ as the measure associated with the $U(a,b)$, $[a,b]\in\mathcal{B}(\R)$, one may legitimately report the statement `\textit{on the interval $[a,b]$, the expected value of $Y$ increases on average by $\widehat{\measure}$ when $X^I$ increases by $1$}'. Moreover, if $X^I$ is considered independent of the remaining regressors, $\widehat{\measure}$ gives nothing more than the \emph{standardized average difference in expectation between $\realx{I}=b$ and $\realx{I}=a$}, i.e. the average difference in expectation divided by the interval length.
The following proposition formalizes this statement.

\begin{restatable}{prop}{UnifProp}
\label{UnifProp}
For some $a_\RI,b_\RI\in\R$ with $a_\RI<b_\RI$, consider the case of the regressor of interest $X^I$ being metric with $\RI=(a_\RI,b_\RI)$ and $\mu$ being chosen or given as the measure associated with the $U(a_\RI,b_\RI)$ distribution.\\
Then, the following holds under assumptions (A.$I$) and (A.$II'$)
\begin{equation}\label{Unifprop1}
    \measure(\theta)=\big(b_\RI-a_\RI\big)^{-1}\cdot\Bigg(\int_{\bar{\RI}}\regf (\left . x \right\vert_{\realx{I}=b_\RI})
    -
    \regf (\left . x \right\vert_{\realx{I}=a_\RI})d\bar{\mu}(\realx{\text{\textup{MC}}})\Bigg)\, .
\end{equation}

\end{restatable}

\begin{restatable}{cor}{UnifCor}
\label{UnifCor}
Given the setting of \cref{UnifProp}, the following holds for $\vert \RI \vert=1$ under assumptions (A.$I$) and (A.$II'$)\begin{align}
          \measure(\theta)&=\EW_{X_{\textup{[MC]}}}\Big[\EW\big[Y\vert X^I=b_\RI, X_{\textup{[MC]}}\big]\Big]-\EW_{X_{\textup{[MC]}}}\Big[\EW\big[Y\vert X^I=a_\RI, X_{\textup{[MC]}}\big]\Big]\,.\label{UnifCor1}
\end{align}

\begin{n}
    \textit{Separately from the setting of \cref{UnifProp}, the RHS of above equation may be written in terms of individualized expectation as}
    \begin{equation}\label{iEdiff}
        e\big(\theta\big\vert\{b_\RI\},\bar{\RI},\delta_{\{b_\RI\}},\bar{\mu}\big)-e\big(\theta\big\vert \{a_\RI\},\bar{\RI},\delta_{\{a_\RI\}},\bar{\mu}\big)\,,
    \end{equation}
   \textit{where $\delta$ denotes the Dirac measure.}
\end{n}

\end{restatable}

\noindent Importantly, similar results may not be obtained under assumption (A.II$''$), as the regressor of interest is not treated as independent. A more detailed reasoning may be found in \cref{unfortunately not} in \cref{Proofs}.

\subsection{Choices of sets for regressor values}
Even more than in the case of probability measures, the proper choice of sets $\RI$, $\RI_I$, and $\bar{\RI}$ entirely depends on the research question that one intends to answer. In this context, we consider the option of equally averaging the effect over a chosen interval or set of points under assumption (A.I) one valuable contribution of our work, as it constitutes a natural and, in the latter case, more interpretable
extension of the \emph{at representative value} and \emph{at means} options from \cref{existingMEsec}. Additionally, the following now provides guidance on how to best quantify the generalized marginal effect in cases where the regressor of interest is metric and the slope of expectation $\bar{\eff}$ changes in a relevant manner over the course of the chosen $\RI$. 
\paragraph{Relevant slope changes for metric regressor of interest}In such cases, we suggest dividing $\RI$ into subsets, with the observed effect being similar within each subset, and calculating $\measure(\theta)$ separately for each subset of $\RI$, using the same measure $\mu$ \emph{but normalized w.r.t. each subset}. Subsequently, one may report a set of effect sizes in combination with the corresponding range of values the regressor of interest takes. Note that, in this context, the term `relevant' again very much depends on the setting and could refer to both, cases where the slope actually changes its sign in the interval of interest and cases with comparatively small slope changes that are, however, detrimental to decision making.

\subsection{Maximizing Comparability}\label{MaxComp}
As mentioned in \cref{gMEsec}, we expect that in certain settings, e.g. multi-analyst studies, the comparability of model results will be of greater importance than their interpretability. By facilitating the separate quantification of main and interaction effects, we have already introduced one method for such situations; the following will make a further suggestion for how to maximize the comparability of models addressing the same research question but either fit on different data sets or with different covariates, i.e. regressors in $M\cup C$.

When comparing two models that were fit on different data, comparability is easily maximized by either utilizing the joint empirical distribution of the combined data sets, if available, or by specifying the exact same probability measures for all common regressors, regardless of the observed values. In the case of different covariates, one could either use all available information by still averaging over all respective regressors or choose to set the probability measure for every covariate that does not appear in both models to $\delta_{\{0\}}$.

Additionally, in the interest of maximizing comparability, we have included suggestions for reporting standardized versions of all proposed quantities in \cref{StandAlg}.

\section{Applications \label{Datasection}}
This section illustrates the relevance of the proposed framework by first, highlighting the important role of identifying and communicating the axiomatic assumptions underlying one's reporting using data from the clinical trial of \cite{RCT}; and, second, demonstrating how the defined quantities may be used to solve a common comparability issue in multi-analyst studies by exemplarily applying them to \mbox{the setting of \cite{SILBERZAHN}.}

\subsection{Application to a clinical trial
\label{RCTsec}}
In their randomized controlled trial, \cite{RCT} compared the hospital length of stay (LOS) of an ethnically diverse group of patients with severe depressive disorders that were either given standard therapy (Group S) or genetically-guided therapy (Group G). 
For our first illustrative example, we followed the sub-analysis of \cite{Chrutchley2022} in omitting all observations with LOS$\,\leq3$ days; and fit a logistic regression model for the question \emph{how does ethnicity affect the probability of staying hospitalized longer than the overall median LOS ($6.625$ days)}, controlling for the variables age, gender, and group \mbox{assignment (G or S).}

Here, race was taken as the categorical regressor of interest, with \emph{White} (58.99\% of the observations) chosen as the reference category and the categories \emph{Latinx}, \emph{Black}, and \emph{Other} making up 24.73\%, 11.73\%, and 4.56\% of the observations, respectively.
Given the corresponding inference results, \cref{fig:ED} gives the expectation plots and generalized marginal effects under axiomatic assumptions (A.II$'$) and (A.II$''$) with $\mu_X$ taken as the joint empirical distribution on the observations used to fit the model.

\begin{figure}[h]
    \centering
    \includegraphics[width=\linewidth]{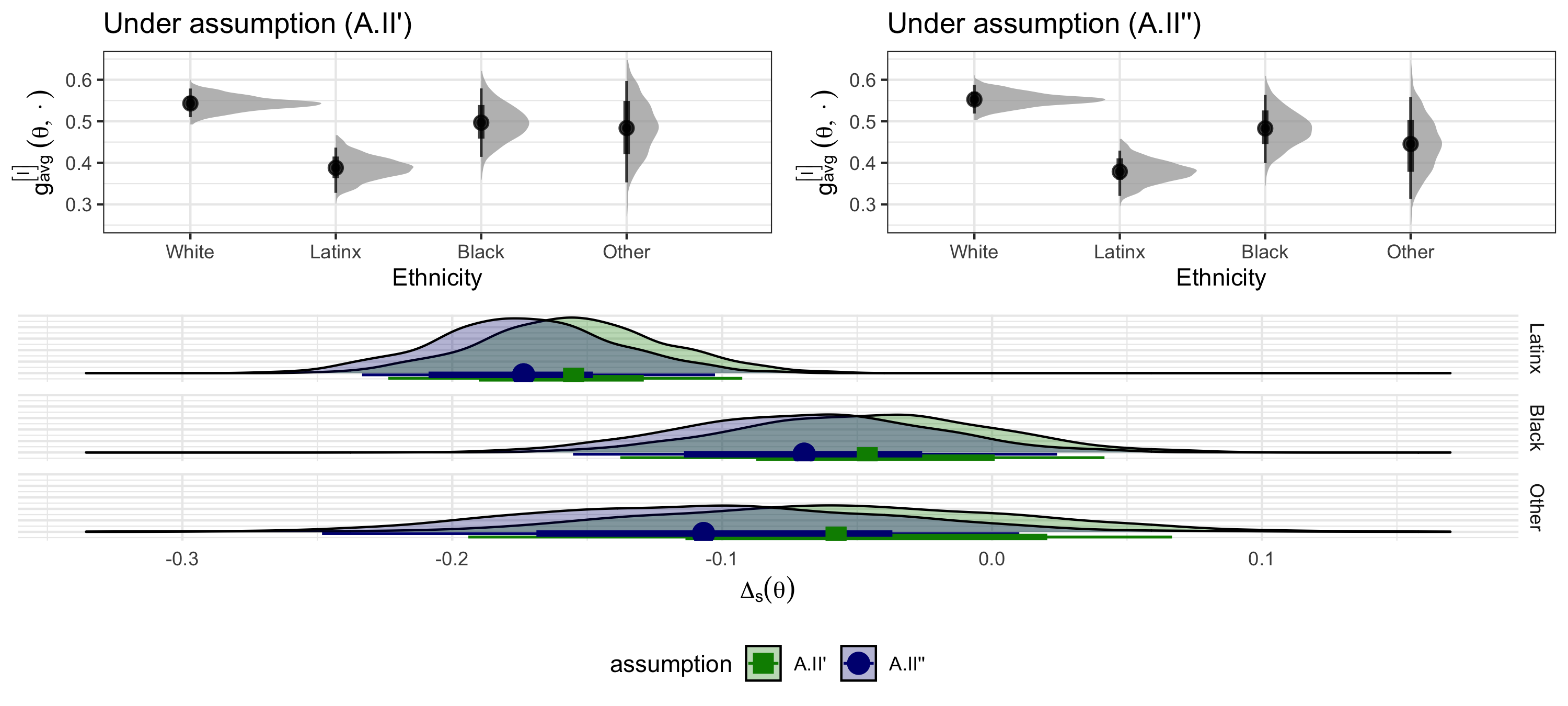}
    \if1\JASA{\vspace{-1cm}}\fi
    \caption{Expectation plots under assumptions (A.II$'$) (topleft) and (A.II$''$) (topright) and comparison of the corresponding generalized marginal effects (bottom) with point estimates and 95\%-uncertainty regions given as points and bars, respectively. \emph{White} was taken as the reference category for the categorical regressor of interest \emph{race}.}
    \label{fig:ED}
\end{figure}\if1\JASA{\vspace{-.8\baselineskip}}\fi

\noindent This example highlights two advantages of the proposed framework.
Firstly, 
the uncertainty regions and point estimates for all quantities retain the information about uncertainty, particularly uncertainty resulting from a low observation count for a certain regressor value, given by the inference regarding the regression parameters. 
In the current example, the uncertainty regions for the categories \emph{Black} and \emph{Other} are wider than for \emph{Latinx}, since the latter makes up a higher percentage of all observations. Secondly, as is especially apparent in the bottom plot of \cref{fig:ED}, the numeric results for the proposed quantities may be quite different depending on which axiomatic assumption underlies their calculation, and, most importantly, have different interpretations. Under assumption (A.II$''$), one derives the \emph{collective} average expectation, effect, etc. of all observations with the regressor of interest taking a certain value, while under assumption (A.II$'$), by treating the regressor of interest as independent, its impact is largely detached from the impacts of the remaining regressors. In the current example, the entry of $\widehat{\measure}$ corresponding to the non-reference category `Other' gives \emph{the difference between the empirical average probability of hospital stay lasting longer than $6.625$ days if each observation's race was changed to `Other' and the empirical average probability if each observation's race was changed to `White'} under assumption (A.II$'$). Meanwhile, under assumption (A.II$''$), the entry of $\widehat{\measure}$ corresponding to the non-reference category `Other' gives \emph{the difference between the empirical average probability of hospital stay lasting longer than $6.625$ days for those observations where `Other' was observed as race and the empirical average \mbox{probability for those observations where `White' was observed as race}.}

\subsection{Application to a multi-analyst study\label{Silberzahn}}
In a recent comment in Nature, \citet{Wagenmakers2022} argue that the scientific community needs to start quantifying analytical variability through multi-analyst projects to address the problem of overconfidence in statistical results.
However, many of the multi-analyst projects that have been completed so far have struggled with the question of how to achieve this goal if the analysis teams are free to choose an arbitrary model class, potentially ruling out the comparison of effect sizes on a common scale, see, e.g., \cite{SCHWEINSBERG} and \cite{Hoogeveen2022}. In one of the first multi-analyst studies, wherein twenty-nine teams of analysts were asked to answer the question "\emph{Are soccer referees more likely to give red cards to dark-skin-toned players than to light-skin-toned players?}"; \cite{SILBERZAHN} appeared to solve this issue, at least in the setting of their study, by asking all teams to report their findings in terms of odds ratios. This allowed them to report the studies results as a forest plot of odds ratios in \cite[fig. 3]{SILBERZAHN}, with the teams ordered by distributional assumption and all but two being classified as either "Linear", "Poisson", or "Logistic".\enlargethispage{12pt}

However, closer examination of the analysis strategies reported by each of the teams reveals that, while the incidence rate ratios resulting from Poisson models were justifiably taken as equivalent to odds ratios under the \emph{rare disease assumption}, the teams that were classified as "Linear" in the aforementioned forest plot indeed first fit a linear model to answer the question, but then additionally fit a logistic model in order to be able to report an odds ratio value. To illustrate how the quantities proposed in the current work solve the issue of effect size-comparability between different model classes, we fit a normal, Poisson, and logistic regression model in which we regressed the number of red cards given per game, the number of red cards given overall with the number of games used as offset, and the number of red cards (successes) per games (trials), respectively, on the standardized mean skin color rating. \Cref{fig:ExpectationSlope} gives the resulting expectation and slope plots, with the offset fixed at ln$(1)$ for the Poisson model. Thereby, the $y$-axes may, in a simplified way, be  interpreted as the `expected number of red cards given per game' and the slope thereof, respectively. \if1\Preprint{

}\fi
\begin{figure}[h]
    \centering
    \includegraphics[width=.9\linewidth]{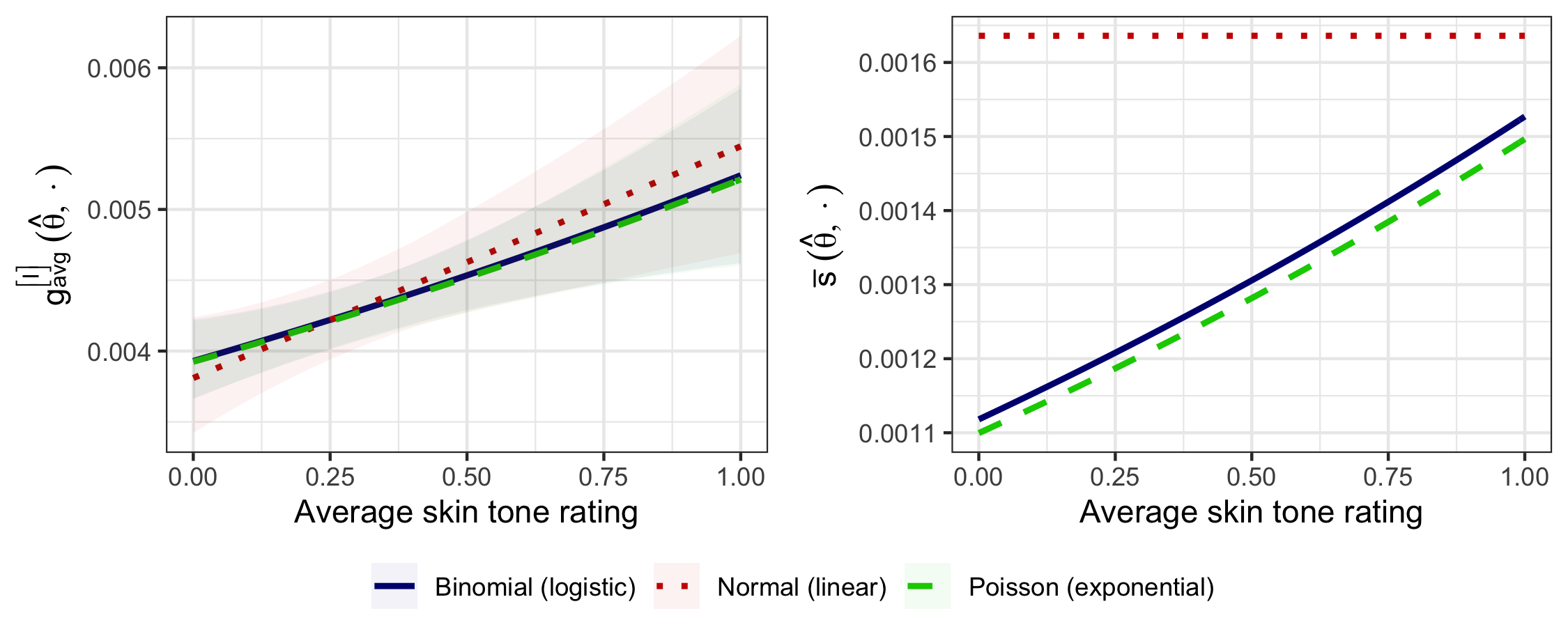}
    \caption{The expectation- (left) and slope- (right) plots for three simple models. The shaded areas represent the 95\%-uncertainty regions as described in \cref{ShadeRem} and were omitted for the slope plots to allow for a better visual comparison of the line estimates.}
    \label{fig:ExpectationSlope}
\end{figure}
\noindent This visualization itself illustrates that these three models can indeed be compared despite their different distributional assumptions. \Cref{fig:ForestPlots} shows how they may, furthermore, be compared in a forest plot of generalized marginal effects, and how the same does not hold for odds ratios or $\beta$-coefficients.

\begin{figure}[h]
    \centering
    \includegraphics[width=\linewidth]{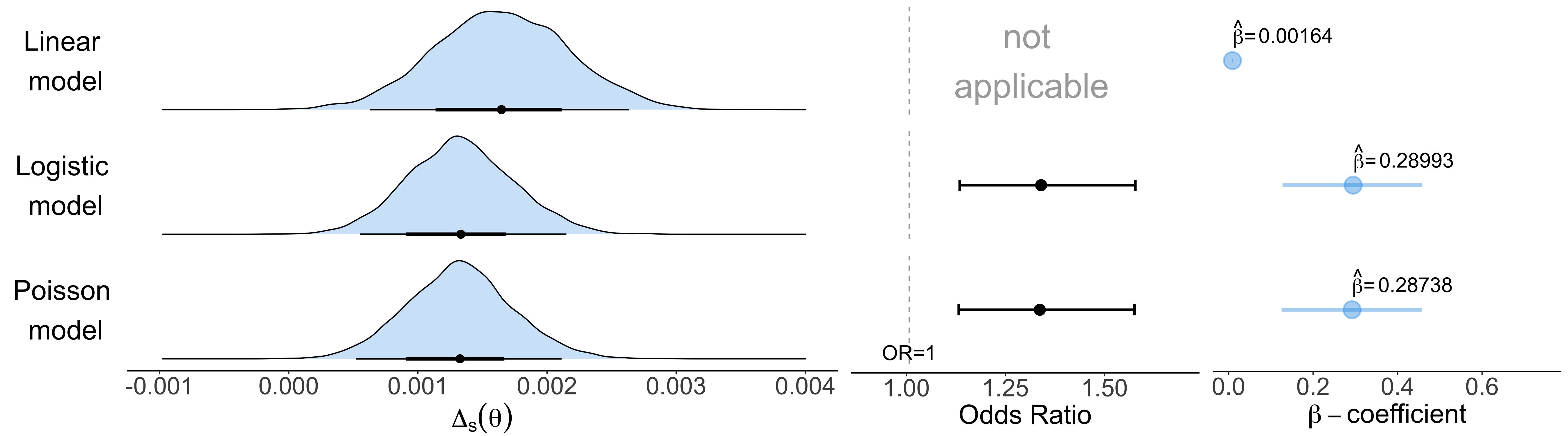}
    \if1\JASA{\vspace{-1cm}}\fi
    \caption{A comparison of three possible reporting methods for the three model results.\\ For deriving $\measure$, $\mu$ was chosen as the measure associated with the $U(0,1)$ distribution.}
    \label{fig:ForestPlots}
\end{figure}

\noindent Please note that, while odds ratios may obviously not be derived from the $\beta$-coefficients of linear models, comparing the results of Poisson models and logistic models in terms of odds ratios, as \cite{SILBERZAHN} did, is also not generally admissible, but only in settings where the prevalence of `successes' (in this case the event that a red card is given) is very low. This was the case in the \cite{SILBERZAHN} study. However, in such situations, the odds ratio, which, in general, is already almost impossible to sensibly interpret (\cite{Sackett1996}) becomes somewhat misleading. Plotting odds ratios in a forest plot with a line drawn at $1$, i.e. the odds ratio value representing `no effect', visually implies that an odds ratio of, for example, $1.3$ must indicate some non-negligible effect - after all, the odds for a player with $1$ as standardized average skin color rating are 30\% higher than for a player whose rating is $0$. However, when the probability of an event, in this case a red card being given, is infinitesimally small, it is not only approximately equal to the odds but a 30\% increase is, by extension, also infinitesimal. \Cref{fig:IndivPost} illustrates this by comparing the line estimates and uncertainty regions of the individualized predictive distribution, as introduced in \cref{UncertaintySection}, under the individualized expectations $e(\theta\vert \{0\},\delta_{\{0\}})$ and $e(\theta\vert \{1\},\delta_{\{1\}})$ for the linear and logistic models, since here $p=1$; and under $e(\theta\vert \{0\},\{1\},\delta_{\{0\}},\delta_{\{1\}})$ and $e(\theta\vert \{1\},\{1\},\delta_{\{1\}},\delta_{\{1\}})$ for the Poisson model, where the number of games may be seen as the second regressor $X_2$, \mbox{which appears as ln$(X_2)$ in the linear predictor.}

\begin{figure}[H]
    \centering
    \includegraphics[width=\linewidth]{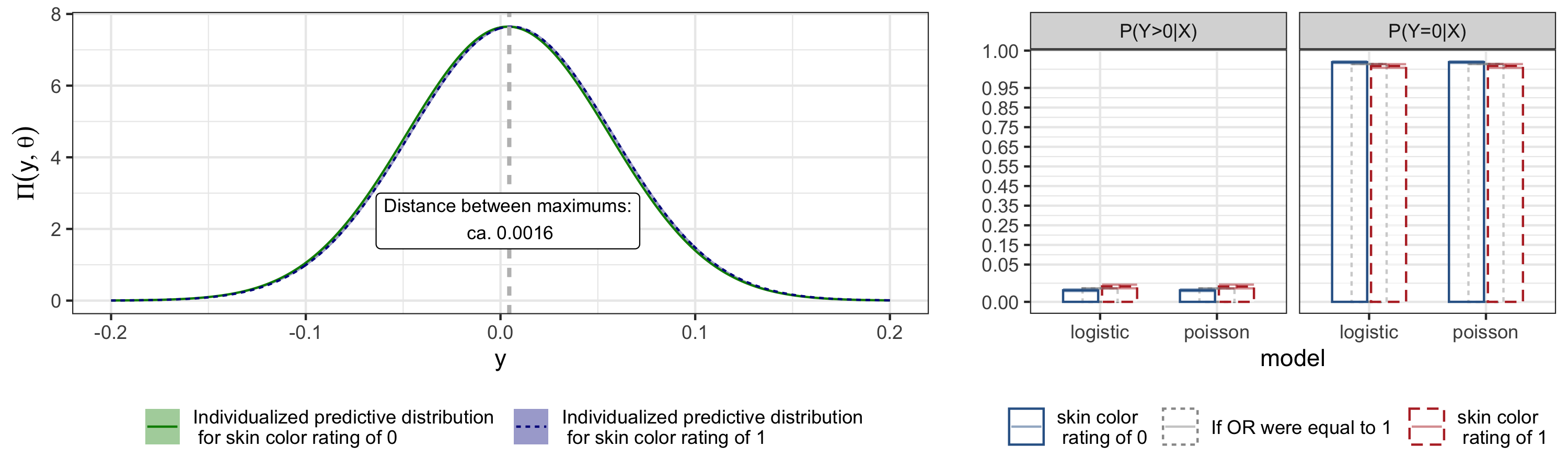}
    \if1\JASA{\vspace{-1cm}}\fi
    \caption{The individualized predictive distribution for the linear model (left) and Poisson \& logistic models (right), with the shaded areas representing the 95\%-uncertainty regions.}
    \label{fig:IndivPost}
\end{figure}\if1\JASA{\vspace{-.8\baselineskip}}\fi

\noindent To summarize, we have now shown that the framework developed in this work is not only useful to effectively compare effect sizes derived from different model classes on a common scale; but also to improve the interpretability of the differences in effect size that one may observe on this common scale.

\section{Conclusion and Future Work}
In this work, we have presented a formal framework for the derivation of a variety of visualization techniques and effect size measures that advance the interpretability, comparability, and communicability of results from  parametric regression. We have shown how numerous existing methodologies fall into this framework and demonstrated how methods derived from it can concretely be used to improve and clarify the reporting of model results in various parametric settings.
Based on this, we believe the following two extensions of the current framework to be of particular interest for future work. Firstly, extending the setting of all methods to include semi- and non- parametric models, especially survival models. Not only are the hazards in the interpretation of hazard ratios well known (\cite{Hernan2010}) and could be avoided by deriving a different effect size measure for proportional hazard models; but we also believe the developed framework to be particularly relevant in the context of medical decision making, e.g. in the setting of the PATH statement formulated by \cite{Kent2020}. Secondly, building on existing work such as \cite{Lei2018} to extend the theory of \cref{UncertaintySection} to settings without parametric distributional
assumptions, both for parametric and nonparametric models. In this context, one could furthermore define predictive credible sets around the individualized expectation that take into account both sampling and estimation uncertainty, with the ultimate aim of advancing how well statistical \mbox{results may be communicated not only to fellow scientists but also to the general public.}
 
\if1\Preprint{\section*{Acknowledgments} The authors gratefully acknowledge funding by LMUexcellence within the framework of the German Strategy of Excellence.}\fi

\if1\JASA{

\section*{Supplementary Materials} The supplementary materials contain appendices A-D as well as R-code for the analyses of \cref{Datasection}. Specifically, \cref{AppDefs} gives relevant mathematical definitions and notation, \cref{Interactions} provides further information regarding separately quantifying generalized marginal effects for settings with interaction terms including the full, formal \cref{SuperDef}(iii), \cref{Proofs} contains all proofs, and suggestions
for reporting standardized versions of all proposed quantities are given in \cref{StandAlg}.}\fi

\small
\bibliographystyle{agsm}
\bibliography{refs}

@article{SILBERZAHN,
	title = {Many analysts, one data set: Making transparent how variations in analytic choices affect results},
	journal = {Advances	in Methods and Practices in Psychological Science},
	volume = {1},
	pages = {337–356},
	year = {2018},
	doi = {https://doi.org/10.1177/2515245917747646},
	author= {Silberzahn, R. and Uhlmann, E. L. and Martin, D. P. and Anselmi, P. and Aust, F. and Awtrey, E. and Nosek, B. A.}
}

@article{SCHWEINSBERG,
title = {Same data, different conclusions: Radical dispersion in empirical results when independent analysts operationalize and test the same hypothesis},
journal = {Organizational Behavior and Human Decision Processes},
volume = {165},
pages = {228-249},
year = {2021},
issn = {0749-5978},
doi = {https://doi.org/10.1016/j.obhdp.2021.02.003},
 author={Schweinsberg, Martin and Feldman, Michael and Staub, Nicola and van den Akker, Olmo R and van Aert, Robbie CM and Van Assen, Marcel ALM and Liu, Yang and Althoff, Tim and Heer, Jeffrey and Kale, Alex and others},
}

@article{Wasserstein2019,
author = {Ronald L. Wasserstein and Allen L. Schirm and Nicole A. Lazar},
title = {Moving to a World Beyond ``p$<$0.05''},
journal = {The American Statistician},
volume = {73},
number = {sup1},
pages = {1-19},
year  = {2019},
publisher = {Taylor & Francis}}

@article{wasserstein2016asa,
  title={The ASA’s statement on p-values: context, process, and purpose},
  author={Wasserstein, Ronald L and Lazar, Nicole A and others},
  journal={The American Statistician},
  volume={70},
  number={2},
  pages={129--133},
  year={2016}
}

@article{Cassidy2019,
  title={Failing grade: 89\% of introduction-to-psychology textbooks that define or explain statistical significance do so incorrectly},
  author={Cassidy, Scott A and Dimova, Ralitza and Gigu{\`e}re, Benjamin and Spence, Jeffrey R and Stanley, David J},
  journal={Advances in Methods and Practices in Psychological Science},
  volume={2},
  number={3},
  pages={233--239},
  year={2019},
  publisher={Sage Publications Sage CA: Los Angeles, CA}
}

@article{Hoogeveen2022,
  title={Many-Analysts Religion Project: Reflection and Conclusion},
  author={Hoogeveen, Suzanne and Sarafoglou, Alexandra and van Elk, Michiel and Wagenmakers, Eric-Jan},
  year={2022},
  publisher={PsyArXiv}
}

@article{Wagenmakers2022,
  title={One statistical analysis must not rule them all},
  author={Wagenmakers, Eric-Jan and Sarafoglou, Alexandra and Aczel, Balazs},
  journal={Nature},
  volume={605},
  number={7910},
  pages={423--425},
  year={2022}
}

@article{Gelman_AveragePredictiveComparison,
	author = {Andrew Gelman and Iain Pardoe},
	title ={2. Average Predictive Comparisons for Models with Nonlinearity, Interactions, and Variance Components},
	journal = {Sociological Methodology},
	volume = {37},
	number = {1},
	pages = {23-51},
	year = {2007},
	doi = {10.1111/j.1467-9531.2007.00181.x},
		eprint = { 
		https://doi.org/10.1111/j.1467-9531.2007.00181.x
		
	}
}

@article{HDRs,
	author = { Rob J.   Hyndman },
	title = {Computing and Graphing Highest Density Regions},
	journal = {The American Statistician},
	volume = {50},
	number = {2},
	pages = {120-126},
	year  = {1996},
	publisher = {Taylor & Francis},
	doi = {10.1080/00031305.1996.10474359},

	eprint = { 
	https://www.tandfonline.com/doi/pdf/10.1080/00031305.1996.10474359
	
	}
	
}

@article{Williams2012,
	author = "Williams, R.",
	title = "Using the margins command to estimate and interpret adjusted predictions and marginal effects",
	journal = "Stata Journal",
	publisher = "Stata Press",
	address = "College Station, TX",
	volume = "12",
	number = "2",
	year = "2012",
	pages = "308-331(24)"
}

@Article{Fubini,
	Author = {Fubini, G.},
	Title = {Sugli integrali multipli.},
	FJournal = {Accademia dei Lincei, Rendiconti, V. Serie},
	Journal = {Rom. Acc. L. Rend. (5)},
	ISSN = {0001-4435},
	Volume = {16},
	Number = {1},
	Pages = {608--614},
	Year = {1907},
	Language = {Italian},
	zbMATH = {2643959},
	JFM = {38.0343.02}
}

@Manual{marginaleffects,
	title = {marginaleffects: Marginal Effects, Marginal Means, Predictions, and Contrasts},
	author = {Vincent Arel-Bundock},
	year = {2022},
	note = {R package version 0.5.0, \textbf{URL:} \url{ https://vincentarelbundock.github.io/marginaleffects/}}
}

@book{StataSoftware,
	author={StataCorp},
	year={2021},
	title={Stata Statistical Software: Release 17},
	note={College Station, TX: StataCorp LLC.}
}

@Article{ALEml,
  author={Daniel W. Apley and Jingyu Zhu},
  title={{Visualizing the effects of predictor variables in black box supervised learning models}},
  journal={Journal of the Royal Statistical Society Series B},
  year=2020,
  volume={82},
  number={4},
  pages={1059-1086},
  month={September},
  doi={10.1111/rssb.12377}
}

@book{Park2017,
  title={Fundamentals of Probability and Stochastic Processes with Applications to Communications},
  author={Park, K.I.},
  isbn={9783319680750},
  year={2017},
  publisher={Springer International Publishing}
}

@article{BischlHeumann,
  title={Marginal Effects for Non-Linear Prediction Functions},
  author={Scholbeck, Christian A and Casalicchio, Giuseppe and Molnar, Christoph and Bischl, Bernd and Heumann, Christian},
  journal={arXiv preprint arXiv:2201.08837},
  year={2022}
}

@article{RCT,
title={Clinical Dataset of the CYP-GUIDES Trial},
author={Joseph Tortora and Saskia Robinson and Seth Baker and Gualberto Ruaño},
year={2020},
journal={Mendeley Data, V1},
note={doi: 10.17632/25yjwbphn4.1}}

@article{Friedmann,
author = {Jerome H. Friedman},
title = {{Greedy function approximation: A gradient boosting machine.}},
volume = {29},
journal = {The Annals of Statistics},
number = {5},
publisher = {Institute of Mathematical Statistics},
pages = {1189 -- 1232},
year = {2001},
doi = {10.1214/aos/1013203451},
}

@article{Goldstein2015,
author = {Alex Goldstein and Adam Kapelner and Justin Bleich and Emil Pitkin},
title = {Peeking Inside the Black Box: Visualizing Statistical Learning With Plots of Individual Conditional Expectation},
journal = {Journal of Computational and Graphical Statistics},
volume = {24},
number = {1},
pages = {44-65},
year  = {2015},
publisher = {Taylor & Francis},
doi = {10.1080/10618600.2014.907095},


}

@book{GelmanHill2007, place={Cambridge}, series={Analytical Methods for Social Research}, title={Data Analysis Using Regression and Multilevel/Hierarchical Models}, DOI={10.1017/CBO9780511790942}, publisher={Cambridge University Press}, author={Gelman, Andrew and Hill, Jennifer}, year={2007}, collection={Analytical Methods for Social Research}}

@article{Kaufman1996,
 ISSN = {00384941, 15406237},
 author = {Robert L. Kaufman},
 journal = {Social Science Quarterly},
 number = {1},
 pages = {90--109},
 publisher = {[University of Texas Press, Wiley]},
 title = {Comparing Effects in Dichotomous Logistic Regression: A Variety of Standardized Coefficients},
 volume = {77},
 year = {1996}
}

@article{Fuller2015,
title = {The Risk GP Model: The standard model of prediction in medicine},
journal = {Studies in History and Philosophy of Science Part C: Studies in History and Philosophy of Biological and Biomedical Sciences},
volume = {54},
pages = {49-61},
year = {2015},
issn = {1369-8486},
doi = {https://doi.org/10.1016/j.shpsc.2015.06.006},
author = {Jonathan Fuller and Luis J. Flores}
}

@article{Waehrer2020,
    doi = {10.1371/journal.pone.0226134},
    author = {Waehrer, Geetha M. AND Miller, Ted R. AND Silverio Marques, Sara C. AND Oh, Debora L. AND Burke Harris, Nadine},
    journal = {PLOS ONE},
    publisher = {Public Library of Science},
    title = {Disease burden of adverse childhood experiences across 14 states},
    year = {2020},
    month = {01},
    volume = {15},
    pages = {1-18},
    number = {1},

}

@article{Buchman2020,

	author = {Buchman, Timothy G and Simpson, Steven Q and Sciarretta, Kimberly L and Finne, Kristen P and Sowers, Nicole and Collier, Michael and Chavan, Saurabh and Oke, Ibijoke and Pennini, Meghan E and Santhosh, Aathira and Wax, Marie and Woodbury, Robyn and Chu, Steve and Merkeley, Tyler G and Disbrow, Gary L and Bright, Rick A and MaCurdy, Thomas E and Kelman, Jeffrey A},
	cin = {Crit Care Med. 2020 Mar;48(3):424-426. PMID: 32058378},
	crdt = {2020/02/15 06:00},
	date = {2020 Mar},
	dcom = {20201023},
	doi = {10.1097/CCM.0000000000004225},
	edat = {2020/02/15 06:00},
	issn = {1530-0293 (Electronic); 0090-3493 (Print); 0090-3493 (Linking)},
	jid = {0355501},
	journal = {Crit Care Med},
	jt = {Critical care medicine},
	language = {eng},
	lid = {10.1097/CCM.0000000000004225 {$[$}doi{$]$}},
	lr = {20201023},
	mh = {Age Factors; Aged; Aged, 80 and over; Centers for Medicare and Medicaid Services, U.S.; Comorbidity; Fee-for-Service Plans/statistics \& numerical data; Female; Health Expenditures/*statistics \& numerical data; Hospitalization/*statistics \& numerical data; Humans; Male; Medicare/*statistics \& numerical data; Medicare Part C/economics; Models, Statistical; Quality of Life; Sepsis/*mortality; Severity of Illness Index; Shock, Septic/mortality; United States/epidemiology},
	mhda = {2020/10/24 06:00},
	month = {Mar},
	number = {3},
	own = {NLM},
	pages = {302--318},
	phst = {2020/02/15 06:00 {$[$}entrez{$]$}; 2020/02/15 06:00 {$[$}pubmed{$]$}; 2020/10/24 06:00 {$[$}medline{$]$}},
	pii = {00003246-202003000-00004},
	pmc = {PMC7017950},
	pmid = {32058368},
	pst = {ppublish},
	pt = {Journal Article; Research Support, Non-U.S. Gov't; Research Support, U.S. Gov't, Non-P.H.S.; Research Support, U.S. Gov't, P.H.S.},
	sb = {IM},
	status = {MEDLINE},
	title = {Sepsis Among Medicare Beneficiaries: 3. The Methods, Models, and Forecasts of Sepsis, 2012-2018.},
	volume = {48},
	year = {2020}}

@article{Kent2020,
	author = {Kent, David M and Paulus, Jessica K and van Klaveren, David and D'Agostino, Ralph and Goodman, Steve and Hayward, Rodney and Ioannidis, John P A and Patrick-Lake, Bray and Morton, Sally and Pencina, Michael and Raman, Gowri and Ross, Joseph S and Selker, Harry P and Varadhan, Ravi and Vickers, Andrew and Wong, John B and Steyerberg, Ewout W},
	cin = {Ann Intern Med. 2020 Jan 7;172(1):63-64. PMID: 31711098; Ann Intern Med. 2020 Jun 2;172(11):776. PMID: 32479147; Ann Intern Med. 2020 Jun 2;172(11):775-776. PMID: 32479148},
	crdt = {2019/11/12 06:00},
	date = {2020 Jan 7},
	dcom = {20200806},
	dep = {20191112},
	doi = {10.7326/M18-3667},
	edat = {2019/11/12 06:00},
	gr = {1900/PCORI{\_}/Patient-Centered Outcomes Research Institute/United States; P30 CA008748/CA/NCI NIH HHS/United States; RR-1705-0001/PCORI{\_}/Patient-Centered Outcomes Research Institute/United States},
	issn = {1539-3704 (Electronic); 0003-4819 (Print); 0003-4819 (Linking)},
	jid = {0372351},
	journal = {Ann Intern Med},
	jt = {Annals of internal medicine},
	language = {eng},
	lid = {10.7326/M18-3667 {$[$}doi{$]$}},
	lr = {20220129},
	mh = {Clinical Decision Rules; Clinical Decision-Making; Evidence-Based Medicine/standards; Humans; Individuality; Models, Statistical; Randomized Controlled Trials as Topic/*standards; Risk Assessment; *Treatment Outcome},
	mhda = {2020/08/07 06:00},
	mid = {HRAMS1620198},
	month = {Jan},
	number = {1},
	own = {NLM},
	pages = {35--45},
	phst = {2019/11/12 06:00 {$[$}pubmed{$]$}; 2020/08/07 06:00 {$[$}medline{$]$}; 2019/11/12 06:00 {$[$}entrez{$]$}},
	pii = {2755582},
	pmc = {PMC7531587},
	pmid = {31711134},
	pst = {ppublish},
	pt = {Journal Article},
	sb = {IM},
	status = {MEDLINE},
	title = {The Predictive Approaches to Treatment effect Heterogeneity (PATH) Statement.},
	volume = {172},
	year = {2020}}

@article{Lei2018,
author = {Jing Lei and Max G’Sell and Alessandro Rinaldo and Ryan J. Tibshirani and Larry Wasserman},
title = {Distribution-Free Predictive Inference for Regression},
journal = {Journal of the American Statistical Association},
volume = {113},
number = {523},
pages = {1094-1111},
year  = {2018},
publisher = {Taylor & Francis},
doi = {10.1080/01621459.2017.1307116}
}

@article{Cook1997,
author = { R.   Dennis   Cook  and  Sanford   Weisberg },
title = {Graphics for Assessing the Adequacy of Regression Models},
journal = {Journal of the American Statistical Association},
volume = {92},
number = {438},
pages = {490-499},
year  = {1997},
publisher = {Taylor & Francis},
doi = {10.1080/01621459.1997.10474002}

}

@article{Zhao2021,
author = {Qingyuan Zhao and Trevor Hastie},
title = {Causal Interpretations of Black-Box Models},
journal = {Journal of Business \& Economic Statistics},
volume = {39},
number = {1},
pages = {272-281},
year  = {2021},
publisher = {Taylor & Francis},
doi = {10.1080/07350015.2019.1624293},

}

@article {Sackett1996,
	author = {Sackett, David L. and Deeks, Jonathan J. and Altman, Doughs G.},
	title = {Down with odds ratios!},
	volume = {1},
	number = {6},
	pages = {164--166},
	year = {1996},
	doi = {10.1136/ebm.1996.1.164},
	publisher = {Royal Society of Medicine},
	issn = {1356-5524},
	eprint = {https://ebm.bmj.com/content/1/6/164.full.pdf},
	journal = {BMJ Evidence-Based Medicine}
}

@article{Greenland1987,
    author = {Greenland, Sander},
    title = "{INTERPRETATION AND CHOICE OF EFFECT MEASURES IN EPIDEMIOLOGIC ANALYSES1}",
    journal = {American Journal of Epidemiology},
    volume = {125},
    number = {5},
    pages = {761-768},
    year = {1987},
    month = {05},
    issn = {0002-9262},
    doi = {10.1093/oxfordjournals.aje.a114593},
    eprint = {https://academic.oup.com/aje/article-pdf/125/5/761/287809/125-5-761.pdf},
}

@article{Hernan2010,
	address = {Department of Epidemiology, Harvard School of Public Health, and the Harvard-MIT Division of Health Sciences and Technology, Boston, MA 02115, USA. miguel_hernan@post.harvard.edu},
	author = {Hern{\'a}n, Miguel A},
	cin = {Epidemiology. 2010 May;21(3):429-30; author reply 430-1. PMID: 20386180; Epidemiology. 2013 Sep;24(5):777-8. PMID: 23903883},
	crdt = {2009/12/17 06:00},
	date = {2010 Jan},
	dcom = {20100309},
	doi = {10.1097/EDE.0b013e3181c1ea43},
	edat = {2009/12/17 06:00},
	ein = {Epidemiology. 2011 Jan;22(1):134},
	gr = {R01 HL080644/HL/NHLBI NIH HHS/United States},
	issn = {1531-5487 (Electronic); 1044-3983 (Print); 1044-3983 (Linking)},
	jid = {9009644},
	journal = {Epidemiology},
	jt = {Epidemiology (Cambridge, Mass.)},
	language = {eng},
	lid = {10.1097/EDE.0b013e3181c1ea43 {$[$}doi{$]$}},
	lr = {20220414},
	mh = {*Epidemiologic Studies; *Proportional Hazards Models},
	mhda = {2010/03/10 06:00},
	mid = {NIHMS461625},
	month = {Jan},
	number = {1},
	own = {NLM},
	pages = {13--15},
	phst = {2009/12/17 06:00 {$[$}entrez{$]$}; 2009/12/17 06:00 {$[$}pubmed{$]$}; 2010/03/10 06:00 {$[$}medline{$]$}},
	pii = {00001648-201001000-00004},
	pmc = {PMC3653612},
	pmid = {20010207},
	pst = {ppublish},
	pt = {Journal Article; Research Support, N.I.H., Extramural},
	sb = {IM},
	status = {MEDLINE},
	title = {The hazards of hazard ratios},
	volume = {21},
	year = {2010},
}

@article{Chrutchley2022,
	address = {Department of Pharmacotherapy, College of Pharmacy and Pharmaceutical Sciences, Washington State University, Yakima, WA, United States.; School of Pharmacy, University of the Western Cape, Cape Town, South Africa.},
	author = {Crutchley, Rustin D and Keuler, Nicole},
	cois = {The authors declare that the research was conducted in the absence of any commercial or financial relationships that could be construed as a potential conflict of interest.},
	copyright = {Copyright {\copyright}2022 Crutchley and Keuler.},
	crdt = {2022/05/02 06:48},
	date = {2022},
	dep = {20220412},
	doi = {10.3389/fphar.2022.884213},
	edat = {2022/05/03 06:00},
	issn = {1663-9812 (Print); 1663-9812 (Electronic); 1663-9812 (Linking)},
	jid = {101548923},
	journal = {Front Pharmacol},
	jt = {Frontiers in pharmacology},
	language = {eng},
	lid = {10.3389/fphar.2022.884213 {$[$}doi{$]$}; 884213},
	lr = {20220716},
	mhda = {2022/05/03 06:01},
	oto = {NOTNLM},
	own = {NLM},
	pages = {884213},
	phst = {2022/02/25 00:00 {$[$}received{$]$}; 2022/03/23 00:00 {$[$}accepted{$]$}; 2022/05/02 06:48 {$[$}entrez{$]$}; 2022/05/03 06:00 {$[$}pubmed{$]$}; 2022/05/03 06:01 {$[$}medline{$]$}},
	pii = {884213},
	pmc = {PMC9039251},
	pmid = {35496293},
	pst = {epublish},
	pt = {Journal Article},
	status = {PubMed-not-MEDLINE},
	title = {Sub-Analysis of CYP-GUIDES Data: Assessing the Prevalence and Impact of Drug-Gene Interactions in an Ethnically Diverse Cohort of Depressed Individuals.},
	volume = {13},
	year = {2022}}

\newpage
\normalsize
\appendix
\setcounter{equation}{0}
\numberwithin{equation}{section}
\titleformat{\section}
{\normalfont\Large\bfseries}{Appendix~\thesection}{1em}{}

\begin{changemargin}{-1cm}{-1cm}
\begin{center}
    \textbf{\LARGE Appendices}
\end{center}\setcounter{page}{1}
\section{Relevant mathematical definitions and notation\label{AppDefs}}

\begin{rems} Given that the usage of terminology regarding the functions that define continuous and discrete probability distributions is not always consistent within the literature, we would like to clarify that, within this work, we use \begin{enumerate}[label=(\roman*)]
		\item the term \emph{density} to refer to any measurable function $f:\R^d\longrightarrow\R$, $d\in\N_{>0}$, fulfilling the following requirements\begin{enumerate}[label=\arabic*.]
			\item $f(x)\geq0,\quad \forall x\in \R^d$.
			\item $\int_{\R^d} f(x)d\lambda(x)=1$, where $\lambda$ denotes the Lebesgue measure.
		\end{enumerate}

		\item the term \textit{probability function} to refer to any function $p:\R^d\longrightarrow[0,1]$, $d\in\N_{>0}$, fulfilling the following requirements\begin{enumerate}[label=\arabic*.]
			\item $\supp(p)=\{x\in\R^d:p(x)\neq0\}$ is countable.
			\item $\underset{x\in\supp(p)}{\sum}p(x)=1$.
		\end{enumerate}
	\end{enumerate}
\end{rems}
\begin{rems}\label{assMrem} $\,$
For a probability measure $\mu$ on $(\R^d,\mathcal{B}(\R^d))$, $d\in\N_{>0}$, we write that
	\begin{enumerate}
		\item \emph{$\mu$ is defined via} density $f$ or probability function $p$, if $\forall A\in\mathcal{B}(\R^d)$\[
			\mu(A)=\int_A f(x)d\lambda(x)\quad \quad \text{or}\quad \quad\mu(A)=\underset{x\in\supp(p)\cap A}{\sum}p(x)
		\] holds, respectively. Furthermore, we write $\mu(A)=\int_Ad\mu(x)$ in both cases.
	\item \emph{$\mu$ is associated with} a known distribution if $\mu$ is defined via the distribution's density or probability function.
	\end{enumerate}
\end{rems}

\begin{rems}\label{Aiinote}
For a probability measure $\mu$ that is defined via a conditional density $f_{X\vert Y}$ or probability function $p_{X\vert Y}$, we use the following notation
\begin{align*}
    \int_A \mu(dx,y)=\begin{cases}
    \int_A f_{X\vert Y}(x,y) d\lambda(x),&\text{if $\mu$ is defined via a density}\\
    \underset{x\in\supp(p)\cap A}{\sum} p_{X\vert Y}(x,y),&\text{if $\mu$ is defined via a probability function}\,.
    \end{cases}
\end{align*}
  Note that these terms are still variable in $y$.
\end{rems}

\begin{Adefi}[Normalized Probability Measures]\label{normmuDef}
	Given $d\in\N_{>0}$ and a probability measure $\mu$ on $\big(\R^d,\mathcal{B}(\R^d)\big)$ defined via either density $f$ or probability function $p$ and a  set $A\in\mathcal{B}(\R^d)$, we define $\mu$ to be \emph{normalized w.r.t. $A$}, if $\mu$ is then defined via the density $
	\mu(A)^{-1} \cdot f(x) \cdot\ind{x\in A}
	$ or the probability function
	$
	\mu(A)^{-1} \cdot p(x) \cdot\ind{x\in A}\,,
	$ respectively.
\end{Adefi}

\begin{Adefi}[Highest Density Regions as proposed by \cite{HDRs}]\label{HDR}
	 Let $\pi$ be a density or probability function via which a probability measure $\mu_\pi$ on $(\R^d,\mathcal{B}(\R^d))$, $d\in\N_{>0}$, is defined. For any $\alpha\in[0,1]$, we then define the \emph{highest density region with level $\alpha$ of $\mu_\pi$}, or \emph{$100(1-\alpha)$\%-HDR of $\mu_\pi$}, as the set  $$R( \pi_\alpha):=\{x\in\R^d\vert \pi(x)\geq\pi_\alpha\}\,,$$ 
	 where $\pi_\alpha$ is the largest constant that satisfies $\mu_\pi\big(R(\pi_\alpha)\big)\geq 1-\alpha$.
	 
\end{Adefi}

\if1\Preprint{\pagebreak}\fi
\begin{Adefi}[Support of measures on specific measurable spaces\label{SupportDef}]$\,$\hfill
		For a non-negative measure $\mu$ on $(\R^d,\mathcal{B}(\R^d))$, $d\in\N_{>0}$, we define the support of $\mu$ as
		$$\supp(\mu):=\{x\in\R^d\,\big\vert\, \forall \varepsilon>0:\mu(B_\varepsilon(x))>0 \}\, ,$$
		where $B_\varepsilon(x)$ denotes the \emph{epsilon ball} around $x$, defined as \mbox{$B_\varepsilon(x):=\{y\in\R^d\,\big\vert\,\Vert x-y\Vert<\varepsilon\}$} for some $\varepsilon>0$.\\[-5pt]
\end{Adefi}

\begin{rems}\label{SplitRem}
 In this work, we often utilize specific notation for subgroups of entries of some vector, e.g. denoting the entries of the vector of regressor-realizations as two distinct vectors $\realx{I}$ and $\realx{\text{MC}}$. Consequently, when we consider  some function \mbox{$h:\R^d\longrightarrow\R$}, $d\in\N_{>0}$, that maps a vector $(x_1,...,x_d)^\top$ to a value in $\R$ and some notation that maps all entries of a vector in $\R^d$ to either $x_a$ or $x_b$, we write $h(x)$ and $h(x_a,x_b)$ interchangeably, depending on which version we consider optimal for clarity. 
\end{rems}

\newpage
\section{Interaction terms\label{Interactions}}
This appendix contains more detailed information about the proposed method to separately quantify the generalized marginal effect $\measure$ for main and interaction effects that were omitted in \cref{gMEsec}. After providing motivation for our general approach to interaction effects, sections \cref{I1} and \cref{MetIntCor} give definitions of appropriate probability measures for separately considering interactions between metric regressors, and the generalized marginal effect for separately quantifying main and interaction effects for interactions between metric and categorical regressors, respectfully. 

\begin{rems} Within this work, we use the expression \emph{interaction term} to refer to the product of two or more entries of the regressor vector $X$ appearing in the function term of $\regf$, if applicable after the inclusion of `new regressor(s)' in the sense of \cref{TransRem}. We say that a regressor is \emph{included in an interaction term} if it is one of the factors in said interaction term. Furthermore, we refer to an interaction term as \emph{being of higher order} than another term if the first may be written as a product in which the latter is one factor.
\end{rems}
As explained in \cref{gMEsec}, our approach to this issue is to operate under the simplifying assumption that the interaction terms behave independently of the regressors that are being multiplied. Since this assumption clearly does not correspond to reality, the resulting generalized marginal effects should be interpreted with great caution. However, the resulting effect sizes appear to be optimal in terms of comparability. Indeed, applying this assumption for linear regression with interactions results simply in the reporting of $\beta$-coefficients, as \cref{BetasProp} proves. The following example provides a more intuitive motivation.

\begin{ex*}
Consider the case of linear regression that includes two \emph{metric} regressors, $X_1$ and $X_2$ as well as their interaction, i.e. the expectation function is given by \begin{align*}
    \regf(X)=\beta_0+\beta_1\cdot X_1+\beta_2\cdot X_2+\beta_3\cdot X_1\cdot X_2\,.
\end{align*}
In such a setting, it is common practice to report the point estimate $(\hat{\beta}_0,\hat{\beta}_1,\hat{\beta}_2,\hat{\beta}_3)^\top$ as part of the inference results. However, importantly, these $\beta$-coefficients may not be interpreted as `\emph{keeping all other regressors constant, the expected value of the target variable increases by $\hat{\beta}_i$ when $X_i$ increases by $1$.}', $i\in\{1,2,3\}$. Neither is the interpretation `\emph{keeping $X_1$ constant, the expected value of the target variable increases by $\hat{\beta}_2+\hat{\beta}_3$ when $X_2$ increases by $1$.}' admissible, since $\EW[Y\vert X]$ actually increases by $\hat{\beta}_2+\hat{\beta}_3\cdot X_1$ when $X_2$ increases by $1$, which is still a function of the metric regressor $X_1$. Still, the $\beta$-coefficients of models such as the one above are being reported, as they still contain valuable information. Making the assumption that an interaction term behaves independently of the regressors being multiplied simply allows us to quantify the equivalent of linear $\beta$-coefficients in non-linear settings.
\end{ex*}

\noindent The following will now, firstly, specify how the set $\RI$ and probability measure $\mu$ are to be chosen if one wants to separate the main and interaction effects of only metric variables and, secondly, give the formal definition omitted in \cref{SuperDef}(iii).

\subsection{Interactions between metric regressors\label{I1}}
For interactions between metric regressors, we propose to treat
the interaction term as a new, independent regressor, as was also explained in \cref{gMEsec}. Specifically, this entails choosing either one single metric regressor or one interaction term as regressor of interest $X^I$ and mapping \emph{all interaction terms of higher order} as well as, if applicable, \emph{all metric regressors involved in said interaction term} to the set of remaining metric regressors $M$.

While choosing to assume independence of an interaction term to separately quantify main and interaction effects is reasonable in the interest of comparability, we do not believe it expedient to assume that interaction terms may take values that are not the product of values in the respective support of the measures chosen for each involved regressor.
Therefore, for the purpose of choosing an appropriate probability measure for either the new regressor of interest or the added elements of $M$, we require the following definition.

\begin{Aadefi}\label{IreqDef}
	\begin{enumerate}[leftmargin=0.25in]
		\item \mbox{$\forall A,B\subseteq\R$ we define $\boldsymbol{(}A\boldsymbol{\cdot} B\boldsymbol{)}:=\{c\in \R\,\big\vert\, c=a\cdot b \,\text{ for some }\, a\in A,\, b\in B\}$}.\\ Note that, if $A$, $B$, or both sets are chosen as $\R$ and neither as $\{0\}$, $\boldsymbol{(}A\boldsymbol{\cdot} B\boldsymbol{)}=\R$.
		\item For two probability measures  $\mu_1$ and $\mu_2$ on $\big(\R,\mathcal{B}(\R)\big)$, we denote by $\boldsymbol{(}\mu_1\boldsymbol{\cdot} \mu_2\boldsymbol{)}$ a measure on $\big(\R,\mathcal{B}(\R)\big)$ that fulfills the following requirements w.r.t.   $\mu_1$ and $\mu_2$ \begin{itemize}[leftmargin=1.5cm]
			\item[(I1)] $\boldsymbol{(}\mu_1\boldsymbol{\cdot} \mu_2\boldsymbol{)}$ is a probability measure on $\big(\R,\mathcal{B}(\R)\big)$
			\item[(I2)]$\supp\big(\boldsymbol{(}\mu_1\boldsymbol{\cdot} \mu_2\boldsymbol{)}\big)=\regprodCl{\supp(\mu_1)}{\supp(\mu_2)}$
			\item[(I3)] $\forall M\in\mathcal{B}(\R)$, the following holds for $M/z:=\{x\in\R\vert x\cdot z\in M\}$\begin{equation}\label{ReqI3}
				\boldsymbol{(}\mu_1\boldsymbol{\cdot} \mu_2\boldsymbol{)}(M)=\int_\R\int_{M/z}d\mu_1(x)d\mu_2(z)=\int_\R\int_{M/z}d\mu_2(x)d\mu_1(z)\,.
			\end{equation}
		\end{itemize}
	\end{enumerate}
\end{Aadefi}
\noindent In general, and in particular under the axiomatic assumption (A.I), we utilize probability measures not necessarily to represent randomness, but also to derive a non-random weighted average of a function's value over a certain set.
However, we expect the utilized measures to often be chosen as associated with a known distribution and, furthermore, believe that, to most readers, the notion of a random variable following a certain distribution is more accessible than a non-random average calculated w.r.t. a probability measure. Therefore, the following proposition expresses the above definition in terms of random variables.

\begin{restatable}{Aprop}{distlem}
\label{distlem}
	Let $X$ and $Y$ denote two independently distributed,  real-valued random variables. If the probability measures $\mu_X$ and $\mu_Y$ on $\big(\R,\mathcal{B}(\R)\big)$ are associated with the distributions of $X$ and $Y$, respectively, the measure $\boldsymbol{(}\mu_X\boldsymbol{\cdot} \mu_Y\boldsymbol{)}$ is associated with the distribution of the random variable $Z:=XY$.
	\end{restatable}
\noindent The proof of this proposition may be found in \cref{Proofs}.

\paragraph{Metric interactions under assumptions (A.I) and (A.II$'$)} Under assumptions (A.I) and (A.II$'$), where the regressor of interest is treated as independent, the choice of sets and probability measures in the case of metric interactions is now very straightforward. One simply makes all specifications that would have been made when not choosing to separately quantify the effects. However, additionally, the interaction term of interest or interaction terms added to $M$ are treated as independent under \emph{both} assumptions and the corresponding sets and measures chosen as follows. For an interaction between two metric regressors $X_i$ and $X_j$, $i,j\in\{1,...,p\}$, one chooses $\regprod{\RI_i}{\RI_j}$ as the set to integrate over for the new `interaction-regressor' and $\regprod{\mu_i}{\mu_j}$ as the measure to integrate w.r.t., with $\mu_i$ and $\mu_j$ denoting any measure chosen for the independent regressors $X_i$ and $X_j$ under assumption (A.I); and, under assumption (A.II$'$), denoting one independently chosen measure for one regressor and the measure associated with the \emph{marginal distribution} of the other regressor.

For interactions of higher order, the set $\RI$ and measure $\mu$ may be derived iteratively in arbitrary order, as stated by the following corollary.

\begin{restatable}{Acor}{MetIntCor}
\label{MetIntCor}
	For any finite sets of probability measures on $\big(\R,\mathcal{B}(\R)\big)$ $\{\mu_i\}_{i=1,...,N}$ and subsets of $\R$ $\{\RI_i\}_{i=1,...,N}$, $N\in\N$, the following holds\begin{enumerate}
		\item $\regprodCl{\regprod{\regprod{\regprod{\RI_1}{\RI_2}}{\RI_3}}{...}}{\RI_N}=\regprodCl{\regprod{\regprod{\regprod{\RI_{\tau(1)}}{\RI_{\tau(2)}}}{\RI_{\tau(3)}}}{...}}{\RI_{\tau(N)}}
		$ for any permutation $\tau$ on the set $\{1,...,N\}$.
		\item $\regprod{\regprod{\regprod{\regprod{\mu_1}{\mu_2}}{\mu_3}}{...}}{\mu_N}=\regprod{\regprod{\regprod{\regprod{\mu_{\tau(1)}}{\mu_{\tau(2)}}}{\mu_{\tau(3)}}}{...}}{\mu_{\tau(N)}}
		$ for any permutation $\tau$  on the set $\{1,...,N\}$.
	\end{enumerate}
	\end{restatable}
\noindent The proof of this corollary may also be found in \cref{Proofs}.
\paragraph{Metric interactions under assumption (A.II$''$)} The approach to metric interactions under assumption (A.II$''$) is almost identical to the one under assumptions (A.I) and (A.II$'$). However, as the regressor of interest is not treated as independent under this assumption, we believe that choosing the measure to integrate w.r.t. for the a new `interaction-regressor' according to \cref{IreqDef} would again not be reasonable. Instead, the ideal approach would be to identify the measure associated with the distribution of a random variable that is defined as the product of two random variables with a joint distribution equal to the marginal distribution of the regressors involved in the interaction term. Of course, choosing the set to integrate over according to \cref{IreqDef} is still completely valid.

\subsection{Interactions between metric and categorical regressors\label{DivInt}}
When separately considering interaction effects of interactions between metric \emph{and} categorical regressors, we propose to again combine all categorical regressors, including the parts of the interaction terms that are categorical, into one regressor and consider the interactions between metric regressors as separate regressors, since the interaction between a categorical and metric variable can not be reasonably combined into one separate regressor.
In order to give a coherent definition for our effect size measure, however, we need to combine all these elements into one vector $X^I$. \\
We propose to denote the metric regressor, or interaction of several metric regressors which is viewed as one new metric variable, by $X^I_{met}$; and the categorical elements by  $X^I_{{cat}}=(X^I_1,...,X^I_{d_M},...,X^I_{d_C})$, where the entries $X^I_1$ to $X^I_{d_M}$  represent the non-reference categories of the categorical regressors that are included in an interaction with a metric regressor (or interaction of metric regressors viewed as a separate metric regressor) and entries $X^I_{d_M+1}$ to $X^I_{d_C}$ represent these same entries interacting with $X^I_{met}$. \\

Overall, the regressor of interest is then given by $
	X^I=\big(X^I_{met},X^I_{cat}\big)$, with $X^I_{met}$ taking values in $\R$ and $X^I_{cat}$ taking values in $\{0,1\}^{d_C}$ and the respective realizations being denoted by $\realx{I_{met}}$ and $\realx{I_{cat}}$. In the interest of consistent notation, we set $d_I=d_C+1$. Furthermore, we again utilize the vectors $\mathbf{v}_l\in \{0,1\}^{d_C}$ and $\mathbf{ref}_l\in \{0,1\}^{d_C}$ as defined in \cref{gMEsec}, for which the following provides a concrete example.
\begin{ex}\label{CatVecEx} \begin{enumerate}[label=(\roman*)]
    \item Consider a regression with categorical regressors \emph{`smoking'}, a binary variable with reference category "non-smoking", and \emph{`gender'}, with reference category "female" and non-reference categories "male" and "non-binary". If one were interested in separating the effect of \emph{`smoking'}, \emph{`gender'}, and their interaction, the combined regressor of interest would be given by \linebreak $X_{cat}^I=(X^I_1,X^I_2,X^I_3,X^I_4,X^I_5)^\top$, with\begin{center}
	$X^I_1\,\hat{=} \,smoking$, $X^I_2\,\hat{=}\,male$, $X^I_3\,\hat{=}\,non$-$binary$, $X^I_4\,\hat{=}\,smoking*male$, $X^I_5\,\hat{=}\,smoking*non$-$binary$\end{center}
and, furthermore, 	$\mathbf{ref}_{1}=\mathbf{ref}_{2}=\mathbf{ref}_{3}=(0,0,0,0,0)^\top$, $\mathbf{ref}_{4}=(1,1,0,0,0)^\top$, $\mathbf{ref}_{5}=(1,0,1,0,0)^\top$ and
		$\mathbf{v}_{1}=\standV{1}$, $\mathbf{v}_{2}=\standV{2}$, $\mathbf{v}_{3}=\standV{3}$, $\mathbf{v}_{4}=(1,1,0,1,0)^\top$, $\mathbf{v}_{5}=(1,0,1,0,1)^\top$.
\item
	Consider the setting from above. However, instead of the interaction between \emph{`gender'} and \emph{`smoking'}, one is now interested in the interaction between \emph{`gender'} and \emph{`age'}, a metric variable. We now propose to consider \mbox{$X_{cat}^I=(X^I_1,X^I_2,X^I_3,X^I_4)^\top\in\{0,1\}^4$} with \begin{center}
		$X^I_1\,\hat{=} \,male$, $X^I_2\,\hat{=}\,non$-$binary$, $X^I_3\,\hat{=}\,male\, interacting\,with\,age$, $X^I_4\,\hat{=}\,non$-$binary\, interacting\,with\,age$.\end{center}
	and	$\mathbf{ref}_{1}=\mathbf{ref}_{2}=(0,0,0,0)^\top$, $\mathbf{ref}_{3}=(1,0,0,0)^\top$, $\mathbf{ref}_{4}=(0,1,0,0)^\top$ and $\mathbf{v}_{1}=\standV{1}$, $\mathbf{v}_{2}=\standV{2}$, $\mathbf{v}_{3}=(1,0,1,0)^\top$, $\mathbf{v}_{4}=(0,1,0,1)^\top$. Furthermore, age is written as $X^I_{met}$ and $d_M=2$, since $\regf$ contains the terms $X^I_{met}\cdot X^I_3$ and $X^I_{met}\cdot X^I_4$. The full combined regressor of interest is then given by $\big(X^I_{met},X^I_{cat}\big)$.
\end{enumerate}
\end{ex}

\noindent For the purposes of isolating the effect size of the metric component $X^I_{met}$, we need to define the following new function which `removes' all other entries of $X^I$ from the expectation function \begin{align}\label{ImetFun}
    \widetilde{\regf}(x):\R^{p-d_C}\longrightarrow\R,\quad
    x\longmapsto\regf\big(\realx{I_{cat}}=\mathbf{0},x\big)\,.
\end{align}

\noindent In order to formally define the proposed effect size measure, we now require only the following definition.

\begin{defi*}For $\bar{\RI}$ and $\bar{\mu}$ chosen according to \cref{AverageSlopeDef} and $\widetilde{\regf}$ as defined in \cref{ImetFun}, we define 
\begin{equation}
		\tilde{\eff}_{\emptyset}(\theta,x):=\bigintssss_{\bar{\RI}} \frac{\partial \widetilde{\regf}}{\partial \realx{I_{met}}}(x)\,\,d\bar{\mu}(\realx{\text{MC}})\, ,
	\end{equation}
\begin{equation}
		\tilde{\eff}(\theta,x,\mathbf{v}):=\bigintssss_{\bar{\RI}} \regf (\left . x \right\vert_{x_{\scaleto{[I_{cat}]}{6pt}}=\mathbf{v}})\,\,d\bar{\mu}(\realx{\text{MC}})\, ,
	\end{equation}

\noindent	as well as the following $\forall\, l\in\{d_M+1,...,d_C\}$:

 \begin{equation}
 		\tilde{\eff}_{I_l}(\theta,x):=	\bigintssss_{\bar{\RI}}\dfrac{\partial\regf}{\partial (\realx{I_l}\cdot \realx{I_{met}})}(\left . x \right\vert_{x_{\scaleto{[I_{cat}]}{6pt}}=\mathbf{v}_l})\,\,d\bar{\mu}(\realx{\text{MC}})\, .
\end{equation}
Note that the notation $\partial (\realx{I_l}\cdot \realx{I_{met}})$ represents the derivation of $\regf$ w.r.t. the interaction of the category $\realx{I_l}$, i.e. the realization of the $l$th entry of $X^I_{cat}$, and the metric regressor $\realx{I_{met}}$ and not the actual product. When $\realx{I_l}=1$, we treat this interaction term as an `independent copy' of $\realx{I_{met}}$ which is denoted by $\realx{I_{met}}^{\bullet I_l}$.

\end{defi*}	\,\\ \enlargethispage{.75cm}

\noindent Finally, the following gives the full, formal version of \cref{SuperDef}(iii).
	\begin{Aadefi}\label{appdefi} In cases where metric \emph{and} categorical regressors are involved in an interaction and one is interested in separately quantifying the main and interaction effects, $\measure$ is defined as the function
			\begin{align*}
				\measure:\Theta\longrightarrow\R^{d_I},\quad\theta\longmapsto \begin{pmatrix}
					\measurei{1}(\theta)\\
					\vdots\\
					\measurei{d_{I}}(\theta)
				\end{pmatrix}:= \begin{pmatrix}
				&\int_\RI\tilde{\eff}_{\emptyset}(\theta,x)d\mu(\realx{I_{met}})\\[10pt]
				\int_\RI\tilde{\eff}(\theta,x,\mathbf{v}_{1})d\mu(\realx{I_{met}})\hspace*{-3.4cm}&-&\hspace*{-3.2cm}\int_\RI\tilde{\eff}(\theta,x,\mathbf{ref}_1)d\mu(\realx{I_{met}})\\
					&\vdots&\\
					\int_\RI\tilde{\eff}(\theta,x,\mathbf{v}_{d_M})d\mu(\realx{I_{met}})\hspace*{-3.2cm}&-&\hspace*{-3.2cm}\int_\RI\tilde{\eff}(\theta,x,\mathbf{ref}_{d_M})d\mu(\realx{I_{met}})\\[10pt]
				&	\int_{\RI}\int_{\RI} \tilde{\eff}_{I_{d_M+1}}(\theta,x)d\mu(\realx{I_{met}})\,d\mu(\realx{I_{met}}^{\bullet I_l})
					\\
					&\vdots&\\
					&\int_{\RI}\int_{\RI} \tilde{\eff}_{I_{d_C}}(\theta,x) d\mu(\realx{I_{met}})\,d\mu(\realx{I_{met}}^{\bullet I_l})
				\end{pmatrix}\,,
		\end{align*}
		where $\RI$ and $\mu$ are set as identical in every term, so that $\mu$ satisfies requirements (M1)-(M3) w.r.t. the function in each integral term and $\RI$, after, if applicable, being chosen in accordance with \cref{I1}.
		\end{Aadefi}

\newpage
\section{Proofs\label{Proofs}}
This appendix contains proofs for all statements made in the main paper and \cref{Interactions} as well as some additional remarks throughout. The order of the proofs generally follows the order in which the statements to be proven were made, with the notable exception of \cref{UnifProp} and \cref{UnifCor}. Their proofs are given immediately after the proof of \cref{BetasProp}, since the derived results may be used to prove some statements regarding how existing methodology falls into the proposed framework more efficiently.
\subsection{Proof of \Cref{BetasProp}}

\BetasProp*
 \begin{proof} We prove this statement separately for each setting of definitions \ref{SuperDef}(i)-\ref{SuperDef}(iii). Throughout, w.l.o.g., we only consider cases without interactions between elements of $M$ and $C$ unless they are of interest themselves. We are able to do so as including such interactions does not affect any statements made over the course of this proof.\\
 Furthermore, we consider a parameter vector of the form $\theta=(\beta,v)^\top\in\Theta$ with
\begin{align*}
     \beta=\big(\beta_0,\beta_1,...,\beta_{\tilde{p}}\big)\subseteq\R^{\tilde{p}+1}\,,
 \end{align*}
where, for $p\in\N_{> 0}$ denoting the number of observed regressor variables and $a\in\N_{\geq 0}$ denoting the number of interaction terms included in $\regf$, $\tilde{p}$ is defined by $\tilde{p}:=p+a$.

 \begin{itemize}\item[\underline{Case 1}:] \textit{The regressor of interest is a metric variable or an interaction between metric regressors considered as a separate regressor (if applicable, interactions of higher order are (also) considered separate regressors).}

\end{itemize}     
     In this case, \Cref{SuperDef}(i) applies. \\Specifically,
     for \mbox{$t(j):=1+m+\big(\ind{(j-1)>0}\cdot\sum_{q=1}^{\max(1,j-1)}d_{C_q}\big),$} $j=1,...,c$, $\regf$, if containing no interaction terms, is given by the following
      \begin{align}
        \regf(x)= \regf^{(i)}(\realx{I},\realx{\text{MC}}):=\beta_0+\beta_1\realx{I}+\overset{m}{\underset{i=1}{\sum}}\beta_{i+1}\realx{M_i}+\overset{c}{\underset{j=1}{\sum}}\big(\beta_{t(j)+1},...,\beta_{t(j)+d_{C_j}}\big)\realx{C_j}^\top\,.
     \end{align}
     
     \noindent If $\regf$ does contain interaction terms, we denote the vector of realizations representing higher order interaction terms and/or the interaction whose effect is of interest by $\realx{\text{IA}}$ and 
     \begin{align*}
         \regf(x)=\regf^{(i)^\ast}(\realx{I},\realx{\text{MC}},\realx{\text{IA}}):=\regf^{(i)}(\realx{I},\realx{\text{MC}})+\big(\beta_{p+1},...,\beta_{p+\vert\text{IA}\vert}\big)\realx{\text{IA}}^\top\,.
     \end{align*}
Let $r$ denote the index corresponding to the $\beta$-coefficient of the regressor of interest, i.e. $r=1$ for the case of no interaction terms and $r\in\{1\}\cup\{p+1,...,p+\vert\text{IA}\vert\}$ when one is interested in the effect of a metric regressor included in an interaction or the effect of the interaction term itself. Then the following holds
     \begin{align}
         \frac{\partial\regf^{(i)}}{\partial\realx{I}}(\realx{I},\realx{\text{MC}})=\beta_r\quad\text{ and }\quad\frac{\partial\regf^{(i)^\ast}}{\partial x_r}(\realx{I},\realx{\text{MC}},\realx{\text{IA}})=\beta_r\,.
     \end{align}
Finally, it follows that, if the regressor of interest is a metric variable or an interaction between metric regressors considered as a separate regressor,
\begin{align}
    \measure(\theta)=\begin{cases}
    \int_{\RI}\int_{\bar{\RI}}\beta_r \,d\bar{\mu}(\realx{\text{MC}})d\mu(\realx{I}),&\text{under assumptions (A.I) and (A.II$'$)}\\
     \int_{\RI}\int_{\bar{\RI}}\beta_r \,\aiimeasure(d\realx{\text{MC}},\realx{I})d\mu(\realx{I}),&\text{under assumption (A.II$''$),}
    \end{cases}
\end{align}
   which is equal to $\beta_r$ for any probability measures $\mu$, $\bar{\mu}$, and $\aiimeasure$, whereby the statement holds for \cref{SuperDef}(i). 
\begin{itemize}
     \item[\underline{Case 2.1}:] \textit{The regressor of interest is a categorical variable that is not included in any interactions.}
 \end{itemize}     
In this case, \Cref{SuperDef}(ii) applies and, for \mbox{$b(j):=d_I+m+\big(\ind{(j-1)>0}\cdot\sum_{q=1}^{\max(1,j-1)}d_{C_q}\big)$}, $j\in\{1,...,c\}$, $\regf$ is given by
\begin{align}\label{gii}
        \regf(x)= \regf^{(ii)}(\realx{I},\realx{\text{MC}}):=\beta_0+\big(\beta_1,...,\beta_{d_I}\big)\realx{I}^\top+\overset{m}{\underset{i=1}{\sum}}\beta_{i+d_I}\realx{M_i}+\overset{c}{\underset{j=1}{\sum}}\big(\beta_{b(j)+1},...,\beta_{b(j)+d_{C_j}}\big)\realx{C_j}^\top\,.
\end{align}
Under assumptions (A.I) and (A.II$'$), the following then holds $\forall r\in\{1,...,d_I\}$ 
\begin{align*}
    \measurei{r}(\theta)&=\int_{\bar{\RI}}\regf^{(ii)} (\left . x \right\vert_{\realx{I}=\standV{r}})- \regf^{(ii)} (\left . x \right\vert_{\realx{I}=\mathbf{0}})d\bar{\mu}(\realx{\text{MC}})\\
    &\overset{(\star)}{=}\int_{\bar{\RI}}\regf^{(ii)} (\left . x \right\vert_{\realx{I}=\standV{r}})d\bar{\mu}(\realx{\text{MC}})- \int_{\bar{\RI}}\regf^{(ii)} (\left . x \right\vert_{\realx{I}=\mathbf{0}})d\bar{\mu}(\realx{\text{MC}})\\
    &=\Big(\int_{\bar{\RI}}\beta_r\,d\bar{\mu}(\realx{\text{MC}})+\int_{\bar{\RI}} \beta_0+\overset{m}{\underset{i=1}{\sum}}\beta_{i+d_I}\realx{M_i}+\overset{c}{\underset{j=1}{\sum}}\big(\beta_{b(j)+1},...,\beta_{b(j)+d_{C_j}}\big)\realx{C_j}^\top  d\bar{\mu}(\realx{\text{MC}})\Big)-\\
    &\,\,\quad\Big(\int_{\bar{\RI}}0\,d\bar{\mu}(\realx{\text{MC}})+\int_{\bar{\RI}} \beta_0+\overset{m}{\underset{i=1}{\sum}}\beta_{i+d_I}\realx{M_i}+\overset{c}{\underset{j=1}{\sum}}\big(\beta_{b(j)+1},...,\beta_{b(j)+d_{C_j}}\big)\realx{C_j}^\top  d\bar{\mu}(\realx{\text{MC}})\Big)\\
    &=\int_{\bar{\RI}}\beta_r\,d\bar{\mu}(\realx{\text{MC}})-\int_{\bar{\RI}}0\,d\bar{\mu}(\realx{\text{MC}})=\beta_r-0\\[10pt]
    &=\beta_r\,,
\end{align*} 
whereby the statement holds for \cref{SuperDef}(ii).
\begin{n}
The equality $(\star)$ does not necessarily hold when one replaces $d\bar{\mu}(\realx{\text{MC}})$ with\linebreak $\aiimeasure(d\realx{\text{MC}},\realx{I})$, which is why the statement of this theorem does not apply to for a categorical regressor of interest under (A.II$''$).
 \end{n}

\begin{itemize}
     \item[\underline{Case 2.2}:] \textit{The regressor of interest is the interaction between categorical variables or a categorical regressor involved in (an) interaction(s) with other categorical regressors and one is interested in separating the effects.}
\end{itemize}     
In this case, \Cref{SuperDef}(ii) again applies and $\regf$ given by $\regf^{(ii)}(\realx{I},\realx{\text{MC}})$ as in \cref{gii}. However, $\realx{I}$ represents realizations of a different vector $X^I$ in this case. Specifically, $X^I=(X^I_1,...,X^I_{d_I})$ contains an entry for each category of the categorical regressors involved in the interaction terms and an entry for the combination of categories in each interaction term. \\
As a result, the term $\big(\beta_1,...,\beta_{d_I}\big)\realx{I}^\top$ is no longer always equal to one single $\beta$ but may be given by a sum over several $\beta$s. In order to formally denote such sums, we 
define\linebreak $L(r):=\Big\{i\in\{1,...,d_I\}\backslash\{r\}\big\vert \text{$X^I_i$ is included in $X^I_r$}\Big\}$, $r=1,...,d_I$.\\

\noindent Under assumptions (A.I) and (A.II$'$), the following then holds  $\forall r\in\{1,...,d_I\}$ 
\begin{align*}
    \measurei{r}(\theta)&=\int_{\bar{\RI}}\regf^{(ii)} (\left . x \right\vert_{\realx{I}=\mathbf{v}_{r}})- \regf^{(ii)} (\left . x \right\vert_{\realx{I}=\mathbf{0}})d\bar{\mu}(\realx{\text{MC}})-
    \int_{\bar{\RI}}\regf^{(ii)} (\left . x \right\vert_{\realx{I}=\mathbf{ref}_{r}})- \regf^{(ii)} (\left . x \right\vert_{\realx{I}=\mathbf{0}})d\bar{\mu}(\realx{\text{MC}})\\
    &\hspace*{-28pt}\overset{(\star\star)}{=}\int_{\bar{\RI}}\regf^{(ii)} (\left . x \right\vert_{\realx{I}=\mathbf{v}_{r}})d\bar{\mu}(\realx{\text{MC}})- \int_{\bar{\RI}}\regf^{(ii)} (\left . x \right\vert_{\realx{I}=\mathbf{ref}_r})d\bar{\mu}(\realx{\text{MC}})\\
    &\hspace*{-25pt}=\Big(\int_{\bar{\RI}}\sum_{l\in L(r)\cup \{r\}} \beta_l\,d\bar{\mu}(\realx{\text{MC}})+\int_{\bar{\RI}} \beta_0+\overset{m}{\underset{i=1}{\sum}}\beta_{i+d_I}\realx{M_i}+\overset{c}{\underset{j=1}{\sum}}\big(\beta_{b(j)+1},...,\beta_{b(j)+d_{C_j}}\big)\realx{C_j}^\top  d\bar{\mu}(\realx{\text{MC}})\Big)-\\
    &\,\,\quad\Big(\int_{\bar{\RI}}\sum_{l\in L(r)} \beta_l\,d\bar{\mu}(\realx{\text{MC}})+\int_{\bar{\RI}} \beta_0+\overset{m}{\underset{i=1}{\sum}}\beta_{i+d_I}\realx{M_i}+\overset{c}{\underset{j=1}{\sum}}\big(\beta_{b(j)+1},...,\beta_{b(j)+d_{C_j}}\big)\realx{C_j}^\top  d\bar{\mu}(\realx{\text{MC}})\Big)\\
    &\hspace*{-25pt}=\int_{\bar{\RI}}\sum_{l\in L(r)\cup \{r\}} \beta_l\,d\bar{\mu}(\realx{\text{MC}})-\int_{\bar{\RI}}\sum_{l\in L(r)} \beta_l\,d\bar{\mu}(\realx{\text{MC}})=\sum_{l\in L(r)\cup \{r\}} \beta_l-\sum_{l\in L(r)} \beta_l\\[10pt]
    &\hspace*{-25pt}=\beta_r\,,
\end{align*} 
whereby the statement holds for \cref{SuperDef}(iii).\\
Please note that, just as in case 2, this does not apply under (A.II$''$) by the same reasoning as for equality $(\star)$ applied to equality $(\star\star)$.\\

\begin{itemize}
     \item[\underline{Case 3}:] \textit{The regressor of interest is the interaction between categorical and metric variables or a regressor involved in (an) interaction(s) with other categorical and/or metric regressor(s) and one is interested in separating the effects.}
\end{itemize}    
In this case, \Cref{SuperDef}(iii) as specified in \cref{DivInt} applies and, with $b(j)$ defined as in case 2, $\regf$ is given by\\
\begin{equation}
\begin{aligned}
    \regf(x)= \regf^{(iii)}(\realx{I},\realx{\text{MC}}):=\beta_0+&\beta_1\realx{I_{met}}+\big(\beta_{2},...,\beta_{d_M+1},...,\beta_{d_I}\big)\big(\realx{I_{cat}}\odot \mathbf{m}\big)^\top+\\
        &\overset{m}{\underset{i=1}{\sum}}\beta_{i+d_I}\realx{M_i}+
        \overset{c}{\underset{j=1}{\sum}}\big(\beta_{b(j)+1},...,\beta_{b(j)+d_{C_j}}\big)\realx{C_j}^\top\,,
\end{aligned}    
\end{equation}
with $\realx{I_{cat}}\odot \mathbf{m}$ denoting the elementwise product of $\realx{I_{cat}}$ and the vector \linebreak $\mathbf{m}:=(\underset{\mathbf{m}_1,...,\mathbf{m}_{d_M}}{\underbrace{1,...,1}},\underset{\mathbf{m}_{d_M+1},...,\mathbf{m}_{d_I-1}}{\underbrace{\realx{I_{met}},...,\realx{I_{met}}}})$.\\\,\\

\noindent It then immediately follows that

\begin{equation}\label{c41}
\begin{aligned}
\measurei{1}(\theta)&=\bigintssss_{\RI}\bigintssss_{\bar{\RI}} \frac{\partial}{\partial\realx{I_{met}}}\bigg(\beta_0+\beta_1\realx{I_{met}}+
        \overset{m}{\underset{i=1}{\sum}}\beta_{i+d_I}\realx{M_i}+\\
        &\quad \hspace{4cm}\overset{c}{\underset{j=1}{\sum}}\big(\beta_{b(j)+1},...,\beta_{b(j)+d_{C_j}}\big)\realx{C_j}^\top\bigg)d\bar{\mu}(\realx{\text{MC}})d\mu(\realx{I_{met}})\\
&=\int_{\RI}\int_{\bar{\RI}}\beta_1 \,d\bar{\mu}(\realx{\text{MC}})d\mu(\realx{I_{met}})=\beta_1\, .
\end{aligned}
\end{equation}
Next, we redefine $L(r)$ as $L(r):=\Big\{i\in\{3,...,d_I\}\backslash\{r\}\big\vert \text{$X^I_{i-1}$ is included in $X^I_{r-1}$}\Big\}$, $r=2,...,d_I$, since, for any $r\in\{1,...,d_C\}$, $X^I_r$ is the $r$th entry of $X^I_{cat}$, but the $(r+1)$th entry of $X^I$.\\ \pagebreak

\noindent Then, the following holds $\forall r\in\{2,...,d_{M}+1\}$
\begin{equation*}
\begin{aligned}	
&\measurei{r}(\theta)= \int_\RI\int_{\bar{\RI}} \beta_0+\beta_1\realx{I_{met}}+\sum_{l\in L(r)\cup \{r\}} \beta_l+
        \overset{m}{\underset{i=1}{\sum}}\beta_{i+d_I}\realx{M_i}+\\
        &\quad \hspace{2cm}
        \overset{c}{\underset{j=1}{\sum}}\big(\beta_{b(j)+1},...,\beta_{b(j)+d_{C_j}}\big)\realx{C_j}^\top d\bar{\mu}(\realx{\text{MC}})d\mu(\realx{I_{met}})-
\int_\RI  \int_{\bar{\RI}} \beta_0+\beta_1\realx{I_{met}}+\\
        &\quad \hspace{2cm}\sum_{l\in L(r)} \beta_l+
        \overset{m}{\underset{i=1}{\sum}}\beta_{i+d_I}\realx{M_i}+\overset{c}{\underset{j=1}{\sum}}\big(\beta_{b(j)+1},...,\beta_{b(j)+d_{C_j}}\big)\realx{C_j}^\top d\bar{\mu}(\realx{\text{MC}})  d\mu(\realx{I_{met}})\\
&=\int_\RI\int_{\bar{\RI}} \beta_0+\beta_1\realx{I_{met}}+\sum_{l\in L(r)\cup \{r\}} \beta_l+
        \overset{m}{\underset{i=1}{\sum}}\beta_{i+d_I}\realx{M_i}+
        \overset{c}{\underset{j=1}{\sum}}\big(\beta_{b(j)+1},...,\beta_{b(j)+d_{C_j}}\big)\realx{C_j}^\top d\bar{\mu}(\realx{\text{MC}})-\\
&\,\,\quad  \int_{\bar{\RI}} \beta_0+\beta_1\realx{I_{met}}+\sum_{l\in L(r)} \beta_l+
        \overset{m}{\underset{i=1}{\sum}}\beta_{i+d_I}\realx{M_i}+
        \overset{c}{\underset{j=1}{\sum}}\big(\beta_{b(j)+1},...,\beta_{b(j)+d_{C_j}}\big)\realx{C_j}^\top d\bar{\mu}(\realx{\text{MC}})  d\mu(\realx{I_{met}})\,.
\end{aligned}
\end{equation*}
Combining the reasoning from case 1 and case 2.2, it follows that
\begin{equation}\label{c42}
\begin{aligned}	
\measurei{r}(\theta)&= \int_\RI\beta_r\,\,  d\mu(\realx{I_{met}})
=\beta_r\,.
\end{aligned}
\end{equation}
\,

\noindent Lastly, $\forall r\in\{d_{M}+2,...,d_I\}$
\begin{align}\label{c43}
         \frac{\partial\regf^{(iii)}}{\partial(\realx{I_{r-1}}\cdot \realx{I_{met}})}(\realx{I},\realx{\text{MC}})=\beta_r\,,
\end{align}
whereby $\measurei{r}(\theta)=\beta_r$ under all assumptions by the same reasoning as was used in case 1.\\

Finally, combining \cref{c41}, \cref{c42}, and \cref{c43} the statement also holds for \cref{SuperDef}(iii).\\
 \end{proof}

\subsection{Proofs of \Cref{UnifProp} and \Cref{UnifCor} plus \Cref{unfortunately not}}

\UnifProp* 
\begin{proof}
Under (A.I) and (A.II$'$), the regressor of interest is treated as an independent variable. Therefore, given that $\int_{\RI\times\bar{\RI}}\vert\regf(x)\vert d(\mu\times\bar{\mu})<\infty$ is required under all assumptions, the following holds by Fubini's theorem (see \cite{Fubini})
\begin{equation}\label{UPprood_ind}
    \begin{aligned}
        \measure(\theta)&=\int_\RI \bar{\eff}(\theta,x)d\mu(x)=\int_\RI\int_{\bar{\RI}}\eff(\theta,x)d\bar{\mu}(\realx{\text{MC}})d\mu(\realx{I})\\
        &=\int_{\bar{\RI}}\int_\RI\eff(\theta,x)d\mu(\realx{I})d\bar{\mu}(\realx{\text{MC}})\,.
    \end{aligned}
\end{equation}
For $\RI=(a_\RI,b_\RI)$ and $\mu$ being chosen or given as the measure associated with the $U(a_\RI,b_\RI)$ distribution it immediately follows that\begin{equation*}
    \begin{aligned}
\measure(\theta)&=\int_{\bar{\RI}}\big(b_\RI-a_\RI\big)^{-1}\cdot\Big\lbrack\regf(x)\Big\rbrack^{\realx{I}=b_\RI}_{\realx{I}=a_\RI} d\bar{\mu}(\realx{\text{\textup{MC}}})\\
            &=\big(b_\RI-a_\RI\big)^{-1}\cdot\Bigg(\int_{\bar{\RI}}\regf (\left . x \right\vert_{\realx{I}=b_\RI})
    -
    \regf (\left . x \right\vert_{\realx{I}=a_\RI})d\bar{\mu}(\realx{\text{\textup{MC}}})\Bigg)\, .
    \end{aligned}
\end{equation*}
\end{proof}
\begin{rems}\label{unfortunately not}
Under assumption (A.II$''$), the regressor of interest is not considered to be independent, and therefore \cref{UPprood_ind} does not hold. In fact, as mentioned in \cref{UnifSec}, this is why similar results to \cref{UnifProp} and \cref{UnifCor} may not be obtained under assumption (A.II$''$). The following proves this statement.\\
Let $f_X$ denote the joint density by which $\mu_X$ is defined. Furthermore, let $f_{I}$ denote the marginal density of the regressor of interest $X^I$ and $f_{X_{[\text{MC}]}\vert X^I}$ the conditional density of all other regressors given $X^I$. By the definition of marginal density, the following holds \begin{equation}\label{unf1}
    f_{X_{[\text{MC}]}\vert X^I}\big(\realx{\text{MC}},\realx{I}\big)=\frac{f_X\big(\realx{I},\realx{\text{MC}}\big)}{f_I(\realx{I})}\cdot\ind{f_I(\realx{I})>0}\,.
\end{equation}
Now, if $\mu$ is given as the measure associated with the $U(a_\RI,b_\RI) $ distribution, it follows that \begin{equation}\label{unf2}
  f_I(\realx{I})=\big(b_\RI-a_\RI\big)^{-1}\cdot\ind{[a_\RI,b_\RI]}(\realx{I}) \,,
\end{equation}
and, therefore, for $\lambda$ denoting the Lebesgue measure,
\begin{align*}
 \measure(\theta)&= \big(b_\RI-a_\RI\big)^{-1}\cdot\int_{a_\RI}^{b_\RI}\int_{\bar{\RI}}\frac{\partial\regf}{\partial \realx{I}} (x)\aiimeasure(d\realx{\text{\textup{MC}}}, \realx{I})d\lambda(\realx{I})\\
 &=\big(b_\RI-a_\RI\big)^{-1}\cdot\int_{a_\RI}^{b_\RI}\int_{\bar{\RI}}\frac{\partial\regf}{\partial \realx{I}} (x)f_{X_{[\text{MC}]}\vert X^I}\big(\realx{\text{MC}},\realx{I}\big)d\lambda\realx{\text{\textup{MC}}})d\lambda(\realx{I})\\
  &=\big(b_\RI-a_\RI\big)^{-1}\cdot\int_{a_\RI}^{b_\RI}\int_{\bar{\RI}}\frac{\partial\regf}{\partial \realx{I}} (x)f_X(x)\cdot\big(b_\RI-a_\RI\big)\cdot\ind{[a_\RI,b_\RI]}(\realx{I}) d\lambda(\realx{\text{\textup{MC}}})d\lambda(\realx{I})\,,
\end{align*}
where the last equality follows from combining \cref{unf1} and \cref{unf2}. \\Applying integration by parts, it now follows that
\begin{align*}
 \hspace*{-.75cm}\measure&(\theta)= \bigintssss_{\bar{\RI}}\bigg(
\Big\lbrack \regf(x)f_X(x) \Big\rbrack^{\realx{I}=b_\RI}_{\realx{I}=a_\RI}
  -  \int_{a_\RI}^{b_\RI}\regf(x)\frac{\partial f_X}{\partial \realx{I}}(x)\cdot\ind{[a_\RI,b_\RI]}(\realx{I}) d\lambda(\realx{I})\bigg)d\lambda(\realx{\text{\textup{MC}}}) \\
  &=\underset{(\ast)}{\underbrace{\big(b_\RI-a_\RI\big)^{-1}\Big(\int_{\bar{\RI}}\regf (\left . x \right\vert_{\realx{I}=b_\RI})\aiimeasure(d\realx{\text{\textup{MC}}},b_\RI)-
  \int_{\bar{\RI}}\regf (\left . x \right\vert_{\realx{I}=a_\RI})\aiimeasure(d\realx{\text{\textup{MC}}},a_\RI)\Big)}}-\\
  &\quad\quad \underset{(\ast\ast)}{\underbrace{\int_{\bar{\RI}}\int_{a_\RI}^{b_\RI}\regf(x)\frac{\partial f_X}{\partial \realx{I}}(x)\cdot\ind{[a_\RI,b_\RI]}(\realx{I}) d\lambda(\realx{I})d\lambda(\realx{\text{\textup{MC}}})}}\,.
\end{align*}
The statement "$\measure(\theta)=(\ast)$" would constitute a similar result to \cref{UnifProp} and \cref{UnifCor} under assumption (A.II$''$), however, term $(\ast\ast)$ can not w.l.o.g. be assumed to equal zero.\\
\end{rems}

\UnifCor*

\begin{proof}
These results follow immediately from basic transformations.

\end{proof}

\subsection{Proofs of \Cref{GelmanProp} and \Cref{KaufmanProp} plus \Cref{GelmanPerdoeRem}}

\GelmanProp*
\begin{proof} First, note that \cite{GelmanHill2007} define the average predictive difference \emph{per unit}, which is an interval denoted by $\big[u^{(lo)},u^{(hi)}\big]\in\mathcal{B}(\R)$. 
Combining equations (21.3) and (21.4) from \cite{GelmanHill2007}, the average predictive difference $B_u$ is given by \begin{equation}
 B_u\big(u^{(lo)},u^{(hi)}\big)=\frac{1}{n}\sum_{i=1}^n\frac{\EW\big[y\vert u^{(hi)} v_i,\theta\big]-\EW\big[y\vert u^{(lo)} v_i,\theta\big]}{u^{(hi)}-u^{(lo)}}\,,
\end{equation}
with $v_i$ denoting the $i$th observation of all regressors in $M\cup C$.\\
Let the sequence $\mathcal{D}_X^{\textup{[MC]}}$ containing values in $\R^{p-d_I}$ denote the sequence of observation-vectors contained in $\mathcal{D}_X$, but with the elements mapped to $I$ removed. 
For the setting in this proposition, with $\mu_X$ given as in \cref{EDdef} and $\bar{\mu}$ the marginal distribution obtained by marginalizing out the regressor of interest and $\bar{\RI}$ given by the set $\{\mathbf{x}^{\scaleto{\text{[MC]}}{5pt}}_i\vert \mathbf{x}^{\scaleto{\text{[MC]}}{5pt}}_i\in \mathcal{D}_X^{\textup{[MC]}}\}$, it now immediately follows that
\begin{align*}
  B_u\big(u^{(lo)},u^{(hi)}\big)&=\Big(u^{(hi)}-u^{(lo)}\Big)^{-1}\cdot\frac{1}{n}\sum_{i=1}^n\Big(\regf(\realx{I}=u^{(hi)},\mathbf{x}^{\scaleto{\text{[MC]}}{5pt}}_i)-\regf(\realx{I}=u^{(lo)},\mathbf{x}^{\scaleto{\text{[MC]}}{5pt}}_i)\Big)\\
  &=\Big(u^{(hi)}-u^{(lo)}\Big)^{-1}\cdot\Bigg(\int_{\bar{\RI}} \regf (\left . x \right\vert_{\realx{I}=u^{(hi)}})
    -
    \regf (\left . x \right\vert_{\realx{I}=u^{(lo)}})d\bar{\mu}(\realx{\text{\textup{MC}}})\Bigg)\\
  &\overset{\text{prop.\ref{UnifProp}}}{=}\measure(\theta)
\end{align*}
under assumption (A.II$'$), proving the statement.
\begin{n}\if1\Preprint{\enlargethispage{11pt}}\fi
Beneath the definition of average predictive difference, \cite{GelmanHill2007} add "\emph{In practice, one must also average over $\theta$ (or plug in a point estimate)}". `Averaging over $\theta$' in the sense of averaging over draws from either the posterior distribution, resulting
from Bayesian inference, or the $N(\hat{\theta},\Sigma_{\hat{\theta}})$ distribution, with frequentist point estimate $\hat{\theta}$ and covariance matrix $\Sigma_{\hat{\theta}}$, would constitute the point estimate we defined in \cref{UncertaintyDef}.
\end{n}
\end{proof}
\begin{rems}\label{GelmanPerdoeRem} The extensions of $B_u$ given above that are proposed in
\cite{Gelman_AveragePredictiveComparison} may also be incorporated into the current framework, as the following details for each identified extension. Throughout this remark, we only consider generalized marginal effects $\measure$ \emph{without} the option of separately quantifying main and interaction effects, as this approach corresponds to \cite{Gelman_AveragePredictiveComparison}, who propose the following for dealing with interaction terms: "When defining the predictive comparison [...], we must alter [an] input wherever it occurs in the model".
\begin{enumerate}[label=\ref{GelmanPerdoeRem}.\alph*)]
    \item For the predictive comparison $\delta_u\big(u^{(1)}\longrightarrow u^{(2)},v,\theta\big)$ defined by \cite[eq.(1)]{Gelman_AveragePredictiveComparison}, the following holds within the current framework.  \begin{enumerate}[label=\ref{GelmanPerdoeRem}.a.\roman*)]
    \item For a metric regressor of interest, let $\measure$ denote the generalized marginal effect under assumption (A.II$\,'$), with $\mu$ chosen as the measure associated with the $U\big(u^{(1)},u^{(2)}\big)$ distribution, for $u^{(1)},u^{(2)}\in\R$, and $\bar{\mu}=\delta_{\{v\}}$ for some fixed \mbox{$v\in\R^{p-d_I}$,} then \begin{equation}\label{apcN}
        \delta_u\big(u^{(1)}\longrightarrow u^{(2)},v,\theta\big)\,\,\hat{=}\,\,\measure(\theta)\,.
    \end{equation}
    \item For a categorical regressor of interest, let $u^{(1)},u^{(2)}$ denote two different categories, including the reference-category, of $X^I$. Furthermore, let $\measure$ again denote the generalized marginal effect under assumption (A.II$\,'$), with $\bar{\mu}=\delta_{\{v\}}$ for some fixed \mbox{$v\in\R^{p-d_I}$,} then \begin{equation}\label{apcC}
        \hspace*{-1.5cm}\delta_u\big(u^{(1)}\longrightarrow u^{(2)},v,\theta\big)\,\,\hat{=}\begin{cases}\measurei{j}(\theta)\,,&\begin{array}{l}\text{for $u^{(2)}$ denoting the $j$th- and}\\
        \text{$u^{(1)}$ the reference-category}\end{array}\\[15pt]
        \measurei{j}(\theta)-\measurei{l}(\theta)\,,&\begin{array}{l}\text{for $u^{(2)}$ denoting the $j$th- and}\\
        \text{$u^{(1)}$ the $l$th- category}\end{array}\\[15pt]
        -\measurei{j}(\theta)\,,&\begin{array}{l}\text{for $u^{(2)}$ denoting the reference- }\\
        \text{and $u^{(1)}$ the $j$th- category,}\end{array}
        \end{cases}
    \end{equation}
    given that $(u^{(1)}-u^{(2)})$ is set to $1$ for a categorical regressor of interest by \cite{Gelman_AveragePredictiveComparison}. 
    \end{enumerate}
\end{enumerate}
    The average predictive comparisons for \emph{Numerical Inputs} and \emph{Unordered Categorical Inputs}, given by \cite[eqs.(2),(5)]{Gelman_AveragePredictiveComparison} and \cite[eqs.(3),(4)]{Gelman_AveragePredictiveComparison}, respectively, is simply translated into this work's framework by averaging over the quantities from equations \ref{apcN} and \cref{apcC}, respectively, equivalently to how they propose averaging over $\delta_u\big(u^{(1)}\longrightarrow u^{(2)},v,\theta\big)$.

Having transferred the above quantities from \cite{Gelman_AveragePredictiveComparison} into our framework, it immediately follows that our framework also applies to \emph{Variance Components Models}, \emph{Inputs That Are Not Always Active}, and \emph{Nonmonotonic Functions} at the very least by proceeding exactly as suggested in sections 3.4., 3.6., and 3.7., respectively, of \cite{Gelman_AveragePredictiveComparison}.

\end{rems}\,

\noindent In order to prove \cref{KaufmanProp}, we now first require the following lemma.
\begin{lemma}\label{KaufmanLemma}
Let $h:\R^d\longrightarrow\R$, $d\in\N_{>0}$, be a continuous and bounded function. For any probability measure $\mu$ on $\Big(\R^d,\mathcal{B}\big(\R^d\big)\Big)$ the following holds \begin{equation}\label{LemEq}
    \exists\, a\in\R^d:\quad \int_{\R^d}h(x)d\mu(x)=h(a)\,.
\end{equation}
\end{lemma}
\begin{proof}
Let $b_l,b_h\in\R$ denote the boundaries of $h$, i.e. $b_l\leq h(x)\leq b_h$ $\forall x\in\R^d$. It then follows by the \emph{intermediate value theorem}, if $d=1$, and the \emph{main theorem of connectedness}, if $d\geq 2$, that the image of $h$ is exactly equal to the closed interval between its boundaries, i.e. $\mathcal{I}(h)=[b_l,b_h]$. Furthermore, since it clearly holds that\begin{align*}
 b_l=\int_{\R^d}b_l\,d\mu(x)  \leq\int_{\R^d}h(x)d\mu(x)\leq \int_{\R^d}b_h\,d\mu(x)=b_h\,,  
\end{align*} we also have that the value of the integral $\int_{\R^d}h(x)d\mu(x)$ lies in the interval $[b_l,b_h]$.

Combining these two statements, it immediately follows that there exists some $a\in\R^d$ so that \if1\JASA{\linebreak}\fi $\int_{\R^d}h(x)d\mu(x)=h(a)$.
\end{proof}

\KaufmanProp*
\begin{proof}
 In order to conform with the notation of \cite{Kaufman1996}, let the metric regressor of interest be the $j$th regressor, i.e. $\realx{I}=x_j$, $j\in\{1,...,p\}$, and let $b_j$ denote the corresponding coefficient in the linear predictor $\eta$.\\ Equation (5) of \cite{Kaufman1996} then defines the \emph{instantaneous slope} as $\partial P_j:=\partial P/\partial X_j$, with $P$ denoting the equivalent of $\regf$ in this setting. This is clearly equal to $s$ as defined in \cref{slope1}, however, equation (5) of \cite{Kaufman1996} also specifies that $\partial P_j=b_jP(1-P)$, limiting the definition to cases in which the regressor of interest is not involved in any interactions.\\
 Furthermore, equation (6) of  \cite{Kaufman1996} defines the \emph{predicted change in probability} as \begin{equation}
     \Delta P_j= \Bigg(1+\exp\bigg(-\Big(\log\bigg(\frac{P_{ref}}{1-P_{ref}}\bigg)+\frac{1}{2}b_j\Big)\bigg)\Bigg)^{-1}-
     \Bigg(1+\exp\bigg(-\Big(\log\bigg(\frac{P_{ref}}{1-P_{ref}}\bigg)-\frac{1}{2}b_j\Big)\bigg)\Bigg)^{-1}
     \,\,,
 \end{equation}

\noindent where $P_{ref}$ is not definitively defined. \cite{Kaufman1996} does suggest the quantities $P_{ref}=$\linebreak $n^{-1}\sum_{\mathbf{x}_i\in\mathcal{D}_X}\regf(\mathbf{x}_i)$, in the notation of \cref{EDdef}, and $P_{ref}=\regf(\bar{x})$, where $\bar{x}$ denotes the observed sample mean, as reasonable choices, however, just as in \cite{Williams2012}, the latter is not clearly defined for cases where $C$ is not an empty set (see \cref{existingMEsec}). 

Note that, considering that $\regf$ is continuous and bounded in this proposition's setting, it follows by \cref{KaufmanLemma} that for any $\bar{\mu}$ chosen according to assumptions (A.I) or (A.II$'$)
\begin{align}\label{LemmaUse1}
    \exists\, a\in\R^{p}\,:\quad \regf(a)=\int_{\bar{\RI}\times\RI}\regf(x)d(\mu\times\bar{\mu})
\end{align}
as well as, for any fixed $\realx{I},k\in\R$
\begin{align}\label{LemmaUse2}
    \exists\, a\in\R^{p-1}\,: \regf\big(\realx{I}+k,a\big)-\regf\big(\realx{I},a\big)=\int_{\bar{\RI}}\regf\big(\realx{I}+k,\realx{\text{MC}}\big)-\regf\big(\realx{I},\realx{\text{MC}}\big)\,d\bar{\mu}(\realx{\text{MC}})\,.
\end{align}
Given \cref{LemmaUse1}, we will extend the `reasonable choices' suggested by \cite{Kaufman1996} to considering, for any choice of $x^{ref}\in\R^p$, $P_{ref}=\regf(x^{ref})$ and show that, under assumptions (A.I) and (A.II$'$), $\measure$ may be specified to correspond to $\Delta P$ for any $x^{ref}\in\R^p$.

\noindent First, we denote by $x^{ref}_{[I]}$ and $x^{ref}_{[\text{MC}]}$, the entries of $x^{ref}$ representing $\realx{I}$ and $\realx{\text{MC}}$, respectively, and define $\realx{I}^{(lo)}:=\realx{I}^{ref}-0.5$ and $\realx{I}^{(hi)}:=\realx{I}^{ref}+0.5$. Since $\eta$ contains no interactions involving $\realx{I}$, it holds that
\begin{equation}\begin{aligned}\label{KDelta1}
    \regf(\realx{I}^{(hi)},\realx{\text{MC}}^{ref})&=\Bigg(1+\exp\bigg(-\Big(\eta(\realx{I}^{ref},\realx{\text{MC}}^{ref})+\frac{1}{2}b_j\Big)\bigg)\Bigg)^{-1}\\
    &=\Bigg(1+\exp\bigg(-\Big(\log\Big(\exp\big(\eta(x^{ref})\big)\Big)+\frac{1}{2}b_j\Big)\bigg)\Bigg)^{-1}\\
    &=\Bigg(1+\exp\bigg(-\Big(\log\bigg(\frac{\regf(x^{ref})}{1-\regf(x^{ref})}\bigg)+\frac{1}{2}b_j\Big)\bigg)\Bigg)^{-1}
\end{aligned}\end{equation}
and, equivalently, 
\begin{equation}\begin{aligned}\label{KDelta2}
   \regf(\realx{I}^{(lo)},\realx{\text{MC}}^{ref})=\Bigg(1+\exp\bigg(-\Big(\log\bigg(\frac{\regf(x^{ref})}{1-\regf(x^{ref})}\bigg)-\frac{1}{2}b_j\Big)\bigg)\Bigg)^{-1}\,.
\end{aligned}\end{equation}
Combining \cref{KDelta1}, \cref{KDelta2}, and \cref{LemmaUse2} gives the following for the proper choice of a probability measure $\bar{\mu}$ on $\Big(\R^{p-1},\mathcal{B}\big(\R^{p-1}\big)\Big)$ corresponding to $x^{ref}$
\begin{align}\label{KDeltaFinal}
\Delta P_j=\int_{\supp(\bar{\mu})}\regf(\realx{I}^{(hi)},\realx{\text{MC}})-\regf(\realx{I}^{(lo)},\realx{\text{MC}})\,d\bar{\mu}(\realx{\text{MC}})\,.
\end{align}
By \cref{UnifProp}, since $\big\vert \realx{I}^{(hi)}-\realx{I}^{(lo)}\big\vert=1$, \cref{KDeltaFinal} clearly corresponds to $\measure$ under either assumption (A.I) or assumption (A.II$'$) with $\mu$ chosen as the measure associated with the $U(\realx{I}^{ref}-1/2,\realx{I}^{ref}+1/2)$ distribution.  

\noindent Therefore, it follows that for any choice of $P_{ref}$, there exists some $u\in\R$ and choice of $\bar{\mu}$ under both assumptions (A.I) and (A.II$'$) so that $\measure$ to corresponds to $\Delta P$ for $\mu$ chosen as the measure associated with the $U(u-1/2,u+1/2)$ distribution, proving the statement.
\end{proof}

\subsection{Proof of \Cref{MLProp} plus \Cref{ScholbeckProp}}

\MLProp*
\begin{proof}\,
 Firstly, note that in the notation of \cite{ALEml} $\mathcal{S}\hat{=}\supp(\mu_X)$, with $\mu_X$ denoting the joint distribution of all regressors. Additionally, for the case of the regressor of interest, denoted by $X_j$ in the setting of this proposition, being metric, $f^j(\boldsymbol{\cdot})\hat{=}\eff(\hat{\theta},\boldsymbol{\cdot})$. Therefore, it immediately follows that requirements (a) and (b) from \cite[theorem 1]{ALEml} are met in the setting of this proposition. Furthermore,  as  $f^j(\boldsymbol{\cdot})=\eff(\hat{\theta},\boldsymbol{\cdot})$, it evidently holds that $\mathbb{E}[f^j(X_j,\mathbf{X}_{\backslash j})\vert X_j]\hat{=}\bar{\eff}(\hat{\theta},x)$ under assumption (A.II$''$) in the current setting. Since, $\forall \realx{I}\in\supp(\mu_j)$, $\eff(\hat{\theta},x)$ is continuous in $\realx{I}$ on $\supp(\mu_j)$ as a finite integral over a continuous function , requirement (c) from \cite[theorem 1]{ALEml} is also met. It then follows for any $z\in\supp(\mu_j)$ that \begin{align*}
        g_{j,ALE}(z)&=\int_{\min\{\sup(\mu_j)\}}^z \bar{\eff}(\hat{\theta},x)\,\,d x_j\\
        &=\underset{=\measure(\hat{\theta}) \text{ for the specifications made in the setting of this proposition}}{\underbrace{\frac{1}{\text{\small $\big(z-\min\{\supp(\mu_{j})\}\big)$}}\int_{\min\{\sup(\mu_j)\}}^z \bar{\eff}(\hat{\theta},x)\,\,d \realx{I}\,\,}}
        \cdot\big(z-\min\{\supp(\mu_{j})\}\big)\\[5pt]
        &=\measure(\hat{\theta})\cdot\big(z-\min\{\supp(\mu_{j})\}\big)\,,
    \end{align*}
proving the statement.
\end{proof}

\begin{rems}\label{ScholbeckProp}
 
Within the parametric setting of \cref{setting}, all versions of (forward) marginal effects defined in equations (2)-(5) of \cite{BischlHeumann} may be expressed in terms of individualized expectation by writing
   \begin{equation}\label{MLProp_fME}
        e\Big(\hat{\theta}\big\vert\{\mathbf{a}^{(1)}\},\{\mathbf{b}^{(1)}\},\delta_{\{\mathbf{a}^{(1)}\}},\delta_{\{\mathbf{b}^{(1)}\}}\Big)-e\Big(\hat{\theta}\big\vert \{\mathbf{a}^{(2)}\},\{\mathbf{b}^{(2)}\},\delta_{\{\mathbf{a}^{(2)}\}},\delta_{\{\mathbf{b}^{(2)}\}}\Big)\,,
    \end{equation}
    with $\mathbf{a}^{(1)},\mathbf{a}^{(2)}\in\R^{d_I}$ and $\mathbf{b}^{(1)},\mathbf{b}^{(2)}\in\R^{p-d_I}$. The choice of  $\mathbf{a}^{(1)},$ $\mathbf{a}^{(2)},$ $\mathbf{b}^{(1)},$ and $\mathbf{b}^{(2)}$ determines which of the definitions from \cite{BischlHeumann} \cref{MLProp_fME} is equal to. More specifically, all those equations have the form $\widehat{f}(x^\ast)-\widehat{f}(x)$, with $\widehat{f}$ in the notation of \cite{BischlHeumann} corresponding to $\regF{\hat{\theta}}$
 within the setting of \cref{setting}. Within the parametric setting of \cref{setting}, therefore,
\cref{MLProp_fME} follows immediately from the fact that for any $\mathbf{a}\in\R^{d_I}$ and $\mathbf{b}\in\R^{p-d_I}$
 \begin{align}\label{ScholbeckproofBasic}
      e\big(\hat{\theta}\big\vert\mathbf{a},\mathbf{b},\delta_{\{\mathbf{a}\}},\delta_{\{\mathbf{b}\}}\big)=\regF{\hat{\theta}}\big(\realx{I}\text{=}\mathbf{a},\realx{\text{MC}}\text{=}\mathbf{b}\big)\,.
\end{align}
Note that, for equations (2), (3), and (4) from \cite{BischlHeumann}, only $\mathbf{a}$ is different between the two terms being subtracted from each other, while for \emph{Forward Differences with Multivariate Feature Value Changes} defined in Scholbeck et al.'s equation (5) both $\mathbf{a}$ and those entries of $\mathbf{b}$ that represent $s-1$ regressors which are not directly of interest differ.
Furthermore, within the parametric setting of \cref{setting}, the forward marginal effect (fAMEs) defined by \cite{BischlHeumann} may be related to $\measure$ defined in the current work as follows.

Let $\mathcal{D}_X=\{\mathbf{x}_i\}_{i=1}^n$ denote the sequence on which the effect size is to be evaluated, not necessarily the original observations, and $\mathcal{D}_{\textup{[I]}}$ the set of unique values $\realx{I}$ takes in $\R$ or $\standB{d_I}$ within $\mathcal{D}_X$.  Furthermore, $\forall i\in\{1,...,\vert\mathcal{D}_{\textup{[I]}}\vert\},$ $h\in\R$, let $n_{\mathcal{D}_{\textup{[$I_i$]}}}$ denote the number of observations in $\mathcal{D}_X$ for which the entry that is mapped to $I$ takes a value that equals $\realx{I}^{(i)}$, the $i$th unique observation of the regressor of interest; and 
    let $\tilde{\Delta}_s^{i,h}$ denote the proposed effect size measure under assumption (A.II$\,'$), with $\mu_X$ representing the joint empirical distribution on the subsequence of $\mathcal{D}_X$ containing those elements where the observation-vector entry mapped to $I$ is equal to $\realx{I}^{(i)}$ and $\mu$ is chosen as the measure associated with the $U(\realx{I}^{(i)},\realx{I}^{(i)}+h)$ distribution for a metric regressor of interest. For a categorical regressor of interest, $\mu$ does not have to be chosen and we denote this version of the proposed effect size measure simply by $\tilde{\Delta}_s^i$.
    For the fAME as defined by \cite{BischlHeumann}, it then holds that in the case of one metric regressor of interest
    \begin{align*}\label{fAMEa}\addtocounter{equation}{1}
        \text{fAME}_{\mathcal{D},h}=
          h\cdot\sum_{i=1}^{\vert\mathcal{D}_{\textup{[I]}}\vert}\frac{n_{\mathcal{D}_{\textup{[$I_i$]}}}}{n}\cdot\tilde{\Delta}_s^{i,h}(\hat{\theta})\,,\nonumber\tag{\theequation a}
    \end{align*}
    and in the case of one categorical regressor of interest
    \begin{equation*}\label{fAMEb}
        \text{fAME}_\mathcal{D}=\big(n-n_{\textup{ref}}\big)^{-1}\cdot
         \sum_{i=1}^{d_I}n_{\mathcal{D}_{\textup{[$I_i$]}}}\cdot\big(-\tilde{\Delta}_s^{i[i]}(\hat{\theta})\big)\tag{\theequation b}\,,
    \end{equation*} 
    where $n_{\textup{ref}}$ denotes the number of observations of the reference category for the regressor of interest in $\mathcal{D}_X$, i.e. the number of elements with $\realx{I}=\mathbf{0}\in\R^{d_I}$.

\begin{proof}
This proof is split into two parts, with the first showing \cref{fAMEa} and the second \cref{fAMEb}. Throughout, we again denote by $\mathcal{D}_X^{\textup{[MC]}}$ the sequence with elements in $\R^{p-d_I}$ of observation-vectors contained in $\mathcal{D}_X$, but with the elements mapped to $I$ removed; and, $\forall i\in\{1,...,\vert\mathcal{D}_{\textup{[I]}}\vert\}$, by $\mathcal{D}^{(i)}_{\textup{[MC]}}$ the subsequence of $\mathcal{D}_X^{\textup{[MC]}}$ with values in $\R^{p-d_I}$ of those observations in $\mathcal{D}_X^{\textup{[MC]}}$ where the removed entry was equal to the $i$th unique observation of the regressor of interest. Note that, for categorical regressors of interest $\vert\mathcal{D}_{\textup[I]}\vert=d_I+1$, and we consider the reference category, i.e. $\mathbf{0}$, the $(d_I+1)$th unique observation in order to leave out the corresponding observations contained in $\mathcal{D}^{(d_I+1)}_{\textup{[MC]}}$ as \cite{BischlHeumann} suggest. Throughout, we denote the number of elements in  $\mathcal{D}^{(i)}_{\textup{[MC]}}$ by $n_{\mathcal{D}_{\textup{[$I_i$]}}}$.\\

\noindent\textbf{Step 1:} Given this propositions setting, the following holds for the regressor of interest being metric
\begin{align}\label{proof_fAMEa}
     h\cdot\sum_{i=1}^{\vert\mathcal{D}_{\textup{[I]}}\vert}&\frac{n_{\mathcal{D}_{\textup{[$I_i$]}}}}{n}\tilde{\Delta}_s^{i,h}(\hat{\theta})\overset{\substack{prop.\ref{UnifProp},\\cor.\ref{UnifCor}}}{=} h\cdot\sum_{i=1}^{\vert\mathcal{D}_{\textup{[I]}}\vert}\frac{n_{\mathcal{D}_{\textup{[$I_i$]}}}}{n}\cdot\frac{1}{h}\cdot\Big(e\big(\hat{\theta}\big\vert \{\realx{I}^{(i)}+h\},\mathcal{D}_X^{\textup{[MC]}},\delta_{\{\realx{I}^{(i)}+h\}},\aiimeasure(\boldsymbol{\cdot},\realx{I}^{(i)}+h)\big)-\nonumber\\
     &\,\hspace*{7cm}
        e\big(\hat{\theta}\big\vert \{\realx{I}^{(i)}\},\mathcal{D}_X^{\textup{[MC]}},\delta_{\{\realx{I}^{(i)}\}},\aiimeasure(\boldsymbol{\cdot},\realx{I}^{(i)})\big)\Big)\nonumber\\
     =&n^{-1}\cdot\sum_{i=1}^{\vert\mathcal{D}_{\textup{[I]}}\vert}n_{\mathcal{D}_{\textup{[$I_i$]}}}\cdot\Big(
    \frac{1}{n_{\mathcal{D}_{\textup{[$I_i$]}}}}\cdot\sum_{v\in\mathcal{D}^{(i)}_{\textup{[MC]}}}\regF{\hat{\theta}}(\realx{I}=\realx{I}^{(i)}+h,\realx{\text{MC}}=v)
     -\\
     &\,\hspace*{3cm}
         \frac{1}{n_{\mathcal{D}_{\textup{[$I_i$]}}}}\cdot\sum_{v\in\mathcal{D}^{(i)}_{\textup{[MC]}}}\regF{\hat{\theta}}(\realx{I}=\realx{I}^{(i)},\realx{\text{MC}}=v)\Big)\nonumber\\
         &=\frac{1}{n} \cdot \sum_{i=1}^{\vert\mathcal{D}_{\textup{[I]}}\vert}\Big(
    \sum_{v\in\mathcal{D}^{(i)}_{\textup{[MC]}}}\regF{\hat{\theta}}(\realx{I}=\realx{I}^{(i)}+h,\realx{\text{MC}}=v)
     -\regF{\hat{\theta}}(\realx{I}=\realx{I}^{(i)},\realx{\text{MC}}=v)\Big)
         \,,
\end{align}
where the last equation follows from the fact that $\mathcal{D}^{(i)}_{\textup{[MC]}}$ is a finite sequence. In the setting of our work, the last term of \cref{proof_fAMEa} is also equal to the definition of fAME for one metric regressor of interest from \cite{BischlHeumann}, proving \cref{fAMEa}.\\

\noindent\textbf{Step 2:} Given this propositions setting, the following holds in the case of the regressor of interest being categorical $\forall i\in \{1,...,d_I\}$
\begin{equation}\label{proof_fAMEb_1}
\begin{aligned}
    \tilde{\Delta}_s^{i[i]}(\hat{\theta})&=\sum_{v\in\mathcal{D}^{(i)}_{\textup{[MC]}}}\Big(\regF{\hat{\theta}}\big(\realx{I}=\standV{i},\realx{\text{MC}}=v\big)- \regF{\hat{\theta}}\big(\realx{I}=\mathbf{0},\realx{\text{MC}}=v\big)\Big)\\
    \Leftrightarrow -\tilde{\Delta}_s^{i[i]}(\hat{\theta})&=\sum_{v\in\mathcal{D}^{(i)}_{\textup{[MC]}}}\Big(\underset{\substack{\hat{=}\textit{Observation-Wise Categorical Marginal Effect}\\\text{as defined in \cite[equation (4)]{BischlHeumann} }\forall v\in\mathcal{D}^{(i)}_{\textup{[MC]}} }}{\underbrace{\regF{\hat{\theta}}\big(\realx{I}=\mathbf{0},\realx{\text{MC}}=v\big)-\regF{\hat{\theta}}\big(\realx{I}=\standV{i},\realx{\text{MC}}=v\big)}}\Big)\,.
\end{aligned}
\end{equation}
It immediately follows that
\begin{align*}
    \big(n-&n_{\textup{ref}}\big)^{-1}\cdot
         \sum_{i=1}^{d_I}n_{\mathcal{D}_{\textup{[$I_i$]}}}\cdot\big(-\tilde{\Delta}_s^{i[i]}(\hat{\theta})\big)\\
         =&\big(n-n_{\textup{ref}}\big)^{-1}\cdot
         \sum_{i=1}^{d_I}n_{\mathcal{D}_{\textup{[$I_i$]}}}\cdot\frac{1}{n_{\mathcal{D}_{\textup{[$I_i$]}}}}\cdot\sum_{v\in\mathcal{D}^{(i)}_{\textup{[MC]}}}\Big(\regF{\hat{\theta}}\big(\realx{I}=\mathbf{0},\realx{\text{MC}}=v\big)-
      \regF{\hat{\theta}}\big(\realx{I}=\standV{i},\realx{\text{MC}}=v\big)\Big)\\
         =&\frac{1}{n-n_{\textup{ref}}} \cdot \sum_{i=1}^{d_I}
      \Big(\sum_{v\in\mathcal{D}^{(i)}_{\textup{[MC]}}}\regF{\hat{\theta}}\big(\realx{I}=\mathbf{0},\realx{\text{MC}}=v\big)-
      \regF{\hat{\theta}}\big(\realx{I}=\standV{i},\realx{\text{MC}}=v\big)\Big)\,,
\end{align*}
which is equal to the definition of fAME for one categorical regressor of interest as specified in section 5.1. of \cite{BischlHeumann}, proving \cref{fAMEb}.

\end{proof}
\end{rems}

\subsection{Proof of \Cref{distlem} and \Cref{MetIntCor}}
\distlem*
	\begin{proof}
	Consider two probability spaces $(\Omega_1,\mathcal{F}_1,P_1)$ and $(\Omega_2,\mathcal{F}_2,P_2)$, with $X:\Omega_1\longrightarrow\R$ and  $Y:\Omega_2\longrightarrow\R$.  Note that, formally, for any set $A\in\mathcal{B}(\R)$,
		$\mu_X(A):=P_1(\{\omega\in\Omega_1\vert X(\omega)\in A\})$ and $\mu_Y(A):=P_2(\{\omega\in\Omega_2\vert Y(\omega)\in A\})$, yet we write $P(X\in A)$ and $P(Y\in A)$, respectively, to ease notation. \\
		The measure $\mu_Z$ associated with the distribution of the random variable \[
		Z:\Omega_1\times\Omega_2\longrightarrow\R,\quad Z\big((\omega_1,\omega_2)\big)\mapsto X(\omega_1)Y(\omega_2)
		\] evidently fulfills requirements (I1) w.r.t. $\mu_X$ and $\mu_Y$. 
		Next, let $\mu_{X,Y}$ denote the measure associated with the two-dimensional random variable $(X,Y):\Omega_1\times\Omega_2\longrightarrow\R^2$. Given that $X$ and $Y$ are independent, it follows that $\mu_{X,Y}$ is equal to the product measure $\mu_{X}\times\mu_Y$, i.e. $\forall A,B\in\mathcal{B}(\R): \mu_{X,Y}(A\times B)=\mu_X(A)\mu_Y(B)$ and therefore 
\begin{align*}
				\mu_Z(M)=P(XY\in M)&=\int_{\R^2}\ind{xy\in M}d\mu_{X,Y}(x,y)=\int_\R\int_\R\ind{xy\in M}d\mu_X(x)d\mu_Y(y)\\&=\int_\R\int_{M/z}d\mu_X(x)d\mu_Y(z)=\int_\R\int_{M/z}d\mu_Y(y)d\mu_X(z)
			\end{align*}
$\forall M\in\mathcal{B}(\R)$, by which $\mu_Z$ fulfills requirement (I3) w.r.t. $\mu_X$ and $\mu_Y$. \\
Furthermore, for any set $A\in\mathcal{B}(\R)$, the following holds
\begin{align*}
	\mu_Z(A)=&\int_\R\int_\R\ind{xy\in A}d\mu_X(x)d\mu_Y(y)\\=&\int_{\R\cap\supp(\mu_Y)}\int_{\R\cap\supp(\mu_X)}\ind{xy\in A}d\mu_X(x)d\mu_Y(y)+\\&
	\int_{\R\cap\supp(\mu_Y)}\int_{\R\backslash\supp(\mu_X)}\ind{xy\in A}d\mu_X(x)d\mu_Y(y)+\\&\int_{\R\backslash\supp(\mu_Y)}\int_{\R\cap\supp(\mu_X)}\ind{xy\in A}d\mu_X(x)d\mu_Y(y)+\\&\int_{\R\backslash\supp(\mu_Y)}\int_{\R\backslash\supp(\mu_X)}\ind{xy\in A}d\mu_X(x)d\mu_Y(y)\,.
\end{align*}
Since all but the first summand of the last term are evidently equal to zero, we have now\vadjust{\vspace{2pt}}\nolinebreak\linebreak established that $\mu_Z(A)=0\,\,\, \forall A\in\mathcal{B}(\R)$ with $A\not\subseteq\regprod{\supp(\mu_X)}{\supp(\mu_Y)}$. Lastly, for $B_\varepsilon(x)$\vadjust{\vspace{3pt}}\nolinebreak\linebreak denoting the epsilon ball around $x\in\R$, the following holds $\forall z\in\regprodCl{\supp(\mu_X)}{\supp(\mu_Y)}$, $\varepsilon>0:$ 
$$\exists\, a\in\supp(\mu_1),\, b\in\supp(\mu_2),\, \varepsilon_1,\varepsilon_2>0:  \regprod{B_{\varepsilon_1}(a)}{B_{\varepsilon_2}(b)}\subseteq B_\varepsilon(z)$$
and
\begin{align*}
	\mu_Z\big(B_\varepsilon(z)\big)&=\int_{\R\cap\supp(\mu_Y)}\int_{\R\cap\supp(\mu_X)}\ind{xy\in B_\varepsilon(z)}d\mu_X(x)d\mu_Y(y)\\&
	\geq\int_{\supp(\mu_Y)}\ind{y\in B_{\varepsilon_2}(b)}\int_{\supp(\mu_X)}\ind{x\in B_{\varepsilon_1}(a)}d\mu_X(x)d\mu_Y(y) >0\, .
\end{align*}
Therefore, $\supp(\mu_Z)=\regprodCl{\supp(\mu_X)}{\supp(\mu_Y)}$, whereby $\mu_Z$ fulfills requirement (I2) w.r.t. $\mu_X$ and $\mu_Y$.\\
	\end{proof}

\MetIntCor*
\begin{proof}
	\begin{enumerate}
	\item This statement follows immediately from the properties of multiplication and the closure operator. 
\item This statement follows immediately from \cref{distlem} and the fact that, for two independently distributed, real-valued random variables $X$ and $Y$, the measure associated with the distribution of the random variable $Z:=XY$ equals the measure associated with the distribution of the random variable $Z:=YX$.
\end{enumerate}\,\\
\end{proof}

\newpage
\section{Suggestions regarding Standardization\label{StandAlg}}
When comparing models that address the same research question, but using a different data sample, it is quite possible that the target variable and/or the regressor of interest represent the same underlying concept, but are given by different observed variables. The following gives suggestions of how to standardize the proposed quantities in order to still be comparable in such situations.

\paragraph{Cases of different observed target variables representing the same latent variable} In such cases, a current standard procedure is to fit each model as a standardized regression, however, this approach is not permissible in many settings and also leads to a considerable loss in interpretability. Instead, we propose, if only the target variables of the models are different observations of the same latent variable, to simply \textit{report} the proposed quantities in a standardized way. For example, instead of directly reporting the numeric value of $\widehat{\measure}$ for each model, first dividing by the standard deviation of the respective target variable, giving the \emph{generalized marginal effect in terms of the standard deviation of $Y$.}  Naturally, the visualization of $\efun$ and $\bar{\eff}$ may be standardized in this situation by simply relabeling the $y$-axis in terms of standard deviation and centering at the respective means.

\paragraph{Cases where the regressors of interest are metric and different observed variables representing the same latent variable} Correspondingly, and, if applicable, additionally, one may simply adjust the $x$-axis in these cases by limiting each $x$-axis between the same two quantiles for each regressor of interest.\\ If one wants to maximize the comparability of such visualizations, this is achieved by additionally normalizing the interval between the two chosen quantiles. Specifically, this is achieved in the following two steps.\\
     \begin{itemize}
        \item[\textbf{Step 1:}] Choose two quantiles and, for each regressor of interest, identify the intervals  between the quantiles of its observations $\mathcal{D}^{[I]}_X$.
        \item[\textbf{Step 2:}]Scale those intervals to the interval $[0,1]$ using the function\begin{align}
        z(x,Q):\R\times\mathcal{B}(\R)\longrightarrow\R,\quad (x,Q)\mapsto\ind{x\in Q}\cdot\frac{x-\min\{Q\}}{\max\{Q\}-\min\{Q\}}\,.   
        \end{align}
        
    \end{itemize}

\noindent The above function $z$ may further be used for computing standardized versions of both the individualized expectation $e$ and generalized marginal effect $\measure$ by choosing $\RI_I$ and $\RI$, respectively, as the interval $[0,1]$ and replacing $\efun$ and $\bar{\eff}$ with $\efun\circ z^{-1}(\boldsymbol{\cdot},Q)$ and $\bar{\eff}\circ z^{-1}(\boldsymbol{\cdot},Q)$, respectively; where, for each regressor of interest, $Q$ is chosen either again as the interval between the two quantiles or the subsequence of observations that fall between the same two quantiles, if one chooses to utilize the measure associated with the (joint) empirical distribution.\enlargethispage{11pt}
    
 \paragraph{Interpretation of the standardized $e$ and $\measure$}   Having standardized the individualized expectation or generalized marginal effect as suggested above, for some choice of \mbox{$q_1,q_2\in[0,1]$} and $q_1<q_2$, one may then report the following:\begin{itemize}
 \item[\textbf{For $e$:}] \begin{center}"\textit{The average expectation of $Y$ between the $q_1$ quantile and  $q_2$ quantile is $\widehat{e}$ \hfill\\(standard deviations below/above the mean of $Y$)}." \end{center}
\item[\textbf{For $\measure$:}] \begin{center}"\textit{The average slope of expectation of $Y$ between the $q_1$ quantile and  $q_2$ quantile is $\widehat{\measure}$ (standard deviations of $Y$)}." \end{center}
\end{itemize}
\begin{n}
    Since the choice of sets and measures is still completely flexible, it is very important to specify which exact average one is referring to.
\end{n}

\if1\Preprint{\newpage\section{R-Code \label{CodeApp}}
This appendix gives, first, all helper functions used for the analysis of \cref{Datasection} in \cref{helper}; and, subsequently, the R-code for the analysis and plots of sections \ref{RCTsec} and \ref{Silberzahn} in \cref{Rrct} and \cref{Rsilberzahn}, respectively.
\subsection{Helper Functions\label{helper}}
\begin{lstlisting}[numbers=none]
### All helper functions to produce the plots of section 4

emp_deriv<-function(x,y){ #this function determines an empirical derivative of a function f given a set of points (x,f(x))
  xres<-numeric(length(x)-1)
  yres<-numeric(length(y)-1)
  for(i in 1:length(x)-1){
    xres[i]<-x[i]+(x[i+1]-x[i])/2
    yres[i]<-(y[i+1]-y[i])/(x[i+1]-x[i])
  }
  return(data.frame(x=xres,y=yres))
}


inv.logit<-function(x){ #this function computes the inverse logit of a value x 
  res<-numeric(length(x))
  for(i in 1:length(x)){
    res[i]<-1/(1+exp(-x[i]))
  }
  return(res)
}

inv.logit.deriv<-function(x,beta){ #this function computes the derivative of the inverse logit of a value x 
  res<-numeric(length(x))
  for(i in 1:length(x)){
    res[i]<-(beta*exp(x[i]))/(1+exp(x[i]))^2
  }
  return(res)
}


logistic_unif<-function(beta){ #this function computes the proposed effect size measure for logistic regression, with X=[0,1] and mu chosen as associated with the U([0,1]) distribution
  return((exp(beta[1])*(exp(beta[2])-1))/((exp(beta[1])+1)*(exp(beta[1]+beta[2])+1)))
}

poisson_unif<-function(beta){ #this function computes the proposed effect size measure for poisson regression, with X=[0,1] and mu chosen as associated with the U([0,1]) distribution
  return(exp(beta[1])*(exp(beta[2])-1))
}

logistic_efun<-function(beta,x){ #this function computes the expectation for logistic regression
  return((exp(beta[1]+beta[2]*x))/(1+exp(beta[1]+beta[2]*x)))
}

poisson_efun<-function(beta,x){ #this function computes the expectation for poisson regression
  return(exp(beta[1]+beta[2]*x))
}

intval<-function(betas,regs,deriv=NULL){# this function computes the values of the expectation
                                        # and slope plots for the clinical trial data of Tortora (2020)
  betas <- as.numeric(betas)
  eta <-betas[1]+betas[2]*regs$assignmentG+betas[3]*regs$age+betas[4]*regs$genderM+betas[5]*regs$ethB+betas[6]*regs$ethL+betas[7]*regs$ethO
  if(is.null(deriv)){return(unlist(inv.logit(eta)))
  }else{
    return(unlist(inv.logit.deriv(eta,betas[deriv])))
  }
}

catAII1<-function(draws,data){# this function computes \Delta_s under assumption (A.I') for the clinical trial data of Tortora (2020)
  df<-dplyr::select(dummy_cols(data[,c("AGE","Assignment","GENDER","Ethnicity")],"Ethnicity"),-Ethnicity)
  names(df)<-c("age","assignmentG","genderM","ethW","ethB","ethL","ethO")
  df$genderM<-ifelse(df$genderM=="M",1,0)
  df$assignmentG<-ifelse(df$assignmentG=="G",1,0)
  out<-list()
  out$ethB<-apply(draws,1,function(x){sum(intval(x,dplyr::select(mutate(df,ethW=0,ethB=1,ethL=0,ethO=0),-"ethW"))-intval(x,dplyr::select(mutate(df,ethW=1,ethB=0,ethL=0,ethO=0),-"ethW")))/nrow(df)})
  out$ethL<-apply(draws,1,function(x){sum(intval(x,dplyr::select(mutate(df,ethW=0,ethB=0,ethL=1,ethO=0),-"ethW"))-intval(x,dplyr::select(mutate(df,ethW=1,ethB=0,ethL=0,ethO=0),-"ethW")))/nrow(df)})
  out$ethO<-apply(draws,1,function(x){sum(intval(x,dplyr::select(mutate(df,ethW=0,ethB=0,ethL=0,ethO=1),-"ethW"))-intval(x,dplyr::select(mutate(df,ethW=1,ethB=0,ethL=0,ethO=0),-"ethW")))/nrow(df)})
  return(out)
}

catAII2<-function(draws,data){# this function computes \Delta_s under assumption (A.I'') for the clinical trial data of Tortora (2020)
  df<-dplyr::select(dummy_cols(data[,c("AGE","Assignment","GENDER","Ethnicity")],"Ethnicity"),-Ethnicity)
  names(df)<-c("age","assignmentG","genderM","ethW","ethB","ethL","ethO")
  df$genderM<-ifelse(df$genderM=="M",1,0)
  df$assignmentG<-ifelse(df$assignmentG=="G",1,0)
  out<-list()
  out$ethB<-apply(draws,1,function(x){(sum(intval(x,dplyr::select(df[which(df$ethB==1),],-"ethW")))/length(which(df$ethB==1)))-(sum(intval(x,dplyr::select(df[which(df$ethW==1),],-"ethW")))/length(which(df$ethW==1)))})  
  out$ethL<-apply(draws,1,function(x){(sum(intval(x,dplyr::select(df[which(df$ethL==1),],-"ethW")))/length(which(df$ethL==1)))-(sum(intval(x,dplyr::select(df[which(df$ethW==1),],-"ethW")))/length(which(df$ethW==1)))}) 
  out$ethO<-apply(draws,1,function(x){(sum(intval(x,dplyr::select(df[which(df$ethO==1),],-"ethW"))/length(which(df$ethO==1))))-(sum(intval(x,dplyr::select(df[which(df$ethW==1),],-"ethW"))/length(which(df$ethW==1))))}) 
  return(out)
}



################################################################################
#### From https://stackoverflow.com/questions/35511951/r-ggplot2-collapse-or-remove-segment-of-y-axis-from-scatter-plot :

squish_trans <- function(from, to, factor) {
  trans <- function(x) {
    if (any(is.na(x))) return(x)
    # get indices for the relevant regions
    isq <- x > from & x < to
    ito <- x >= to
    # apply transformation
    x[isq] <- from + (x[isq] - from)/factor
    x[ito] <- from + (to - from)/factor + (x[ito] - to)
    return(x)}
  inv <- function(x) {
    if (any(is.na(x))) return(x)
    # get indices for the relevant regions
    isq <- x > from & x < from + (to - from)/factor
    ito <- x >= from + (to - from)/factor
    # apply transformation
    x[isq] <- from + (x[isq] - from) * factor
    x[ito] <- to + (x[ito] - (from + (to - from)/factor))
    return(x)}
  # return the transformation
  return(trans_new("squished", trans, inv))
}








\end{lstlisting}
\subsection{Analysis from \cref{RCTsec}\label{Rrct}}
\begin{lstlisting}[numbers=none]
library(plyr)
library(dplyr)
library(tidyr)
library(tidybayes)
library(ggplot2)
library(ggpubr)
library(ggridges)
library(ggeasy)
library(readxl)
library(brms)
library(fastDummies)
library(HDInterval)
library(ggdist)
library(latex2exp)

utils::download.file(url="https://data.mendeley.com/public-files/datasets/25yjwbphn4/files/05601aea-ad43-4a93-8d28-adf0fa74b3c9/file_downloaded",
                           destfile = "CYPtrialData.xlsx")
data<-read_xlsx("CYPtrialData.xlsx")

data<-data[complete.cases(data),]
data<-data[which(data$LOS>24*3),]

data$Assignment<-factor(data$Assignment,levels=c("S","G"))
data$`RACE/ETHNICITY`<-factor(data$`RACE/ETHNICITY`,levels=c("W","B","L","O/U"))

names(data)[which(names(data)=="# Psychotropic Medications")]<-"Psychotropic.Medications"
names(data)[which(names(data)=="# Administrations")]<-"Administrations"
names(data)[which(names(data)=="RACE/ETHNICITY")]<-"Ethnicity"


cat(paste0("White:  ",round(length(which(data$Ethnicity=="W"))/nrow(data)*100,2)," %\n",
             "Latinx:  ",round(length(which(data$Ethnicity=="L"))/nrow(data)*100,2)," %\n",
             "Black:  ",round(length(which(data$Ethnicity=="B"))/nrow(data)*100,2)," %\n",
             "Other:  ",round(length(which(data$Ethnicity=="O/U"))/nrow(data)*100,2)," %"))


blogisticLOS<-brm(as.numeric(LOS>median(LOS))~Assignment+AGE+GENDER+Ethnicity,family=bernoulli(),data=data,seed=2022)

set.seed(2022)
logLOS_draws<-blogisticLOS %>%
  spread_draws(b_Intercept,
               b_AssignmentG,
               b_AGE,
               b_GENDERM,
               b_EthnicityB,
               b_EthnicityL,
               b_EthnicityODU,
               ndraws = 1000) %>%
  dplyr::select(- .chain,-.iteration,-.draw) %>%
  as.data.frame()



empEth1<-catAII1(logLOS_draws,data)
empEth2<-catAII2(logLOS_draws,data)


df<-dplyr::select(dummy_cols(data[,c("AGE","Assignment","GENDER","Ethnicity")],"Ethnicity"),-Ethnicity)
names(df)<-c("age","assignmentG","genderM","ethW","ethB","ethL","ethO")
df$genderM<-ifelse(df$genderM=="M",1,0)
df$assignmentG<-ifelse(df$assignmentG=="G",1,0)

ethCat1<-list()
ethCat1$ethW<-apply(logLOS_draws,1,function(x){sum(intval(x,dplyr::select(mutate(df,ethW=1,ethB=0,ethL=0,ethO=0),-"ethW")))/nrow(df)})
ethCat1$ethB<-apply(logLOS_draws,1,function(x){sum(intval(x,dplyr::select(mutate(df,ethW=0,ethB=1,ethL=0,ethO=0),-"ethW")))/nrow(df)})
ethCat1$ethL<-apply(logLOS_draws,1,function(x){sum(intval(x,dplyr::select(mutate(df,ethW=0,ethB=0,ethL=1,ethO=0),-"ethW")))/nrow(df)})
ethCat1$ethO<-apply(logLOS_draws,1,function(x){sum(intval(x,dplyr::select(mutate(df,ethW=0,ethB=0,ethL=0,ethO=1),-"ethW")))/nrow(df)})

ethCat2<-list()
ethCat2$ethW<-apply(logLOS_draws,1,function(x){sum(intval(x,dplyr::select(df[which(df$ethW==1),],-"ethW")))/length(which(df$ethW==1))})
ethCat2$ethB<-apply(logLOS_draws,1,function(x){sum(intval(x,dplyr::select(df[which(df$ethB==1),],-"ethW")))/length(which(df$ethB==1))})
ethCat2$ethL<-apply(logLOS_draws,1,function(x){sum(intval(x,dplyr::select(df[which(df$ethL==1),],-"ethW")))/length(which(df$ethL==1))})
ethCat2$ethO<-apply(logLOS_draws,1,function(x){sum(intval(x,dplyr::select(df[which(df$ethO==1),],-"ethW")))/length(which(df$ethO==1))})

Pcat1<-ldply(ethCat1, data.frame) %>% mutate(.id=as.factor(.id))
names(Pcat1)<-c("Ethnicity","value")
Pcat1<-merge(Pcat1,
             Pcat1 %>%
               group_by(Ethnicity)%>%
               dplyr::summarize_at("value",mean) %>% dplyr::rename(mean=value),
             by="Ethnicity")
levels(Pcat1$Ethnicity)<- list(White  = "ethW", Latinx  = "ethL", Black  = "ethB", Other  = "ethO")
Pcat2<-ldply(ethCat2, data.frame) %>% mutate(.id=as.factor(.id))
names(Pcat2)<-c("Ethnicity","value")
Pcat2<-merge(Pcat2,
             Pcat2 %>%
               group_by(Ethnicity)%>%
               dplyr::summarize_at("value",mean) %>% dplyr::rename(mean=value),
             by="Ethnicity")
levels(Pcat2$Ethnicity)<- list(White  = "ethW", Latinx  = "ethL", Black  = "ethB", Other  = "ethO")

cat_p1<-ggplot(Pcat1, aes(x = value, y = Ethnicity)) +
  stat_halfeye(alpha=0.75,point_interval = "mean_hdi")+
  scale_x_continuous(labels = scales::percent)+
  xlab(TeX("$g^{\\left[ I\\right]}_{avg}\\,(\\theta,\\cdot)$"))+
  coord_flip()+theme_bw()+ggtitle("Under assumption (A.II')")+xlim(0.25,0.65)

cat_p2<-ggplot(Pcat2, aes(x = value, y = Ethnicity)) +
  stat_halfeye(alpha=0.75,point_interval = "mean_hdi")+
  scale_x_continuous(labels = scales::percent)+
  xlab(TeX("$g^{\\left[ I\\right]}_{avg}\\,(\\theta,\\cdot)$"))+
  coord_flip()+theme_bw()+ggtitle("Under assumption (A.II'')")+xlim(0.25,0.65)

p1<-ggarrange(cat_p1,cat_p2, nrow=1,common.legend = TRUE,legend="bottom")


Effcat1<-ldply(empEth1, data.frame) %>% mutate(.id=as.factor(.id))
names(Effcat1)<-c("Ethnicity","value")
Effcat1<-merge(Effcat1,
               Effcat1 %>%
                 group_by(Ethnicity)%>%
                 dplyr::summarize_at("value",mean) %>% dplyr::rename(mean=value),
               by="Ethnicity")

Effcat2<-ldply(empEth2, data.frame) %>% mutate(.id=as.factor(.id))
names(Effcat2)<-c("Ethnicity","value")
Effcat2<-merge(Effcat2,
               Effcat2 %>%
                 group_by(Ethnicity)%>%
                 dplyr::summarize_at("value",mean) %>% dplyr::rename(mean=value),
               by="Ethnicity")
Effcat<-rbind(cbind(data.frame(assumption="A.II'"),Effcat1),
              cbind(data.frame(assumption="A.II''"),Effcat2))
levels(Effcat$Ethnicity)<- list(White  = "ethW", Latinx  = "ethL", Black  = "ethB", Other  = "ethO")


p2<-ggplot(Effcat,aes(x=value,fill=assumption))+
  geom_density(alpha=0.3)+theme_minimal()+
  stat_pointinterval(aes(color=assumption,shape=assumption),position = position_dodge(width = 3, preserve = "single"),point_interval = "mean_hdi",point_size=4)+
  easy_remove_y_axis()+facet_grid(Ethnicity~.)+
  scale_shape_manual(values=c(15,19))+
  scale_fill_manual(values=c("green4","navyblue"))+scale_color_manual(values=c("green4","navyblue"))+
  theme(legend.position = "bottom")+xlab(TeX("$\\Delta_s (\\theta)$"))

################################################################################
# figure 1
ggarrange(p1,p2,nrow = 2,heights=c(1.5,2))
################################################################################









\end{lstlisting}
\subsection{Analysis from \cref{Silberzahn}\label{Rsilberzahn}}
\begin{lstlisting}[numbers=none]
library(readr)
library(ggplot2)
library(ggpubr)
library(latex2exp)
library(brms)
library(tidybayes)
library(HDInterval)
library(gridExtra)
library(grid)
library(mvtnorm)
library(ggeasy)
library(ggdist)
library(dplyr)
library(tidyr)
library(ggrepel)
library(MASS)
library(ggridges)
library(rgl)
library(viridis) 
library(zoo)
library(scales)


utils::download.file(url="https://osf.io/download/fv8c3/",destfile = "SilberzahnData.csv")
SilberzahnData <- read_csv("SilberzahnData.csv")

SilberzahnData$rating_mean<-rowMeans(cbind(SilberzahnData$rater1,SilberzahnData$rater2),na.rm=TRUE)
SilberzahnData$redCard_rate<-SilberzahnData$redCards/SilberzahnData$games

lm<-brm(redCard_rate~rating_mean,data=SilberzahnData,seed=54321)
glm_bin<-brm(redCards|trials(games)~rating_mean,data=SilberzahnData,family=binomial(),seed=54321)
glm_poi<-brm(redCards~rating_mean+offset(log(games)),family = poisson(),data=SilberzahnData,seed=54321)



set.seed(600)
LinearBetas<-lm %>%
  spread_draws(b_Intercept,
               b_rating_mean,
               ndraws = 1000) %>%
  dplyr::select(- .chain,-.iteration,-.draw) %>%
  as.data.frame()
names(LinearBetas)<-substring(names(LinearBetas), 3)

set.seed(600)
PoissonBetas<- glm_poi %>%
  spread_draws(b_Intercept,
               b_rating_mean,
               ndraws = 1000) %>%
  dplyr::select(- .chain,-.iteration,-.draw) %>%
  as.data.frame()
names(PoissonBetas)<-substring(names(PoissonBetas), 3)


set.seed(600)
LogisticBetas<- glm_bin %>%
  spread_draws(b_Intercept,
               b_rating_mean,
               ndraws = 1000) %>%
  dplyr::select(- .chain,-.iteration,-.draw) %>%
  as.data.frame()
names(LogisticBetas)<-substring(names(LogisticBetas), 3)


rating<-seq(0,1,by=0.01)

LinPred<-apply(LinearBetas,1,function(x)(x[1]+x[2]*rating))
BinPred<-apply(LogisticBetas,1,function(x)(inv.logit(x[1]+x[2]*rating)))
PoiPred<-apply(PoissonBetas,1,function(x)(exp(x[1]+x[2]*rating)))

Pred<-rbind(data.frame(Distribution="Normal (linear)",
                       x=rating,
                       mean=apply(LinPred,1,mean),
                       ymin=predict(loess(apply(LinPred,1,function(x)HDInterval::hdi(x,credMass = 0.95)[1])~rating),data.frame(rating=rating)),
                       ymax=predict(loess(apply(LinPred,1,function(x)HDInterval::hdi(x,credMass = 0.95)[2])~rating),data.frame(rating=rating))
),
data.frame(Distribution="Poisson (exponential)",
           x=rating,
           mean=apply(PoiPred,1,mean),
           ymin=predict(loess(apply(PoiPred,1,function(x)HDInterval::hdi(x,credMass = 0.95)[1])~rating),data.frame(rating=rating)),
           ymax=predict(loess(apply(PoiPred,1,function(x)HDInterval::hdi(x,credMass = 0.95)[2])~rating),data.frame(rating=rating))
),
data.frame(Distribution="Binomial (logistic)",
           x=rating,
           mean=apply(BinPred,1,mean),
           ymin=predict(loess(apply(BinPred,1,function(x)HDInterval::hdi(x,credMass = 0.95)[1])~rating),data.frame(rating=rating)),
           ymax=predict(loess(apply(BinPred,1,function(x)HDInterval::hdi(x,credMass = 0.95)[2])~rating),data.frame(rating=rating))
))


pred_plot<-ggplot(Pred)+geom_line(aes(x=x,y=mean,color=Distribution,linetype=Distribution),size=1)+
  geom_ribbon(aes(x=x,ymin=ymin,ymax=ymax,fill=Distribution),alpha=0.05)+
  scale_colour_manual(values=c("navyblue","red3","green3"),name="") + 
  scale_linetype_manual(values=c("solid","dotted","dashed"),name="")+
  scale_fill_manual(values=c("navyblue","red3","green3"),name="")+
  labs(y=TeX("$g^{\\left[ I\\right]}_{avg}\\,(\\hat{\\theta},\\cdot)$"),x="Average skin tone rating",color="Distribution Assumption")+
  theme_bw()



LinSlope<-apply(LinearBetas,1,function(x)(rep(x[2],length(rating))))
BinSlope<-apply(LogisticBetas,1,function(x)(inv.logit.deriv(x[1]+x[2]*rating,x[2])))
PoiSlope<-apply(PoissonBetas,1,function(x)(x[2]*exp(x[1]+x[2]*rating)))

Slope<-rbind(data.frame(Distribution="Normal (linear)",
                        x=rating,
                        mean=apply(LinSlope,1,mean),
                        ymin=predict(loess(apply(LinSlope,1,function(x)HDInterval::hdi(x,credMass = 0.95)[1])~rating),data.frame(rating=rating)),
                        ymax=predict(loess(apply(LinSlope,1,function(x)HDInterval::hdi(x,credMass = 0.95)[2])~rating),data.frame(rating=rating))
),
data.frame(Distribution="Poisson (exponential)",
           x=rating,
           mean=apply(PoiSlope,1,mean),
           ymin=predict(loess(apply(PoiSlope,1,function(x)HDInterval::hdi(x,credMass = 0.95)[1])~rating),data.frame(rating=rating)),
           ymax=predict(loess(apply(PoiSlope,1,function(x)HDInterval::hdi(x,credMass = 0.95)[2])~rating),data.frame(rating=rating))
),
data.frame(Distribution="Binomial (logistic)",
           x=rating,
           mean=apply(BinSlope,1,mean),
           ymin=predict(loess(apply(BinSlope,1,function(x)HDInterval::hdi(x,credMass = 0.95)[1])~rating),data.frame(rating=rating)),
           ymax=predict(loess(apply(BinSlope,1,function(x)HDInterval::hdi(x,credMass = 0.95)[2])~rating),data.frame(rating=rating))
))


slope_plot<-ggplot(Slope)+geom_line(aes(x=x,y=mean,color=Distribution,linetype=Distribution),size=1)+
  #geom_ribbon(aes(x=x,ymin=ymin,ymax=ymax,fill=Distribution),alpha=0.05)+
  scale_colour_manual(values=c("navyblue","red3","green3"),name="") + 
  scale_linetype_manual(values=c("solid","dotted","dashed"),name="")+
  scale_fill_manual(values=c("navyblue","red3","green3"),name="")+
  labs(y=TeX("$\\bar{s}\\,(\\hat{\\theta},\\cdot)$"),x="Average skin tone rating",color="Distribution Assumption")+
  theme_bw()

################################################################################
# figure 2
ggarrange(pred_plot,slope_plot, ncol=2, nrow=1, common.legend = TRUE, legend="bottom")
################################################################################


set.seed(55)
LinearBetas<-lm %>%
  spread_draws(b_Intercept,
               b_rating_mean,
               ndraws = 4000) %>%
  dplyr::select(- .chain,-.iteration,-.draw) %>%
  as.data.frame()
names(LinearBetas)<-substring(names(LinearBetas), 3)

set.seed(55)
PoissonBetas<- glm_poi %>%
  spread_draws(b_Intercept,
               b_rating_mean,
               ndraws = 4000) %>%
  dplyr::select(- .chain,-.iteration,-.draw) %>%
  as.data.frame()
names(PoissonBetas)<-substring(names(PoissonBetas), 3)


set.seed(55)
LogisticBetas<- glm_bin %>%
  spread_draws(b_Intercept,
               b_rating_mean,
               ndraws = 4000) %>%
  dplyr::select(- .chain,-.iteration,-.draw) %>%
  as.data.frame()
names(LogisticBetas)<-substring(names(LogisticBetas), 3)





LogisticUnif<-apply(LogisticBetas,1,logistic_unif)
PoissonUnif<-apply(PoissonBetas,1,poisson_unif)



up_linear<-ggplot(LinearBetas,aes(x=rating_mean))+
  geom_density(alpha=0.25,fill=4)+
  xlim(-0.001,0.004)+theme_minimal()+easy_remove_axes()+
  theme(panel.grid.major = element_blank(),
        panel.grid.minor = element_blank(),
        panel.border = element_blank(),
        panel.background = element_blank())+
  stat_pointinterval(point_interval = "mean_hdi")


up_poisson<-ggplot(data.frame(rating_mean=PoissonUnif),aes(x=rating_mean))+
  geom_density(alpha=0.25,fill=4)+
  xlim(-0.001,0.004)+theme_minimal()+easy_remove_axes()+
  theme(panel.grid.major = element_blank(),
        panel.grid.minor = element_blank(),
        panel.border = element_blank(),
        panel.background = element_blank())+
  stat_pointinterval(point_interval = "mean_hdi")

up_logistic<-ggplot(data.frame(rating_mean=LogisticUnif),aes(x=rating_mean))+
  geom_density(alpha=0.25,fill=4)+
  xlim(-0.001,0.004)+theme_minimal()+easy_remove_axes()+
  theme(panel.grid.major = element_blank(),
        panel.grid.minor = element_blank(),
        panel.border = element_blank(),
        panel.background = element_blank())+
  stat_pointinterval(point_interval = "mean_hdi")


axis1<-ggplot()+xlim(-0.001,0.004)+theme_bw()+easy_remove_y_axis()+
  xlab(TeX("$\\Delta_s(\\theta)$"))+
  theme(axis.line = element_line(colour = "black"),
        panel.grid.major = element_blank(),
        panel.grid.minor = element_blank(),
        panel.border = element_blank(),
        panel.background = element_blank())+
  easy_x_axis_labels_size(20)+easy_x_axis_title_size(22)


eb_poisson<-ggplot(data.frame(x=0,y=exp(summary(glm_poi)$fixed[2,1]),ymax = exp(summary(glm_poi)$fixed[2,4]),ymin = exp(summary(glm_poi)$fixed[2,3])),aes(y=y,x=x))+
  geom_point(size=4)+geom_hline(aes(yintercept=1),color="darkgray",linetype="dashed")+
  geom_errorbar(aes(ymax = ymax, ymin = ymin),size=1)+
  ylim(0.9,1.7)+xlim(-4,4)+theme_minimal()+easy_remove_axes()+
  theme(panel.grid.major = element_blank(),
        panel.grid.minor = element_blank(),
        panel.border = element_blank(),
        panel.background = element_blank())+
  coord_flip()+
  annotate("text",y=1,x=-3.5,label="OR=1",hjust=0.5,size=6)

eb_logistic<-ggplot(data.frame(x=0,y=exp(summary(glm_bin)$fixed[2,1]),ymax = exp(summary(glm_bin)$fixed[2,4]),ymin = exp(summary(glm_bin)$fixed[2,3])),aes(y=y,x=x))+
  geom_point(size=4)+geom_hline(aes(yintercept=1),color="darkgray",linetype="dashed")+
  geom_errorbar(aes(ymax = ymax, ymin = ymin),size=1)+
  ylim(0.9,1.7)+xlim(-4,4)+theme_minimal()+easy_remove_axes()+
  theme(panel.grid.major = element_blank(),
        panel.grid.minor = element_blank(),
        panel.border = element_blank(),
        panel.background = element_blank())+
  coord_flip()

eb_linear<-ggplot()+easy_remove_axes()+
  ylim(0,2)+xlim(0.9,1.7)+
  geom_vline(aes(xintercept=1),color="darkgray",linetype="dashed")+
  theme(panel.grid.major = element_blank(),
        panel.grid.minor = element_blank(),
        panel.border = element_blank(),
        panel.background = element_blank())+
  annotate("text",x=1.3,y=1,label="not\n applicable",size=10,color="darkgray")

axis2<-ggplot()+xlim(0.9,1.7)+theme_bw()+easy_remove_y_axis()+
  xlab("Odds Ratio")+
  theme(axis.line = element_line(colour = "black"),
        panel.grid.major = element_blank(),
        panel.grid.minor = element_blank(),
        panel.border = element_blank(),
        panel.background = element_blank(),
        plot.margin=grid::unit(c(-7,5.5,5.5,5.5), "pt"))+
  easy_x_axis_labels_size(20)+easy_x_axis_title_size(22)


beta_linear<-ggplot(data.frame(x=NA,y=NA),aes(x=x,y=y))+ylim(0,0.75)+theme_bw()+easy_remove_axes()+
  annotate("pointrange", x =  0, y =summary(lm)$fixed[2,1], ymin = summary(lm)$fixed[2,3], ymax = summary(lm)$fixed[2,4],colour = 4, size = 1.5, alpha=0.4)+
  annotate("text", x =  0, y =summary(lm)$fixed[2,1],label=expression(paste(hat(beta),"=")),vjust=-0.45,hjust=-0.001,size=6)+
  annotate("text", x =  0, y =summary(lm)$fixed[2,1],label=round(summary(lm)$fixed[2,1],5),vjust=-1.2,hjust=-0.35,size=6)+
  theme(panel.grid.major = element_blank(),
        panel.grid.minor = element_blank(),
        panel.border = element_blank(),
        panel.background = element_blank())+
  coord_flip()

beta_poisson<-ggplot(data.frame(x=NA,y=NA),aes(x=x,y=y))+ylim(0,0.75)+theme_bw()+easy_remove_axes()+
  annotate("pointrange", x =  0, y =summary(glm_poi)$fixed[2,1], ymin = summary(glm_poi)$fixed[2,3], ymax = summary(glm_poi)$fixed[2,4],colour = 4, size = 1.5, alpha=0.4)+
  annotate("text", x =  0, y =summary(glm_poi)$fixed[2,1],label=expression(paste(hat(beta),"=")),vjust=-0.45,hjust=-0.001,size=6)+
  annotate("text", x =  0, y =summary(glm_poi)$fixed[2,1],label=round(summary(glm_poi)$fixed[2,1],5),vjust=-1.2,hjust=-0.35,size=6)+
  theme(panel.grid.major = element_blank(),
        panel.grid.minor = element_blank(),
        panel.border = element_blank(),
        panel.background = element_blank())+
  coord_flip()

beta_logistic<-ggplot(data.frame(x=NA,y=NA),aes(x=x,y=y))+ylim(0,0.75)+theme_bw()+easy_remove_axes()+
  annotate("pointrange", x =  0, y =summary(glm_bin)$fixed[2,1], ymin = summary(glm_bin)$fixed[2,3], ymax = summary(glm_bin)$fixed[2,4],colour = 4, size = 1.5, alpha=0.4)+
  annotate("text", x =  0, y =summary(glm_bin)$fixed[2,1],label=expression(paste(hat(beta),"=")),vjust=-0.45,hjust=-0.001,size=6)+
  annotate("text", x =  0, y =summary(glm_bin)$fixed[2,1],label=round(summary(glm_bin)$fixed[2,1],5),vjust=-1.2,hjust=-0.35,size=6)+
  theme(panel.grid.major = element_blank(),
        panel.grid.minor = element_blank(),
        panel.border = element_blank(),
        panel.background = element_blank())+
  coord_flip()


axis3<-ggplot()+xlim(0,0.75)+theme_bw()+easy_remove_y_axis()+
  xlab(TeX("$\\beta-$coefficient"))+
  theme(axis.line = element_line(colour = "black"),
        panel.grid.major = element_blank(),
        panel.grid.minor = element_blank(),
        panel.border = element_blank(),
        panel.background = element_blank(),
        plot.margin=grid::unit(c(2,5.5,5.5,5.5), "pt"))+
  easy_x_axis_labels_size(20)+easy_x_axis_title_size(22)


linM<-ggplot()+easy_remove_axes()+
  theme(panel.grid.major = element_blank(),
        panel.grid.minor = element_blank(),
        panel.border = element_blank(),
        panel.background = element_blank())+
  annotate("text",x=1,y=1,label="Linear\n model",size=8,color="black")
logisticM<-ggplot()+easy_remove_axes()+
  theme(panel.grid.major = element_blank(),
        panel.grid.minor = element_blank(),
        panel.border = element_blank(),
        panel.background = element_blank())+
  annotate("text",x=1,y=1,label="Logistic\n model",size=8,color="black")
poiM<-ggplot()+easy_remove_axes()+
  theme(panel.grid.major = element_blank(),
        panel.grid.minor = element_blank(),
        panel.border = element_blank(),
        panel.background = element_blank())+
  annotate("text",x=1,y=1,label="Poisson\n model",size=8,color="black")


################################################################################
# figure 3
grid.arrange(linM,up_linear,eb_linear,beta_linear,
             logisticM,up_logistic,eb_logistic,beta_logistic,
             poiM,up_poisson,eb_poisson,beta_poisson,ggplot()+
               theme(panel.grid.major = element_blank(),
                     panel.grid.minor = element_blank(),
                     panel.border = element_blank(),
                     panel.background = element_blank()),axis1,axis2,axis3,nrow=4,ncol=4,heights=c(2,2,2,1),widths=c(1,6,3,3))
################################################################################


set.seed(600)
LinearBetas_sigma<-lm %>%
  spread_draws(b_Intercept,
               b_rating_mean,
               sigma,
               ndraws = 1000) %>%
  dplyr::select(- .chain,-.iteration,-.draw) %>%
  as.data.frame()
names(LinearBetas_sigma)[1:2]<-substring(names(LinearBetas_sigma)[1:2], 3)

set.seed(600)
PoissonBetas<- glm_poi %>%
  spread_draws(b_Intercept,
               b_rating_mean,
               ndraws = 1000) %>%
  dplyr::select(- .chain,-.iteration,-.draw) %>%
  as.data.frame()
names(PoissonBetas)<-substring(names(PoissonBetas), 3)


set.seed(600)
LogisticBetas<- glm_bin %>%
  spread_draws(b_Intercept,
               b_rating_mean,
               ndraws = 1000) %>%
  dplyr::select(- .chain,-.iteration,-.draw) %>%
  as.data.frame()
names(LogisticBetas)<-substring(names(LogisticBetas), 3)



y<-seq(-0.2,0.2,by=0.001)

linPstart<-data.frame(group=c(rep("Individualized predictive distribution\n for skin color rating of 0",length(y)),
                              rep("Individualized predictive distribution\n for skin color rating of 1",length(y))),
                      x=rep(y,2))

linP<-data.frame()
set.seed(928)
for(i in 1:nrow(LinearBetas_sigma)){
  linP<-rbind(linP,cbind(linPstart,
                         data.frame(y=c(dnorm(y,mean=LinearBetas_sigma$Intercept[i],sd=LinearBetas_sigma$sigma[i]),
                                        dnorm(y,mean=LinearBetas_sigma$Intercept[i]+LinearBetas_sigma$rating_mean[i],sd=LinearBetas_sigma$sigma[i]))),
                         lines=paste0("beta",i)))
  
}

linPlines<-linP %>%
  dplyr::group_by(group,x) %>%
  dplyr::summarise(ymin=mean_hdi(y)$ymin,ymax=mean_hdi(y)$ymax,mean=mean(y))



fy0<-linPlines$mean[which(linPlines$group=="Individualized predictive distribution\n for skin color rating of 0")]
fy1<-linPlines$mean[which(linPlines$group=="Individualized predictive distribution\n for skin color rating of 1")]

linpp<-ggplot(linPlines,aes(x=x,color=group)) +
  geom_ribbon(aes(x=x,ymin=ymin,ymax=ymax,fill=group),alpha=0.4, colour = NA)+
  geom_line(aes(x=x,y=mean,linetype=group))+
  scale_color_manual(name='', values=c("green4","blue4"))+
  scale_fill_manual(name='', values=c("green4","blue4"))+
  theme_bw()+ theme(legend.position="bottom")+xlim(-0.2,0.2)+
  xlab("y")+ylab(TeX("$\\Pi(y,\\theta)$"))+
  geom_vline(xintercept=y[which.max(fy0)],linetype="dashed",color="grey")+
  geom_vline(xintercept=y[which.max(fy1)],linetype="dashed",color="grey")+
  annotate(x=0.0038,y=3,label="Distance between maximums:\n ca. 0.0016",vjust=1,geom="label",size=3)+ 
  labs(color  = "", linetype = "",fill="")




S<-SilberzahnData[,c("redCards","games","redCard_rate","rating_mean")]%>%drop_na()

log0<-sum(apply(LogisticBetas,1,function(x){logistic_efun(beta=x,x=0)}))/nrow(LogisticBetas)
log1<-sum(apply(LogisticBetas,1,function(x){logistic_efun(beta=x,x=1)}))/nrow(LogisticBetas)

poi0<-sum(apply(PoissonBetas,1,function(x){poisson_efun(beta=x,x=0)}))/nrow(PoissonBetas)
poi1<-sum(apply(PoissonBetas,1,function(x){poisson_efun(beta=x,x=1)}))/nrow(PoissonBetas)


nonlinP1<-data.frame(model=c(rep("logistic",6),rep("poisson",6)),sc=factor(rep(c("skin color\n rating of 0","If OR were equal to 1","skin color\n rating of 1"),4),levels=c("skin color\n rating of 0","If OR were equal to 1","skin color\n rating of 1")),
                     x=factor(rep(c(rep("P(Y>0|X)",3),rep("P(Y=0|X)",3)),2),levels=c("P(Y>0|X)",rep("P(Y=0|X)"))),
                     y=c(log0,rollmean(c(log0,log1),2),log1,
                         1-log0,rollmean(c(1-log0,1-log1),2),1-log1,
                         1-dpois(0,poi0),rollmean(c(1-dpois(0,poi0),1-dpois(0,poi1)),2),1-dpois(0,poi1),
                         dpois(0,poi0),rollmean(c(dpois(0,poi0),dpois(0,poi1)),2),dpois(0,poi1)),
                     alpha=rep(c(-1,-1,-1),4))

nonlinP2<-rbind(data.frame(model="logistic",sc="skin color\n rating of 0",x="P(Y>0|X)",y=0,alpha=-1,draws=apply(LogisticBetas,1,function(x){logistic_efun(beta=x,x=0)})),
                data.frame(model="logistic",sc="skin color\n rating of 1",x="P(Y>0|X)",y=0,alpha=-1,draws=apply(LogisticBetas,1,function(x){logistic_efun(beta=x,x=1)})),
                data.frame(model="poisson",sc="skin color\n rating of 0",x="P(Y>0|X)",y=0,alpha=-1,draws=apply(PoissonBetas,1,function(x){poisson_efun(beta=x,x=0)})),
                data.frame(model="poisson",sc="skin color\n rating of 1",x="P(Y>0|X)",y=0,alpha=-1,draws=apply(PoissonBetas,1,function(x){poisson_efun(beta=x,x=1)})),
                data.frame(model="logistic",sc="skin color\n rating of 0",x="P(Y=0|X)",y=0,alpha=-1,draws=1-apply(LogisticBetas,1,function(x){logistic_efun(beta=x,x=0)})),
                data.frame(model="logistic",sc="skin color\n rating of 1",x="P(Y=0|X)",y=0,alpha=-1,draws=1-apply(LogisticBetas,1,function(x){logistic_efun(beta=x,x=1)})),
                data.frame(model="poisson",sc="skin color\n rating of 0",x="P(Y=0|X)",y=0,alpha=-1,draws=1-apply(PoissonBetas,1,function(x){poisson_efun(beta=x,x=0)})),
                data.frame(model="poisson",sc="skin color\n rating of 1",x="P(Y=0|X)",y=0,alpha=-1,draws=1-apply(PoissonBetas,1,function(x){poisson_efun(beta=x,x=1)}))
)

nonlinP<-merge(nonlinP1 %>%
                 mutate(mscx=paste0(model,sc,x)),
               nonlinP2 %>%
                 mutate(mscx=paste0(model,sc,x))%>%
                 group_by(mscx)%>%
                 dplyr::summarise(ymax=mean_hdi(draws)$ymax,ymin=mean_hdi(draws)$ymin),
               by="mscx", all=TRUE)
nonlinP[which(nonlinP$sc=="If OR were equal to 1"),"ymin"]<-nonlinP[which(nonlinP$sc=="If OR were equal to 1"),"y"]
nonlinP[which(nonlinP$sc=="If OR were equal to 1"),"ymax"]<-nonlinP[which(nonlinP$sc=="If OR were equal to 1"),"y"]



catpp<-ggplot(data=nonlinP,aes(y=y,x=factor(model),color=sc))+
  geom_bar(aes(linetype=sc,alpha = alpha),fill="white",stat="identity",position = position_dodge(width = 0.5))+
  theme_bw()+ facet_grid(.~ factor(x))+ 
  scale_y_continuous(trans = squish_trans(0.01, 0.99, 15),
                     breaks = c(0,seq(0.05,0.95, by = 0.1),1)) + scale_alpha(guide = 'none')+
  labs(color="")+xlab("model")+ylab("")+ 
  geom_errorbar(aes(ymin=ymin,ymax=ymax),position = position_dodge(width = 0.5),width=.9,alpha=0.5)+
  scale_colour_manual(values=c("#316395","grey60","#B82E2E"),name="") + 
  scale_linetype_manual(values=c("solid","dotted","dashed"),name="")+
  theme(legend.position= "bottom",legend.justification = "left",legend.box.margin = margin(0,0,0,-0.75,"cm"))


################################################################################
# figure 4
ggarrange(linpp,catpp,nrow=1,widths = c(1.5,1))
################################################################################



\end{lstlisting}
}\fi

\end{changemargin}
\end{document}